\documentclass[a4, leqno, 12pt]{article}

\usepackage{latexsym}
\usepackage{amsmath}
\usepackage{amssymb}
\usepackage{amsfonts}
\usepackage{graphicx,epstopdf}
\usepackage[caption=false]{subfig}

\newcommand{\prob}{\mathbb P}     
\newcommand{\expect}{\mathbb E}   
\newcommand{\filts}{\mathbb F}    
\newcommand{\filt}{\mathcal F}    
\newcommand{\real}{\mathbb R}     
\newcommand{\ig}{\mathcal L}      
\newcommand{\diff}{\mathrm d}     
\newcommand{\p}{\partial}         

\newcommand{\1}{\mbox{1}\hspace{-0.25em}\mbox{l}}           
\newcommand{\esssup}[1]{\underset{#1}{\mathrm{esssup}}\;}   

\newcommand{\SSd}[1]{\mathbb S^2[#1,T]}                 
\newcommand{\HH}[2]{\mathbb H^2_{#2}[#1,T]}             
\newcommand{\KK}[1]{\mathbb K^2[#1,T]}                  
\newcommand{\control}[2]{\mathbb A_{#2}[#1,T]}          
\newcommand{\controlinf}[2]{\mathbb A_{#2}[#1,\infty)}  

\usepackage{geometry}
\usepackage{ntheorem}
\usepackage{hyperref}
\usepackage{setspace} 
\usepackage{cleveref} 

\newtheorem{theorem}{Theorem}
\newtheorem{assumption}[theorem]{Hypothesis}
\newtheorem{definition}[theorem]{Definition}
\newtheorem{proposition}[theorem]{Proposition}
\newtheorem{lemma}[theorem]{Lemma}
\newtheorem{remark}[theorem]{Remark}
\crefname{theorem}{Theorem}{Theorems}
\crefname{assumption}{Hypothesis}{Hypotheses}
\crefname{definition}{Definition}{Definitions}
\crefname{proposition}{Proposition}{Propositions}
\crefname{lemma}{Lemma}{Lemmas}
\crefname{remark}{Remark}{Remarks}
\crefname{figure}{Figure}{Figures}
\crefname{equation}{}{}
\crefname{align}{}{}
\crefname{multline}{}{}
\newenvironment{proof}[1]{\bigskip \noindent \textit{#1}.}{\hfill $\Box$ \bigskip}

\newcommand{\email}[1]{\footnote{E-mail address: \href{mailto:#1}{#1}}}
\makeatletter

\def\department#1{\def\@department{#1}}

\def\@maketitle{
\null
\vskip 2em
\begin{center}
\vspace{30mm}
{\LARGE \@title \par}
\vspace{5mm}
{\large \@author \par}
\vspace{5mm}
{\small \@department \par}
\vspace{5mm}
{\large \@date}
\end{center}
\par\vskip 1.5em
}

\makeatother

\title{Optimal Switching under Ambiguity and Its Applications in Finance}
\author{Yuki Shigeta\email{sy46744@gmail.com}}
\department{Graduate School of Economics, Kyoto University, Japan}

\begin{document}
\maketitle

\begin{abstract}
In this paper, we study optimal switching problems under ambiguity.
To characterize the optimal switching under ambiguity in the finite horizon, we use multidimensional reflected backward stochastic differential equations (multidimensional RBSDEs) and show that a value function of the optimal switching under ambiguity coincides with a solutions to multidimensional RBSDEs with allowing negative switching costs.
Furthermore, we naturally extend the finite horizon problem to the infinite horizon problem.
In some applications, we show that ambiguity affects an optimal switching strategy with the different way to a usual switching problem without ambiguity.
\end{abstract}





\onehalfspacing

\section{Introduction}\label{sec:introduction}
Optimal switching problems are widely used to describe many situations in finance and economics.
For example, they are applied to natural resource extractions in \cite{brennan1985evaluating} and \cite{brekke1994optimal}, reversible investments in \cite{ly2007explicit}, and entry and exit decisions of firms in \cite{dixit1989afirms}.
In plain words, optimal switching problems are the problems that a decision maker chooses his or her actions from a discrete state space to maximize his or her profit (objective function).

In this paper, our aims are to propose optimal switching problems under ambiguity and to derive general properties of solutions to these problems.
A concept of ambiguity aversion is one of prominent issues in recent finance and economics.
Ambiguity aversion (also known as the Knightian uncertainty aversion or model uncertainty aversion) is the behavior that an economic agent prefers avoiding the event whose occurrence probability is unknown.
Ellsberg \cite{ellsberg1961risk} first provides illustrative examples of ambiguity aversion, 
and these examples are economically characterized in \cite{gilboa1989maxmin} and \cite{schmeidler1989subjective}.
After these works, many researchers study applications of ambiguity aversion such as continuous-time, consumption-investment problems in \cite{chen2002ambiguity} and \cite{liu2011dynamic}, and optimal stopping problems in \cite{riedel2009optimal} and \cite{cheng2013optimal}.

Using a concept of ambiguity aversion, one can describe the properties not captured by a usual trade-off between returns and risks.
Therefore, we can consider a more practical optimal switching problem.
In existing literature, the model in \cite{hamadene2010switching} can be applied to optimal switching problems under ambiguity.
Furthermore, in \cite{bayraktar2016robust}, it is shown that value functions in finite-horizon optimal switching problems under the Knightian uncertainty (ambiguity) are characterized as viscosity solutions to some system of partial differential equations.
The approach in \cite{bayraktar2016robust} allows a more general type of ambiguity than that in this paper, but non-negativity of switching costs is assumed.
In this paper, we allow negative switching costs even though we focus on a specific type of ambiguity,
and we also consider optimal switching problems under ambiguity in infinite horizon.
So, the results in this paper have different implications to \cite{bayraktar2016robust}.

To deal with optimal switching problems under ambiguity, we use frameworks of backward stochastic differential equations (hereafter BSDEs).
BSDEs are introduced in \cite{bismut1978controle} and a general theory of BSDEs is developed in \cite{pardoux1990adapted}.
Many researchers (e.g., \cite{el1997backward}, \cite{rouge2000pricing}, \cite{chen2002ambiguity} and \cite{cheng2013optimal}) apply the theory of BSDEs to various problems in finance and economics.
Recently, a theory of multidimensional reflected BSDEs (hereafter multidimensional RBSDEs) is developed in \cite{hamadene2010switching}, \cite{hu2010multi} and \cite{hamadene2013viscosity} to study the optimal switching problems.
This approach makes us naturally incorporate ambiguity aversion into the optimal switching problems.
Therefore, multidimensional RBSDEs have an important role in this study.

In this paper, our contributions are as follows.
\begin{enumerate}
 \item We characterize optimal switching problems under ambiguity in both of the finite horizon and infinite horizon using multidimensional RBSDEs.
 \item We show that value functions of the optimal switching problems under ambiguity are viscosity solutions to some system of partial differential equations.
 \item Unlike existing literature, we do not assume non-negativity of switching costs.
\end{enumerate}
We first define optimal switching problems under ambiguity and characterize them using the theory of multidimensional RBSDEs in \cite{hamadene2010switching}.
In \cite{hamadene2010switching}, it is assumed that switching costs are non negative, and this assumption has an important role in \cite{hamadene2010switching}.
However, there are optimal switching problems that definitely need negative switching costs (i.e., positive switching benefits) such as the buy low and sell high problem in \cite{zhang2008trading} and the pair-trading problem in \cite{ngo2016optimal}.
Therefore, we do not assume non-negativity of switching costs, and we need to modify the proof in \cite{hamadene2010switching} to allow negative switching costs.
In order to allow negative switching costs, we add a weak assumption for switching costs.
Since existing literature usually assumes non-negativity of switching costs (for example, \cite{hamadene2010switching}, \cite{hu2010multi} and \cite{hamadene2013viscosity}), our results are more general than those of the existing literature in the sense of allowing negative switching costs.
Furthermore, using the results in \cite{hamadene2013viscosity}, we show that value functions of optimal switching problems under ambiguity are viscosity solutions to some system of partial differential equations.

Moreover, we show that under some conditions, value functions in the finite horizon problem converges to value functions in the infinite horizon.
In \cite{el2010optimal}, the infinite horizon problem is investigated with using multidimensional RBSDEs under a non-negativity assumption of switching costs, but the most of existing studies mainly focus on the finite horizon problem.
Therefore, our results may provide new insights in optimal switching problems using multidimensional RBSDEs.

Finally, we give some examples of optimal switching problems under ambiguity in finance.
We show that under certain conditions, optimal switching problems under ambiguity can be interpreted as optimal switching problems under a certain probability measure determined a priori.
Therefore, the results in existing literature can be used to optimal switching problems under ambiguity.
However, the problems not meeting these conditions provide more interesting results.
In \cref{subsec:buy.low.sell.high}, we consider the buy low and sell high problem under ambiguity, which does not satisfy these conditions.
Our results indicate that effects of ambiguity in this problem can not be reproduced by a simple change of a probability measure.

The rest of this paper is organized as follows.
\Cref{sec:preliminaries} defines optimal switching problems under ambiguity in the finite horizon using the concept of multiple priors introduced by \cite{chen2002ambiguity}.
\Cref{sec:rbsde} introduces multidimensional RBSDEs and proves the existence of their solutions.
\Cref{sec:verification.viscosity} verifies that the value functions in the optimal switching problems under ambiguity are characterized by solutions to the multidimensional RBSDEs, and derives the system of partial differential equations which the value functions satisfy.
\Cref{sec:infinite.horizon} considers the infinite horizon problem.
\Cref{sec:financial.appliciation} provides some applications of optimal switching problems under ambiguity in finance.
Lengthy proofs are in Appendix.

\section{Preliminaries and Problem Formulation}\label{sec:preliminaries}
Let $(\Omega,\filt,\prob)$ be a probability space endowed with a $d$-dimensional Brownian motion $W=(W_t)_{t\geq 0}$.
Let $T>0$ be a finite constant time.
We first consider an optimal switching problem during $[0,T]$.
Let $\filts=(\filt_t)_{t\geq 0}$ be an augmentation of the natural filtration generated by $W$.

We denote by $\alpha = (\alpha_t)_{t\geq 0}$ a control process such that
\begin{equation}\label{eq:control}
\alpha_t = \sum_{k\geq 0}i_k\1_{[\tau_k,\tau_{k+1})}(t),
\end{equation}
where $(i_k)_{k\geq 0}$ is a regime process taking values in a discrete state space $\mathcal I = \{1,\dots,I\},$ $I>0$, and $(\tau_k)_{k\geq 0}$ is a non-decreasing sequence of stopping times.
$\1_A(x)$ is an indicator function such that for a given set $A$,
\[
\1_A(x) = \left\{\begin{array}{ll} 1,&\quad \mbox{if }x\in A,\\ 0,&\quad\mbox{otherwise.}\end{array}\right.
\]
We suppose that each $i_k$ is $\filt_{\tau_k}$-measurable.
Under a control $\alpha$, a decision maker chooses a regime $i_k$ on $[\tau_k,\tau_{k+1})$ for all $k\geq 0$.
For convenience, we also write a control as a sequence of pairs of regimes and stopping times: $\alpha = (\tau_k,i_k)_{k\geq 0}$.

Let $X=(X_t)_{0\leq t\leq T}$ be a $d$-dimensional stochastic process satisfying the following stochastic differential equation (hereafter SDE):
\begin{equation}\label{eq:X.SDE}
\diff X_t = b(t,X_t,\alpha_t)\diff t + \sigma(t,X_t,\alpha_t)\diff W_t,
\end{equation}
where $\alpha=(\alpha_t)_{0\leq t\leq T}$ is a control process.
$b$ and $\sigma$ are measurable functions as follows.

\begin{assumption}\label{Assump:X.SDE}
$b:[0,T]\times\real^d\times\mathcal I \rightarrow\real^d$ and $\sigma:[0,T]\times\real^d\times\mathcal I \rightarrow\real^{d\times d}$ satisfy the following Lipschitz condition and quadratic growth condition:
\begin{align*}
\|b(t,x,i) -b(t,y,i)\| + \|\sigma(t,x,i) -\sigma(t,y,i)\| &\leq L\|x-y\|, \\
\|b(t,x,i)\|^2 + \|\sigma(t,x,i)\|^2 &\leq L^2(1 + \|x\|^2),
\end{align*}
for every $t\in[0,T],\;i\in\mathcal I,$ and $x,y\in\real^d$, where $L$ is a positive constant, and $\|x\|$ is the Euclid norm of $x\in\real^d$.
\end{assumption}

Let $L^q_t(\real^d)$ be a set of $d$-dimensional, $q$-th integrable  (that is, an $L^q$ norm on $(\Omega,\filt,\prob)$ is finite), and $\filt_t$-measurable random vectors.
Let $\mathcal T_t^T$ be a set of stopping times taking values in $[t,T]$.
Let $\widetilde{\mathcal I}_t$ be a set of $\filt_t$-measurable random variables taking values in $\mathcal I$.
We define $\widetilde{\mathcal K^q_T}$ and $\overline{\mathcal K_T}$ as follows,
\begin{align*}
\widetilde{\mathcal K^q_T} &:= \left\{(\nu,\eta,\iota)\;|\;\nu\in\mathcal T_0^T,\;\eta\in L_\nu^q(\real^d),\;\iota\in\widetilde{\mathcal I}_\nu\right\},\\
\overline{\mathcal K_T} &:= [0,T]\times \real^d \times \mathcal I.
\end{align*}
By \cref{Assump:X.SDE}, for every $(\nu,\eta,\iota)\in\widetilde{\mathcal K_T^2}$ and progressively measurable control $\alpha$ starting from $\alpha_\nu = \iota$,
there exists a unique strong solution to the SDE \cref{eq:X.SDE} on $[\nu,T]$ starting from $X_\nu=\eta$ and controlled by $\alpha$.
We denote this controlled process by $X^{\nu,\eta,\iota,\alpha} = (X^{\nu,\eta,\iota,\alpha}_s)_{\nu\leq s\leq T}$.
Furthermore, it is well known that the moments of $X$ is upper bounded (e.g., Corollary 2.5.12 in \cite{krylov1980controlled} and Theorem 5.2.9 in \cite{karatzas1991brownian}).
We shortly summarize the results of the moment estimates of $X$.

\begin{proposition}\label{Prop:q.th.integrable}
Under \cref{Assump:X.SDE}, for every $q>0$, there exist constants $C_{q,X}\geq 1$ and $C_q >0$ such that
\begin{align*}
\expect\left[\max_{t\leq s\leq T}\|X_s^{t,x,i,\alpha}\|^q\right] &\leq C_{q,X}(1 + \|x\|^q)e^{C_q(T-t)},
\end{align*}
for all $0\leq t\leq T,\;x\in\real^d,\;i\in\mathcal I$ and control $\alpha$.
Note that $C_{q,X}$ and $C_q$ do not depend on $t,T,x,i$ and $\alpha$.
Furthermore, if a constant $\rho$ is sufficiently large such that $\rho>C_q$, then there exists a positive constant $C^\infty_{q,X}$ such that
\begin{equation}\label{eq:Lp.estimate.X.infinite}
\expect\left[\max_{s\geq t}e^{-\rho s}\Big(1 + \|X_s^{t,x,i,\alpha}\|^q\Big)\right] \leq C^\infty_{q,X}(1 + \|x\|^q)e^{-(\rho - C_q)t},
\end{equation}
for all $0\leq t,\;x\in\real^d,\;i\in\mathcal I$ and control $\alpha$.
Note that $C^\infty_{q,X}$ does not depend on $t,x,i$ and $\alpha$.
\end{proposition}
The proof of \cref{Prop:q.th.integrable} is in \cref{subsec:moment.estimate.X}.
Moreover, we can easily show that the results of \cref{Prop:q.th.integrable} hold in the case when the initial time is a stopping time.
For every $\nu\in\mathcal T_0^T,\;\eta\in L_\nu^{2q}(\real^d),\;i\in\mathcal I$ and control $\alpha$,
we have
\[
\expect\left[\max_{\nu\leq s\leq T}\|X_s^{\nu,\eta,i,\alpha}\|^q\;\Big|\;\filt_\nu\right] \leq C_{q,X}(1 + \|\eta\|^q)e^{C_q(T-\nu)}.
\]

We first consider an optimal switching problem without ambiguity.
An objective function of the optimal switching problem without ambiguity is
\begin{multline}\label{eq:objective.wo.ambiguity}
J^{na}(t,x,i,\alpha) 
:= \expect\biggl[\int_t^TD_s^{t,x,i,\alpha}\psi(s,X_s^{t,x,i,\alpha},\alpha_s)\diff s\\
+ D_T^{t,x,i,\alpha} g(X_T^{t,x,i,\alpha},\alpha_T) 
- \sum_{t\leq \tau_k\leq T}D_{\tau_k}^{t,x,i,\alpha} c_{i_{k-1},i_k}(\tau_k,X_{\tau_k}^{t,x,i,\alpha})\;\Big|\;\filt_t\biggr],
\end{multline}
where $\psi,g,$ and $c$ are measurable functions.
$\psi$ represents running rewards for the switching problem without ambiguity.
$g$ represents a terminal payoff.
$c$ is a switching cost function.
$c_{i,j}(t,x)$ represents a switching cost from regime $i$ to $j$ at time $t$ and $X_t=x$.
$D^{t,x,i,\alpha}$ is a discount factor such that for any $(t,x,i)\in\overline{\mathcal K_T}$ and control $\alpha$,
\begin{equation}\label{eq:discount.factor}
D^{t,x,i,\alpha}_s = \exp\left\{-\int_t^s\rho(t,X_u^{t,x,i,\alpha},\alpha_u)\diff u\right\},\quad s\in[t,T],
\end{equation}
where $\rho(t,x,i)$ is a bounded measurable function.
By the definition \cref{eq:discount.factor}, we allow the discount rate to be random and controllable.
Therefore, the objective function \cref{eq:objective.wo.ambiguity} represents the expected and discounted total profit on $[t,T]$.

For all $\nu\in \mathcal T_0^T$ and $\iota\in\widetilde{\mathcal I}_\nu$, let $\control{\nu}{\iota}$ be a set of controls such that
\begin{equation}\label{eq:admissible.control.set}
\control{\nu}{\iota}
:= \left\{\alpha=(\alpha_s)_{\nu\leq s\leq T}\;\Big|\;\begin{array}{c}
 \expect\left[\Big|\sum_{\nu\leq \tau_k\leq T}c_{i_{k-1},i_k}(\tau_k,X_{\tau_k}^{\nu,x,\iota,\alpha})\Big|^2\right] < \infty,\;\forall x\in\real^d, \\
 \mbox{and }\alpha_\nu = \iota.
\end{array}
 \right\}.
\end{equation}
We call a control in $\control{\nu}{\iota}$ an admissible control.
The optimal switching problem without ambiguity is
\begin{equation}\label{eq:optimal.switching.wo.ambiguity}
\sup_{\alpha\in\control{t}{i}}J^{na}(t,x,i,\alpha),
\end{equation}
for all $(t,x,i)\in\overline{\mathcal K_T}$.

The optimal switching problems expressed as \cref{eq:optimal.switching.wo.ambiguity} are well studied by many researchers (e.g., \cite{brekke1994optimal}, \cite{ly2007explicit}, \cite{djehiche2009finite}, and \cite{bayraktar2010one}).
However, one of the weakness of the optimal switching problem \cref{eq:optimal.switching.wo.ambiguity} is not to take into account ambiguity.
The problem \cref{eq:optimal.switching.wo.ambiguity} assumes that the decision maker knows functional forms of the distribution parameters $b$ and $\sigma$ a priori, whereas we do not know them in practice.
Therefore, it needs to take into account uncertainty about the distribution of $X$ in order to derive more useful switching strategies.
Hence, we consider an optimal switching problem under ambiguity hereafter.

We first define a set of degrees of ambiguity.
For $t\in[0,T]$, let $\Theta_t$ be a set of $d$-dimensional $\filt_t$-measurable random variables.
We assume the form of $\Theta_t$ as follows.
\begin{assumption}\label{Assump:rect}
\
\begin{enumerate}
 \item There exists a non-negative constant $C$ such that
 \[
  \prob(\|\theta_t\|\leq C,\;\forall \theta_t\in\Theta_t,\;t\in[0,T]) = 1.
 \]
 \item $\Theta_t$ is convex and compact valued for all $t\in[0,T]$.
 \item $\Theta_t$ is a progressively measurable correspondence for all $t\in[0,T]$.
 \item $0\in\Theta_t$ $\diff t\otimes\prob$-a.e..
\end{enumerate}
\end{assumption}
Let
\begin{equation*}
\Theta[t,T] := \left\{\theta=(\theta_s)_{t\leq t\leq T}\;\Big|\; \begin{array}{c}
 \theta \mbox{ is right-continuous with left limits and} \\
 \theta_s\in\Theta_s\mbox{ for all $s\in[t,T]$.}
\end{array}
\right\}.
\end{equation*}
For all $\theta\in\Theta[t,T]$, we define a density process $\zeta^{\theta,t} = (\zeta^{\theta,t}_s)_{t\leq s\leq T}$ such that
\begin{equation*}
\zeta_s^{\theta,t} := \exp\left\{-\int_t^s\theta_u^\prime\diff W_u - \frac{1}{2}\int_t^s\|\theta_u\|^2\diff u\right\},\quad s\in[t,T],
\end{equation*}
where $x^\prime$ is a transpose of a vector $x\in\real^d$.
By \cref{Assump:rect}, for all $\theta\in\Theta[t,T]$, $\zeta^{\theta,t}$ is a martingale with respect to $\filts$.
Therefore, for all $\theta\in\Theta[t,T]$, we can define a new probability measure such that
\begin{equation*}
\prob^{\theta}_T(A) := \expect[\1_A\zeta_T^{\theta,t}],\quad A\in\filt_T.
\end{equation*}
We denote by $\expect_T^{\theta}$ the expectation operator under the probability measure $\prob^{\theta}_T$.

Under the probability measure $\prob^\theta_T$, by the Girsanov theorem, the SDE \cref{eq:X.SDE} can be expressed as
\[
\diff X_t = \Big(b(t,X_t,\alpha_t) - \sigma(t,X_t,\alpha_t)\theta_t\Big)\diff t + \sigma(t,X_t,\alpha_t)\diff W_t^\theta,\; t\in[0,T],
\]
where $W^\theta$ is a $d$-dimensional Brownian motion under $\prob^\theta_T$.
This implies that we can take account of the ambiguity about the drift of $X$ under $\prob^\theta_T$.

$\Theta$ represents a set of priors of the decision maker.
In Chen and Epstein (2002) \cite{chen2002ambiguity}, a decision making problem under ambiguity in continuous time is studied, in which the decision maker would like to avoid the event whose occurrence probability is unknown.
To incorporate ambiguity into an optimal switching problem, we use the concept in \cite{chen2002ambiguity}.
In the model in \cite{chen2002ambiguity}, the decision maker chooses his or her subjective probability measure before choosing her decision as if his or her expected utility is minimized.
Chen and Epstein succeed to pose such a decision making problem under \cref{Assump:rect}.
They call \cref{Assump:rect} the rectangular condition.

The objective function under ambiguity is
\begin{multline*}
J(t,x,i,\alpha) 
:= \inf_{\theta\in\Theta[t,T]}\expect_T^{\theta}\biggl[\int_t^TD_s^{t,x,i,\alpha}\Big(\psi(s,X_s^{t,x,i,\alpha},\alpha_s) - \theta_s^\prime\phi(s,X_s^{t,x,i,\alpha},\alpha_s)\Big)\diff s \\
+ D_T^{t,x,i,\alpha} g(X_T^{t,x,i,\alpha},\alpha_T) 
- \sum_{t\leq \tau_k\leq T}D_{\tau_k}^{t,x,i,\alpha} c_{i_{k-1},i_k}(\tau_k,X_{\tau_k}^{t,x,i,\alpha})\;\Big|\;\filt_t\biggr],
\end{multline*}
where $\phi$ is a measurable function from $[0,T]\times\real^d\times\mathcal I$ onto $\real^d$.
$\phi$ determines a running premium for ambiguity.
Our settings allow choices of ambiguity levels to affect the running rewards through the term $\theta_\cdot^\prime\phi(\cdot,X_\cdot^{t,x,i},\alpha_\cdot)$.
The optimal switching problem under ambiguity is
\begin{equation*}
\sup_{\alpha\in\control{t}{i}}J(t,x,i,\alpha),
\end{equation*}
for all $(t,x,i)\in\overline{\mathcal K_T}$.

Furthermore, we assume the functions, $\rho,\psi,\phi,g,$ and $c$ as follows.
\begin{assumption}\label{Assump:funcs}
\
\begin{enumerate}
 \item $\rho(\cdot,\cdot,i)$ is a continuous, non-negative and upper bounded function for all $i\in\mathcal I$.
 \item \textit{Polynomial growth condition} \\
 $\psi(\cdot,\cdot,i),\;\phi(\cdot,\cdot,i),\;g(\cdot,i)$ and $c_{i,j}(\cdot,\cdot)$ are continuous for all $i,j\in\mathcal I$, and $c_{i,i}(t,x) = 0$ for all $(t,x,i)\in\overline{\mathcal K_T}$.
 Furthermore, there exist positive constants $C_f$ and $q$ such that
 \begin{equation*}
  |\psi(t,x,i)| + \|\phi(t,x,i)\| + |g(x,i)| + |c_{i,j}(t,x)| \leq C_f(1 + \|x\|^q),
 \end{equation*}
 for all $(t,x,i,j)\in[0,T]\times\real^d\times(\mathcal I)^2$. Without loss of generality, we assume $q\geq 1$.
 \item \textit{Non-free loop conditions}
 \
       \begin{enumerate}
       \item For all finite loops $(i_0,i_1,\dots,i_m)\in\mathcal I^{m+1}$ with $i_0 = i_m$ and $i_0\neq i_1$ and for all $(t,x)\in[0,T]\times\real^d$, $c$ satisfies
       \begin{equation*}
        c_{i_0,i_1}(t,x) + \cdots + c_{i_{m-1},i_m}(t,x) > 0.
       \end{equation*}
       \item $g$ satisfies the following inequality,
       \begin{equation*}
        g(x,i) \geq \max_{j\in\mathcal I\setminus\{i\}}\{g(x,j) - c_{i,j}(T,x)\},
       \end{equation*}
       for all $(x,i)\in\real^d\times\mathcal I$.
       \end{enumerate}
 \item \textit{Strong triangular condition} \\
 Let
 \begin{align*}
  \mathcal N &= \left\{i\in\mathcal I \;\Big|\;\exists j\in\mathcal I,\;j\neq i,\; \int_{[0,T]\times\real^d}\1\{c_{i,j}(t,x) < 0\}(t,x)\diff t\diff x > 0\right\}, \\
  C_i &= -\min_{j\in\mathcal I,\;x\in\real^d,\;t\in[0,T]}\frac{c_{i,j}(t,x)}{1 + \|x\|^q},\quad i\in\mathcal N,
 \end{align*}
 where $q$ is defined in \cref{Assump:funcs}.2.
 Then, for all $i\in\mathcal N$,
 \begin{equation}\label{eq:strong.triangular}
  c_{k,j}(t,x) \leq c_{k,i}(t,x) - C_i(1 + C_{q,X}(1 + \|x\|^q)e^{C_q(T - t)}),
 \end{equation}
 for all $t\in[0,T],\;x\in\real^d$ and $(j,k)\in\mathcal I$ with $j\neq i$ and $k\neq i$, where $C_{q,X}$ and $C_q$ are defined in \cref{Prop:q.th.integrable}.
\end{enumerate}
\end{assumption}

\cref{Assump:funcs}.1 implies that the discount rate is upper bounded and non-negative.
The non-negativity is usual, and the assumption of upper boundedness guarantees the Lipschitz condition of a generator in the BSDE literature.
\cref{Assump:funcs}.2 and \cref{Prop:q.th.integrable} guarantee the value function of our optimal switching problems to be finite.
Therefore, it is needed in order to consider meaningful problems.

The non-free loop conditions (\cref{Assump:funcs}.3) say that whenever one first stands in some regime (call regime $A$), next instantaneously goes to the other regimes, and finally goes back to the regime $A$ at the same time, then he or she has to pay a positive cost.
Hence, the non-free loop conditions exclude the possibility that one can gain a positive profit by a looping switching strategy at the same time.
If the non-free loop conditions are not postulated, then the value function diverges as the decision maker obtains an infinitely large reward by such a looping strategy.
Since it is an arbitrage, the non-free loop conditions are natural in the optimal switching problems.

Unlike the previous literature, we do not assume non-negativity of the cost functions.
Our specification of ambiguity allows this generalization.
However, we need an additional assumption in this case.
If some cost function can take a negative value, it needs to satisfy the strong triangular condition (\cref{Assump:funcs}.4).

The strong triangular condition means that the switching benefits are not too large to take these benefits.
Heuristically speaking, if one first stands in the regime $k$ and if $c_{i,j} < 0$, then the cost that he or she goes to the regime $j$ via the regime $i$ is at least as large as the cost that he or she directly goes to the regime $j$.
The strong triangular condition implies the standard triangle inequality.
Indeed, by the inequality \cref{eq:strong.triangular}, we have
\begin{align*}
c_{k,i}(t,x) + c_{i,j}(t,x)
 &\geq c_{k,i}(t,x) - C_i(1 + \|x\|^q) \\
 &\geq c_{k,i}(t,x) - C_i(1 + C_{q,X}(1 + \|x\|^q)e^{C_q(T - t)}) 
 \geq c_{k,j}(t,x),
\end{align*}
for all $i\in\mathcal N,\;(j,k)\in\mathcal I,\;(t,x)\in[0,T]\times\real^d$ with $k\neq i$ and $j\neq i$.
Therefore, our triangular condition \cref{eq:strong.triangular} is stronger than the standard triangle inequality.

By \cref{Prop:q.th.integrable,Assump:funcs}, we can show that an expected total cost does not diverge for every admissible control.
\begin{proposition}\label{Prop:cost.upper.bounded}
Under \cref{Assump:X.SDE,Assump:funcs},
\begin{equation}\label{eq:cost.upper.bounded}
\expect\left[-\sum_{t\leq \tau_k\leq T}D_{\tau_k}^{t,x,i,\alpha} c_{i_{k-1},i_k}(\tau_k,X_{\tau_k}^{t,x,i,\alpha})\right] \leq C_f(1 + C_{q,X}(1 + \|x\|^q)e^{C_q(T-t)}),
\end{equation}
for all $(t,x,i)\in\overline{\mathcal K_T}$ and $\alpha=(\tau_k,i_k)_{k\geq 0}\in\control{t}{i}$.
\end{proposition}
\begin{proof}[Proof of \cref{Prop:cost.upper.bounded}]
Fix an arbitrary $(t,x,i)\in\overline{\mathcal K_T}$ and $\alpha=(\tau_k,i_k)_{k\geq 0}\in\control{t}{i}$.
We first prove
\[
\expect\left[-\sum_{k=1}^n D_{\tau_k}^{t,x,i,\alpha} c_{i_{k-1},i_k}(\tau_k,X_{\tau_k}^{t,x,i,\alpha})\right] \leq C_f(1 + C_{q,X}(1 + \|x\|^q)e^{C_q(T-t)}),
\]
for all $n\geq 1$.
If $\prob(i_{n-1}\in\mathcal N\;|\;\filt_{\tau_{n-1}}) = 0$, then $c_{i_{n-1},i_n}(\tau_n,X_{\tau_n}^{t,x,i,\alpha})\geq 0$.
Hence, we have
\begin{multline}\label{eq:positive.cost.inequality}
-D_{\tau_{n-1}}^{t,x,i,\alpha} c_{i_{n-2},i_{n-1}}(\tau_{n-1},X_{\tau_{n-1}}^{t,x,i,\alpha})-D_{\tau_n}^{t,x,i,\alpha} c_{i_{n-1},i_n}(\tau_n,X_{\tau_n}^{t,x,i,\alpha}) \\
\leq -D_{\tau_{n-1}}^{t,x,i,\alpha} c_{i_{n-2},i_{n-1}}(\tau_{n-1},X_{\tau_{n-1}}^{t,x,i,\alpha}).
\end{multline}
If $\prob(i_{n-1}\in\mathcal N\;|\;\filt_{\tau_{n-1}}) > 0$, then, by \cref{Prop:q.th.integrable}, we have
\begin{align*}
&\expect\left[-D_{\tau_{n-1}}^{t,x,i,\alpha} c_{i_{n-2},i_{n-1}}(\tau_{n-1},X_{\tau_{n-1}}^{t,x,i,\alpha})-D_{\tau_n}^{t,x,i,\alpha} c_{i_{n-1},i_n}(\tau_n,X_{\tau_n}^{t,x,i,\alpha})\;\Big|\;\filt_{\tau_{n-1}}\right] \\
&\qquad\leq -\expect\Big[D_{\tau_{n-1}}^{t,x,i,\alpha}\Big( c_{i_{n-2},i_{n-1}}(\tau_{n-1},X_{\tau_{n-1}}^{t,x,i,\alpha}) 
 -  C_{i_{n-1}}\Big(1 + \|X_{\tau_n}^{t,x,i,\alpha}\|^q\Big)\Big)\1_{\{i_{n-1}\in\mathcal N\}}\\
&\qquad\qquad\qquad\qquad + D_{\tau_{n-1}}^{t,x,i,\alpha}c_{i_{n-2},i_{n-1}}(\tau_{n-1},X_{\tau_{n-1}}^{t,x,i,\alpha})\1_{\{i_{n-1}\notin\mathcal N\}} \;\Big|\;\filt_{\tau_{n-1}}\Big] \\
&\qquad\leq -\expect\Big[D_{\tau_{n-1}}^{t,x,i,\alpha}\Big( c_{i_{n-2},i_{n-1}}(\tau_{n-1},X_{\tau_{n-1}}^{t,x,i,\alpha}) \\
&\qquad\qquad\qquad\qquad - C_{i_{n-1}}\Big(1 + C_{q,X}\Big(1 + \|X_{\tau_{n-1}}^{t,x,i,\alpha}\|^q\Big)e^{C_q(T-\tau_{n-1})}\Big)\Big)\1_{\{i_{n-1}\in\mathcal N\}}\\
&\qquad\qquad\qquad\qquad + D_{\tau_{n-1}}^{t,x,i,\alpha}c_{i_{n-2},i_{n-1}}(\tau_{n-1},X_{\tau_{n-1}}^{t,x,i,\alpha})\1_{\{i_{n-1}\notin\mathcal N\}} \;\Big|\;\filt_{\tau_{n-1}}\Big].
\end{align*}
By \cref{Assump:funcs}.4, there exists an $\filt_{\tau_{n-1}}$-measurable random variable $\widetilde i_{n-1}$ taking values in $\mathcal I$ such that
\begin{align*}
&-\expect\Big[D_{\tau_{n-1}}^{t,x,i,\alpha}\Big( c_{i_{n-2},i_{n-1}}(\tau_{n-1},X_{\tau_{n-1}}^{t,x,i,\alpha}) \\
&\qquad\qquad\qquad - C_{i_{n-1}}\Big(1 + C_{q,X}\Big(1 + \|X_{\tau_{n-1}}^{t,x,i,\alpha}\|^q\Big)e^{C_q(T-\tau_{n-1})}\Big)\Big)\1_{\{i_{n-1}\in\mathcal N\}}\;\Big|\;\filt_{\tau_{n-1}}\Big] \\
&\leq -\expect\Big[D_{\tau_{n-1}}^{t,x,i,\alpha} c_{i_{n-2},\widetilde i_{n-1}}(\tau_{n-1},X_{\tau_{n-1}}^{t,x,i,\alpha})\1_{\{i_{n-1}\in\mathcal N\}}\;\Big|\;\filt_{\tau_{n-1}}\Big].
\end{align*}
Hence, we obtain
\begin{align}\label{eq:negative.cost.inequality}
&\expect\left[-D_{\tau_{n-1}}^{t,x,i,\alpha} c_{i_{n-2},i_{n-1}}(\tau_{n-1},X_{\tau_{n-1}}^{t,x,i,\alpha})-D_{\tau_n}^{t,x,i,\alpha} c_{i_{n-1},i_n}(\tau_n,X_{\tau_n}^{t,x,i,\alpha})\;\Big|\;\filt_{\tau_{n-1}}\right]  \\
&\qquad\leq\quad -\expect\Big[D_{\tau_{n-1}}^{t,x,i,\alpha}\Big(c_{i_{n-2},\widetilde i_{n-1}}(\tau_{n-1},X_{\tau_{n-1}}^{t,x,i,\alpha})\1_{\{i_{n-1}\in\mathcal N\}} \nonumber\\
&\qquad\qquad\qquad\quad+ c_{i_{n-2},i_{n-1}}(\tau_{n-1},X_{\tau_{n-1}}^{t,x,i,\alpha})\1_{\{i_{n-1}\notin\mathcal N\}}\Big)\;\Big|\;\filt_{\tau_{n-1}}\Big] \nonumber \\
&\qquad\leq -\expect\Big[D_{\tau_{n-1}}^{t,x,i,\alpha}c_{i_{n-2},i_{n-1}^*}(\tau_{n-1},X_{\tau_{n-1}}^{t,x,i,\alpha})\;\Big|\;\filt_{\tau_{n-1}}\Big],\nonumber
\end{align}
where
\[
i_{n-1}^* = \arg\min_{j\in\mathcal I\setminus\{i_{n-2}\}}\left\{c_{i_{n-2},j}(\tau_{n-1},X_{\tau_{n-1}}^{t,x,i,\alpha})\right\},
\]
and $i_{n-1}^*$ is obviously $\filt_{\tau_{n-1}}$-measurable.
Therefore, the inequalities \cref{eq:positive.cost.inequality,eq:negative.cost.inequality} lead to
\begin{align*}
\expect\left[-\sum_{k=1}^n D_{\tau_k}^{t,x,i,\alpha} c_{i_{k-1},i_k}(\tau_k,X_{\tau_k}^{t,x,i,\alpha})\right] 
&\leq \expect\left[-D_{\tau_1}^{t,x,i,\alpha} c_{i,i_1^*}(\tau_1,X_{\tau_1}^{t,x,i,\alpha})\right] \\
&\leq C_f\left(1 + \expect\left[\|X_{\tau_1}^{t,x,i,\alpha}\|^q\right]\right) \\
&\leq C_f\left(1 + \expect\left[\max_{t\leq s\leq T}\|X_{s}^{t,x,i,\alpha}\|^q\right]\right) \\
&\leq C_f(1 + C_{q,X}(1 + \|x\|^q)e^{C_q(T-t)}).
\end{align*}
Since $\alpha\in\control{t}{i}$, by the Lebesgue dominated convergence theorem, we obtain the inequality \cref{eq:cost.upper.bounded}.
\end{proof}

\Cref{Prop:cost.upper.bounded} has an important role in our switching problem.
The other studies assuming non-negativity of switching costs naturally derive a lower boundary of the total expected costs, this is 0.
However, we do not naturally say that the total costs are non-negative since our switching costs can take a negative value.
Therefore, we need to estimate a lower boundary of the total expected costs by \cref{Prop:cost.upper.bounded}.

\begin{remark}\label{Remark:cost.upper.bounded}
Even if the cost functions do not satisfy the strong triangular condition,
it is possible that \cref{Prop:cost.upper.bounded} holds.
In this case, the following discussion in this paper also holds.
Essentially, we need
\[
\expect\left[-\sum_{t\leq \tau_k\leq T} D_{\tau_k}^{t,x,i,\alpha} c_{i_{k-1},i_k}(\tau_k,X_{\tau_k}^{t,x,i,\alpha})\right] \leq C(1 + \|x\|^q),
\]
for all $(t,x,i)\in\overline{\mathcal K_T}$ and $\alpha\in\control{t}{i}$, where $C$ is a positive constant not depending on $(t,x,i)$ and $\alpha$.
\end{remark}

\section{Multidimensional Reflected BSDEs}\label{sec:rbsde}
Next, we consider a representation of the objective function by BSDEs.

For all $\nu\in\mathcal T_0^T$, we denote by $\SSd{\nu}$ the set of real-valued progressively measurable processes $Y$ such that
\[
\expect\left[\sup_{\nu\leq t\leq T}|Y_t|^2\right] < \infty,
\]
and by $\HH{\nu}{d}$ the set of $\real^d$-valued progressively measurable processes $Z$ such that
\[
\expect\left[\int_\nu^T\|Z_t\|^2\diff t\right] < \infty.
\]
Especially, we denote by $\mathbb S^2_c[\nu,T]$ a set of all continuous processes in $\SSd{\nu}$ and by $\KK{\nu}$ a set of all non-decreasing processes in $\SSd{\nu}$.

We consider the following BSDE: For given $(\nu,\eta,\iota)\in \widetilde{\mathcal K^{2q}_T},\;\theta\in\Theta[\nu,T]$ and $\alpha\in\control{\nu}{\iota}$,
\begin{align}\label{eq:BSDE.general}
-\diff Y_t^{\nu,\eta,\iota,\theta,\alpha} 
&= \Big(\psi(t,X_t^{\nu,\eta,\iota,\alpha},\alpha_t) - \rho(t,X_t^{\nu,\eta,\iota,\alpha},\alpha_t) Y_t^{\nu,\eta,\iota,\theta,\alpha} \nonumber\\
&\quad - \theta_t^\prime\Big(\phi(t,X_t^{\nu,\eta,\iota,\alpha},\alpha_t) + Z_t^{\nu,\eta,\iota,\theta,\alpha} \Big)\Big)\diff t  \nonumber \\
&\quad - (Z_t^{\nu,\eta,\iota,\theta,\alpha})^\prime\diff W_t - \diff A_t^{\nu,\eta,\iota,\alpha},\; t\in[\nu,T], \\
Y_T^{\nu,\eta,\iota,\theta,\alpha} &= g(X_T^{\nu,\eta,\iota,\alpha},\alpha_T), \quad
A_t^{\nu,\eta,\iota,\alpha} = \sum_{t\leq \tau_k\leq T}c_{i_{k-1},i_k}(\tau_k,X^{\nu,\eta,\iota,\alpha}_{\tau_k}),\; t\in[\nu,T], \nonumber\\
&(Y^{\nu,\eta,\iota,\theta,\alpha},Z^{\nu,\eta,\iota,\theta,\alpha}) \in \SSd{\nu}\times\HH{\nu}{d}. \nonumber
\end{align}
Since $g(X_T^{\nu,\eta,\iota,\alpha},\alpha_T) \in L_T^2(\real)$ and $(\phi(t,X_t^{\nu,\eta,\iota,\alpha},\alpha_t))_{\nu\leq t\leq T},\;(\psi(t,X_t^{\nu,\eta,\iota,\alpha},\alpha_t))_{\nu\leq t\leq T}\in\HH{\nu}{1}$ and since $\theta$ and $\rho$ are uniformly bounded by \cref{Assump:X.SDE,Assump:rect,Assump:funcs},
the BSDE \cref{eq:BSDE.general} has a unique solution in $\SSd{\nu}\times\HH{\nu}{d}$.
Furthermore, by Proposition 2.2 in \cite{el1997backward}, the solution of the BSDE \cref{eq:BSDE.general}, also denoted by $(Y_t^{\nu,\eta,\iota,\theta,\alpha},Z_t^{\nu,\eta,\iota,\theta,\alpha})_{\nu\leq t\leq T}$, can be represented as the following form.
\begin{align}\label{eq:BSDE.probabilistic.representation}
Y_t^{\nu,\eta,\iota,\theta,\alpha} &= \frac{1}{D^{\nu,\eta,\iota,\alpha}_t\zeta_t^{\theta,\nu}}\expect\left[\int_t^T D^{\nu,\eta,\iota,\alpha}_s\zeta^{\theta,\nu}_s\Big(\psi(s,X_s^{\nu,\eta,\iota,\alpha},\alpha_s) - \theta_s^\prime\phi(s,X_s^{\nu,\eta,\iota,\alpha},\alpha_s)\Big)\diff s\right. \\
&\quad \left. + D^{\nu,\eta,\iota,\alpha}_T\zeta_T^{\theta,\nu} g(X_T^{\nu,\eta,\iota,\alpha},\alpha_T) - \sum_{t\leq \tau_k\leq T}D^{\nu,\eta,\iota,\alpha}_{\tau_k}\zeta^{\theta,\nu}_{\tau_k}c_{i_{k-1},i_k}(\tau_k,X^{\nu,\eta,\iota,\alpha}_{\tau_k})\;\Big|\;\filt_t\right] \nonumber\\
 &= \expect_T^{\theta}\left[\int_t^T\frac{D^{\nu,\eta,\iota,\alpha}_s}{D^{\nu,\eta,\iota,\alpha}_t}\Big(\psi(s,X_s^{\nu,\eta,\iota,\alpha},\alpha_s) - \theta_s^\prime\phi(s,X_t^{\nu,\eta,\iota,\alpha},\alpha_s)\Big)\diff s \right. \nonumber\\
 &\quad\quad \left.+ \frac{D^{\nu,\eta,\iota,\alpha}_T}{D^{\nu,\eta,\iota,\alpha}_t}g(X_T^{0,x,i,\alpha},\alpha_T) - \sum_{t\leq \tau_k\leq T}\frac{D^{\nu,\eta,\iota,\alpha}_{\tau_k}}{D^{\nu,\eta,\iota,\alpha}_t}c_{i_{k-1},i_k}(\tau_k,X^{\nu,\eta,\iota,\alpha}_{\tau_k})\;\Big|\;\filt_t\right],\nonumber
\end{align}
where we have used the Bayes rule in the second equality.

Now, we also consider another BSDE such that
\begin{align}\label{eq:BSDE.minimal.temp}
-\diff Y_t^{\nu,\eta,\iota,\alpha} 
&= \Big(\psi(t,X_t^{\nu,\eta,\iota,\alpha},\alpha_t)  - \rho(t,X_t^{\nu,\eta,\iota,\alpha},\alpha_t) Y_t^{\nu,\eta,\iota,\alpha} \nonumber\\
&\qquad\qquad - \max_{\theta_t\in\Theta_t}\left\{\theta_t^\prime \Big(\phi(t,X_t^{\nu,\eta,\iota,\alpha},\alpha_t)+ Z_t^{\nu,\eta,\iota,\alpha}\Big)\right\} \Big)\diff t  \nonumber\\
&\qquad\qquad - (Z_t^{\nu,\eta,\iota,\alpha})^\prime\diff W_t - \diff A_t^{\nu,\eta,\iota,\alpha},\; t\in[\nu,T], \\
Y_T^{\nu,\eta,\iota,\alpha} &= g(X_T^{\nu,\eta,\iota,\alpha},\alpha_T),\quad
A_t^{\nu,\eta,\iota,\alpha} = \sum_{t\leq \tau_k\leq T}c_{i_{k-1},i_k}(\tau_k,X^{\nu,\eta,\iota,\alpha}_{\tau_k}),\; t\in[\nu,T], \nonumber \\
&(Y^{\nu,\eta,\iota,\alpha},Z^{\nu,\eta,\iota,\alpha}) \in \SSd{\nu}\times\HH{\nu}{d}. \nonumber
\end{align}
The BSDE \cref{eq:BSDE.minimal.temp} also has a unique solution in $\SSd{\nu}\times\HH{\nu}{d}$.
From the comparison theorem, the solution to the BSDE \cref{eq:BSDE.minimal.temp} is a minimum value of $Y_t^{\nu,\eta,\iota,\theta,\alpha}$ over $\theta\in\Theta[\nu,T]$, that is, the following inequality holds.
\begin{equation}\label{eq:Y.theta.minimal}
Y_t^{\nu,\eta,\iota,\theta,\alpha} \geq Y_t^{\nu,\eta,\iota,\alpha},
\end{equation}
$\prob$-almost surely for all $t\in[\nu,T]$ and $\theta\in\Theta[\nu,T]$.

Combining the inequality \cref{eq:Y.theta.minimal} with the equality \cref{eq:BSDE.probabilistic.representation}, we deduce that
\begin{align*}
Y_t^{t,x,i,\alpha} &= \inf_{\theta\in\Theta[t,T]}\expect_T^{\theta}\left[\int_t^T D^{t,x,i,\alpha}_s\Big(\psi(s,X_s^{t,x,i,\alpha},\alpha_s) - \theta_s^\prime\phi(s,X_s^{t,x,i,\alpha},\alpha_s)\Big)\diff s \right.\\
&\qquad \left.+ D^{t,x,i,\alpha}_T g(X_T^{t,x,i,\alpha},\alpha_T) - \sum_{s\leq \tau_k\leq T}D^{t,x,i,\alpha}_{\tau_k}c_{i_{k-1},i_k}(\tau_k,X^{t,x,i,\alpha}_{\tau_k})\;\Big|\;\filt_t\right] \\
&= J(t,x,i,\alpha),
\end{align*}
for all $(t,x,i)\in\overline{\mathcal K_T}$ and $\alpha\in\control{t}{i}$.
Therefore, $Y_t^{t,x,i,\alpha}$ is the objective function in the optimal switching problem under ambiguity.

For the sake of brevity, we assume for $\Theta_t$ as follows.
\begin{assumption}\label{Assump:varsigma}
Suppose that $\Theta_t$ is measurable with respect to the $\sigma$-algebra generated by $X_t$ and $\alpha_t$ for all $t\in[0,T]$.
We denote by $\Theta_t^{x,i}$ a $\Theta_t$ with $X_t = x$ and $\alpha_t = i$.
For all $(t,x,i)\in\overline{\mathcal K_T}$ and $z\in\real^d$,
let
\[
\varsigma(t,x,i,z) := \max_{\theta_t\in\Theta_t^{x,i}}\left\{\theta_t^\prime \Big(\phi(t,x,i) + z\Big)\right\}.
\]
Then, suppose that $\varsigma$ is a deterministic and measurable function.
Moreover, suppose that $\varsigma(\cdot,\cdot,i,\cdot)$ is continuous for all $i\in\mathcal I$.
\end{assumption}
By \cref{Assump:rect}.1 and \ref{Assump:funcs}, $\varsigma$ satisfy the polynomial growth condition with respect to $x$ and $z$ and the Lipschitz condition with respect to $z$: There exists a positive constant $C_\varsigma$ such that
\begin{align*}
|\varsigma(t,x,i,z)| \leq C_\varsigma(1 + \|x\|^q + \|z\|), \quad
|\varsigma(t,x,i,z) - \varsigma(t,x,i,\widetilde z)| \leq C_\varsigma\|z - \widetilde z\|,
\end{align*}
for all $(t,x,i,z,\widetilde z)\in\overline{\mathcal K_T}\times(\real^d)^2$.

Under \cref{Assump:varsigma},
the BSDE \cref{eq:BSDE.minimal.temp} can be expressed as
\begin{align}\label{eq:BSDE.minimal}
-\diff Y_t^{\nu,\eta,\iota,\alpha} 
&= \Big(\psi(t,X_t^{\nu,\eta,\iota,\alpha},\alpha_t)  - \rho(t,X_t^{\nu,\eta,\iota,\alpha},\alpha_t) Y_t^{\nu,\eta,\iota,\alpha}  - \varsigma(t,X_t^{\nu,\eta,\iota,\alpha},\alpha_t,Z_t^{\nu,\eta,\iota,\alpha}) \Big)\diff t \nonumber\\
&\qquad\qquad - (Z_t^{\nu,\eta,\iota,\alpha})^\prime\diff W_t - \diff A_t^{\nu,\eta,\iota,\alpha},\; t\in[\nu,T], \\
Y_T^{\nu,\eta,\iota,\alpha} &= g(X_T^{\nu,\eta,\iota,\alpha},\alpha_T),\quad
A_t^{\nu,\eta,\iota,\alpha} = \sum_{t\leq \tau_k\leq T}c_{i_{k-1},i_k}(\tau_k,X^{\nu,\eta,\iota,\alpha}_{\tau_k}),\; t\in[\nu,T], \nonumber \\
&(Y^{\nu,\eta,\iota,\alpha},Z^{\nu,\eta,\iota,\alpha}) \in \SSd{\nu}\times\HH{\nu}{d}. \nonumber
\end{align}

\bigskip

Now, let us consider a multidimensional RBSDE.
For given $\nu\in \mathcal T_0^T$ and $\eta\in L^{2q}_\nu(\real^d)$ and for all $i\in\mathcal I$,
\begin{align}\label{eq:multidimensional.RBSDE}
-\diff Y_t^{\nu,\eta,i} &= \Big(\psi(t,X_t^{\nu,\eta,i},i) - \rho(t,X_t^{\nu,\eta,i},i)Y_t^{\nu,\eta,i} - \varsigma(t,X_t^{\nu,\eta,i},i,Z_t^{\nu,\eta,i}) \Big)\diff t \nonumber \\
&\qquad\qquad - (Z_t^{\nu,\eta,i})^\prime\diff W_t + \diff K_t^{\nu,\eta,i},\; t\in[\nu,T],\nonumber \\
Y_T^{\nu,\eta,i} &= g(X_T^{\nu,\eta,i},i),\quad K^{\nu,\eta,i}_\nu = 0, \quad
Y_t^{\nu,\eta,i} \geq \max_{j\in\mathcal I\setminus\{i\}}\{Y_t^{\nu,\eta,j} - c_{i,j}(t,X_t^{\nu,\eta,i})\},\; t\in[\nu,T],\\
&\int_\nu^T\Big(Y_t^{\nu,\eta,i} - \max_{j\in\mathcal I\setminus\{i\}}\{Y_t^{\nu,\eta,j} - c_{i,j}(t,X_t^{\nu,\eta,i})\}\Big)\diff K^{\nu,\eta,i}_t = 0,  \nonumber\\
&(Y^{\nu,\eta,i},Z^{\nu,\eta,i},K^{\nu,\eta,i}) \in \SSd{\nu}\times\HH{\nu}{d}\times\KK{\nu}, \quad i \in\mathcal I, \nonumber
\end{align}
where $X^{\nu,\eta,i}=(X^{\nu,\eta,i}_t)_{\nu\leq t\leq T}$ is a strong solution to the following SDE,
\begin{equation}\label{eq:fix.X.SDE}
\diff X_t = b(t,X_t,i)\diff t + \sigma(t,X_t,i)\diff W_t,\;t\in[\nu,T],\quad X_\nu = \eta.
\end{equation}

In the next section, we show that a solution $Y_t^{t,x,i}$ to the multidimensional RBSDE \cref{eq:multidimensional.RBSDE} is a value function of the optimal switching problem under ambiguity.
In this section, we first prove the existence of solutions to the multidimensional RBSDE \cref{eq:multidimensional.RBSDE}.

\begin{theorem}\label{Thm:existence}
Under \cref{Assump:X.SDE,Assump:rect,Assump:funcs,Assump:varsigma}, the multidimensional RBSDE \cref{eq:multidimensional.RBSDE} has a solution in $(\mathbb S^2_c[\nu,T]\times\HH{\nu}{d}\times\KK{\nu})^I$ for any $\nu\in\mathcal T_0^T$ and $\eta\in L_\nu^{2q}(\real^d)$.
\end{theorem}

When the switching costs are always non-negative,
\cref{Thm:existence} are proved by Theorem 3.2 in \cite{hamadene2010switching} and Theorem 2.1 in \cite{hu2010multi}.
We use the strategy of the proof of Theorem 3.2 in \cite{hamadene2010switching}, but there is a problem for a priori estimates of Picard's iterations of the multidimensional RBSDE \cref{eq:multidimensional.RBSDE}.
In the setting of \cite{hamadene2010switching} i.e., under the assumption that all switching costs are non negative, the process in $\SSd{\nu}$ that is larger than all Picard's iterations can be defined, however, this process may not be larger than Picard's iterations in our problem since we allow the switching costs to be negative.
Therefore, we can not use the results in \cite{hamadene2010switching} straightforwardly.
However, thanks to \cref{Prop:cost.upper.bounded}, we can define the other process in $\SSd{\nu}$ that is larger than all Picard's iterations in our problem.

\begin{proof}[Proof of \cref{Thm:existence}]
\
Throughout this proof, we fix an arbitrary $\nu\in\mathcal T_0^T$ and $\eta\in L_\nu^{2q}(\real^d)$.

\textit{Step.1 Picard's iterations.}
Let $(Y^{\nu,\eta,i,0},Z^{\nu,\eta,i,0})$ be a solution to the following BSDE.
\begin{align*}
-\diff Y_t^{\nu,\eta,i,0} &= \Big(\psi(t,X_t^{\nu,\eta,i},i) - \rho(t,X_t^{\nu,\eta,i},i)Y_t^{\nu,\eta,i,0} - \varsigma(t,X_t^{\nu,\eta,i},i,Z_t^{\nu,\eta,i,0}) \Big)\diff t \\
&\qquad\qquad - (Z_t^{\nu,\eta,i,0})^\prime\diff W_t,\; t\in[\nu,T], \\
Y_T^{\nu,\eta,i,0} &= g(X_T^{\nu,\eta,i},i),\quad (Y^{\nu,\eta,i,0},Z^{\nu,\eta,i,0}) \in \SSd{\nu}\times\HH{\nu}{d},
\end{align*}
for all $i\in\mathcal I$.
Then, by \cref{Assump:X.SDE,Assump:rect,Assump:funcs,Assump:varsigma}, the above BSDE has a unique solution.
For any $n\geq 1$, we consider the following RBSDE recursively.
\begin{align}\label{eq:picard.iteration}
-\diff Y_t^{\nu,\eta,i,n} &= \Big(\psi(t,X_t^{\nu,\eta,i},i) - \rho(t,X_t^{\nu,\eta,i},i)Y_t^{\nu,\eta,i,n} - \varsigma(t,X_t^{\nu,\eta,i},i,Z_t^{\nu,\eta,i,n}) \Big)\diff t \nonumber\\
&\qquad\qquad - (Z_t^{\nu,\eta,i,n})^\prime\diff W_t + \diff K_t^{\nu,\eta,i,n},\; t\in[\nu,T], \nonumber\\
Y_T^{\nu,\eta,i,n} &= g(X_T^{\nu,\eta,i},i), \quad K^{\nu,\eta,i,n}_\nu = 0,\nonumber \\
Y_t^{\nu,\eta,i,n} &\geq \max_{j\in\mathcal I\setminus\{i\}}\{Y_t^{\nu,\eta,j,n-1} - c_{i,j}(t,X_t^{\nu,\eta,i})\},\;t\in[\nu,T], \\
&\int_\nu^T\Big(Y_t^{\nu,\eta,i,n} - \max_{j\in\mathcal I\setminus\{i\}}\{Y_t^{\nu,\eta,j,n-1} - c_{i,j}(t,X_t^{\nu,\eta,i})\}\Big)\diff K^{\nu,\eta,i,n}_t = 0, \nonumber \\
&(Y^{\nu,\eta,i,n},Z^{\nu,\eta,i,n},K^{\nu,\eta,i,n}) \in \SSd{\nu}\times\HH{\nu}{d}\times\KK{\nu},\quad i\in\mathcal I. \nonumber
\end{align}
Under \cref{Assump:X.SDE,Assump:rect,Assump:funcs,Assump:varsigma}, by Theorem 5.2 in \cite{el1997reflected},
the RBSDE \cref{eq:picard.iteration} has a unique solution for all $n$ and $i$.
Furthermore, by the comparison theorem (Theorem 4.1 in \cite{el1997reflected}), 
we have $Y_t^{\nu,\eta,i,n-1}\leq Y_t^{\nu,\eta,i,n},\;\prob$-a.s. for all $i$ and $n$.

\textit{Step.2 Non-ambiguity processes.}
Consider the following BSDE.
\begin{align*}
-\diff U_t^{\nu,\eta,i,0} &= \Big(\psi(t,X_t^{\nu,\eta,i},i) - \rho(t,X_t^{\nu,\eta,i},i)U_t^{\nu,\eta,i,0}\Big)\diff t - (V_t^{\nu,\eta,i,0})^\prime\diff W_t,\; t\in[\nu,T], \\
U_T^{\nu,\eta,i,0} &= g(X_T^{\nu,\eta,i},i),\quad (U^{\nu,\eta,i,0},V^{\nu,\eta,i,0}) \in \SSd{\nu}\times\HH{\nu}{d},\quad i\in\mathcal I.
\end{align*}
Then, the above BSDE has a unique solution.
Similarly, we consider the following RBSDE for any $n\geq 1$.
\begin{align*}
-\diff U_t^{\nu,\eta,i,n} &= \Big(\psi(t,X_t^{\nu,\eta,i},i) - \rho(t,X_t^{\nu,\eta,i},i)U_t^{\nu,\eta,i,n}\Big)\diff t - (V_t^{\nu,\eta,i,n})^\prime\diff W_t + \diff S_t^{\nu,\eta,i,n},\; t\in[\nu,T],\\
U_T^{\nu,\eta,i,n} &= g(X_T^{\nu,\eta,i},i),\quad S^{\nu,\eta,i,n}_\nu = 0, \\
U_t^{\nu,\eta,i,n} &\geq \max_{j\in\mathcal I\setminus\{i\}}\{U_t^{\nu,\eta,j,n-1} - c_{i,j}(t,X_t^{\nu,\eta,i})\},\; t\in[\nu,T],\\
&\int_\nu^T\Big(U_t^{\nu,\eta,i,n} - \max_{j\in\mathcal I\setminus\{i\}}\{U_t^{\nu,\eta,j,n-1} - c_{i,j}(t,X_t^{\nu,\eta,i})\}\Big)\diff S^{\nu,\eta,i,n}_t = 0, \\
&(U^{\nu,\eta,i,n},V^{\nu,\eta,i,n},S^{\nu,\eta,i,n}) \in \SSd{\nu}\times\HH{\nu}{d}\times\KK{\nu},\quad i\in\mathcal I.
\end{align*}
Then, the above RBSDE has a unique solution, and we obtain that $U_t^{\nu,\eta,i,n}\geq U_t^{\nu,\eta,i,n-1},\;\prob$-a.s. for all $(t,i)\in[\nu,T]\times\mathcal I$ and $n\geq 1$ by the comparison theorem.
By the definition of $\varsigma$ and \cref{Assump:rect}.4,
we have
\[
\varsigma(t,x,i,z) \geq 0,\quad \forall (t,x,i,z)\in[0,T]\times\real^d\times\mathcal I\times\real^d.
\]
Hence, applying the comparison theorem again to $U_t^{\nu,\eta,i,n}$ and $Y_t^{\nu,\eta,i,n}$, we obtain that $U_t^{\nu,\eta,i,n}\geq Y_t^{\nu,\eta,i,n},\;\prob$-a.s. for all $(t,i)\in[\nu,T]\times\mathcal I$ and $n\geq 1$.
Furthermore, $U^{\nu,\eta,i,n}$ has a Snell envelope representation such that
\begin{multline*}
U^{\nu,\eta,i,n}_t = \esssup{\tau^*\in\mathcal T_t^T}\expect\biggl[\int_t^{\tau^*}\frac{D_s^{\nu,\eta,i}}{D_t^{\nu,\eta,i}}\psi(s,X_s^{\nu,\eta,i},i)\diff s  + \frac{D_T^{\nu,\eta,i}}{D_t^{\nu,\eta,i}}g(X_T^{\nu,\eta,i},i)\1_{\{\tau^*=T\}}\\
+ \frac{D_{\tau^*}^{\nu,\eta,i}}{D_t^{\nu,\eta,i}}\max_{j\in\mathcal I\setminus\{i\}}\left\{U^{\nu,\eta,j,n-1}_{\tau^*} - c_{i,j}(\tau^*,X_{\tau^*}^{\nu,\xi,i})\right\}\1_{\{\tau^*<T\}}\;\Big|\;\filt_t\biggr],
\end{multline*}
for all $t\in[\nu,T]$ and $n\geq 1$, where
\[
D_t^{\nu,\eta,i} = \exp\left\{-\int_\nu^t \rho(s,X_s^{\nu,\eta,i},i)\diff s\right\}, \quad t\in[\nu,T].
\]

\textit{Step.3 A priori estimates.}
Fix an arbitrary $t\in[\nu,T]$ and an natural number $n$.
Let $(\tau_0,i_0) = (t,i)$ and
\begin{align*}
\tau_k &= \inf\left\{s\in[\tau_{k-1},T]\;\Big|\;U^{\nu,\eta,i_{k-1},n-(k-1)}_{\tau_n} 
 = \max_{j\in\mathcal I\setminus\{i_{n-1}\}}\left\{U^{\nu,\eta,j,n-k}_{\tau_n} - c_{i_{k-1},j}(\tau_k,X_{\tau_k}^{\nu,\eta,i,\alpha})\right\}\right\},\\
i_k &\mbox{ is such that }U^{\nu,\eta,i_{k-1},n-(k-1)}_{\tau_n} = U^{\nu,\eta,i_{k},n-k}_{\tau_n} - c_{i_{k-1},i_k}(\tau_k,X_{\tau_k}^{\nu,\eta,i,\alpha}),
\end{align*}
for all $k=1,\dots,n$.
Then, we define $\alpha^n = (\tau_k,i_k)_{k\geq 0}$ and it holds that
\begin{multline*}
U^{\nu,\eta,i,n}_t = \expect\biggl[\int_t^T D_s^{t,X_t^{\nu,\eta,i},i,\alpha^n}\psi(s,X_s^{\nu,\eta,i,\alpha^n},\alpha_t^n)\diff s
 + D_T^{t,X_t^{\nu,\eta,i},i,\alpha^n}g(X_T^{\nu,\eta,i,\alpha^n},\alpha_T^n) \\
 - \sum_{k=1}^n D_{\tau_k}^{t,X_t^{\nu,\eta,i},i,\alpha^n}c_{i,j}(\tau_k,X_{\tau_k}^{\nu,\eta,i,\alpha^n})\1_{\{\tau_k<T\}}\;\Big|\;\filt_t\biggr],
\end{multline*}
by Proposition 2.3 in \cite{el1997reflected}.
Furthermore, by the polynomial growth condition for $c$, it is easy to check that $\alpha^n$ is in $\control{\nu}{i}$.
Thus, by \cref{Prop:cost.upper.bounded}, we have
\[
\expect\left[- \sum_{k=1}^n D_{\tau_k}^{t,X_t^{\nu,\eta,i},i,\alpha^n}c_{i,j}(\tau_k,X_{\tau_k}^{\nu,\eta,i,\alpha^n})\1_{\{\tau_k<T\}}\;\Big|\;\filt_t\right]
\leq C_f(1 + C_{q,X}(1 + \|X_{t}^{\nu,\eta,i}\|^q)e^{C_{2q}T}).
\]
On the other hand, by \cref{Prop:q.th.integrable}, there exists a constant $C_T>0$ such that
\begin{align*}
&\expect\left[\int_t^T D_s^{t,X_t^{\nu,\eta,i},i,\alpha^n}\psi(s,X_s^{\nu,\eta,i,\alpha^n},\alpha_t^n)\diff s + D_T^{t,X_t^{\nu,\eta,i},i,\alpha^n}g(X_T^{\nu,\eta,i,\alpha^n},\alpha_T^n)\;\Big|\;\filt_t\right] \\
&\leq \expect\left[\int_t^T |\psi(s,X_s^{\nu,\eta,i,\alpha^n},\alpha_t^n)|\diff s + |g(X_T^{\nu,\eta,i,\alpha^n},\alpha_T^n)|\;\Big|\;\filt_t\right] \\
&\leq C_T(1 + \|X_t^{\nu,\eta,i}\|^q).
\end{align*}
Finally, there exists a positive constant $C_M> 0$ such that
\begin{align*}
U^{\nu,\eta,i,n}_t &= \expect\left[\int_t^T D_s^{t,X_t^{\nu,\eta,i},i,\alpha^n}\psi(s,X_s^{\nu,\eta,i,\alpha^n},\alpha_t^n)\diff s \right.
 + D_T^{t,X_t^{\nu,\eta,i},i,\alpha^n}g(X_T^{\nu,\eta,i,\alpha^n},\alpha_T^n) \\
&\qquad\qquad \left. - \sum_{k=1}^n D_{\tau_k}^{t,X_t^{\nu,\eta,i},i,\alpha^n}c_{i,j}(\tau_k,X_{\tau_k}^{\nu,\eta,i,\alpha^n})\1_{\{\tau_k<T\}}\;\Big|\;\filt_t\right] \\
&\leq C_M(1 + \|X_t^{\nu,\eta,i}\|^q).
\end{align*}
Note that $C_M$ does not depend on $n$ and $t$.
This implies that
\[
U_t^{\nu,\eta,i,n} \leq M_t^{\nu,\eta} := C_M\left(1 + \sum_{j\in\mathcal I}\|X_t^{\nu,\eta,j}\|^q\right),
\]
for all $t\in[\nu,T],\;i\in\mathcal I$ and $n\geq 1$.
By \cref{Prop:q.th.integrable}, $M^{\nu,\eta}$ is in $\SSd{\nu}$.
Since $Y_t^{\nu,\eta,i,0}\leq Y_t^{\nu,\eta,i,n}\leq U_t^{\nu,\eta,i,n}\leq M_t^{\nu,\eta}$ for all $t\in[\nu,T],\;i\in\mathcal I$ and $n\geq 1$ and since $Y^{\nu,\eta,i,0}\in\SSd{\nu}$ for all $i\in\mathcal I$,
there exists a finitely positive constant $C_a$ such that
\begin{equation}\label{eq:Y.upper.bound.all.n}
\sum_{i\in\mathcal I}\expect\left[\sup_{\nu\leq t\leq T}|Y_t^{\nu,\eta,i,n}|^2\right] \leq C_a,
\end{equation}
for all $n\geq 0$.
Furthermore, by the polynomial growth condition for $c$, \cref{Prop:q.th.integrable} and the inequality \cref{eq:Y.upper.bound.all.n}, there exists a positive constant $C_b$ such that
\[
\expect\left[\sup_{\nu\leq t\leq T}\Big|\Big(\max_{j\in\mathcal I\setminus\{i\}}\{Y_t^{\nu,\eta,j,n-1} - c_{i,j}(t,X_t^{\nu,\eta,i})\}\Big)^+\Big|^2\right]
\leq C_b,
\]
for all $n\geq 0$.
Hence, Proposition 3.5 in \cite{el1997reflected} leads to that there exists a finitely positive constant $C_c$ such that
\begin{equation}\label{eq:apriori.estimate}
\expect\left[\sup_{\nu\leq t\leq T}|Y_t^{\nu,\eta,i,n}|^2 + \int_\nu^T\|Z_t^{\nu,\eta,i,n}\|^2\diff t + |K_T^{\nu,\eta,i,n}|^2\right] \leq C_c,
\end{equation}
for all $n\geq 0$ and $i\in\mathcal I$.

\textit{Step.4}
The rest of this proof is exactly the same as step 3-5 in the proof of Theorem 3.2 in \cite{hamadene2010switching}.
Thanks to the inequality \cref{eq:apriori.estimate}, we can use the monotone limit theorem in \cite{peng1999monotonic} and show that a limit of $(Y^{\nu,\eta,i,n})_{n\geq 0}$ and associated processes $(Z^{\nu,\eta,i},K^{\nu,\eta,i})$ satisfy properties of the solution to the multidimensional RBSDE \cref{eq:multidimensional.RBSDE}.
This limit, denoted by $(Y^{\nu,\eta,i})$, and $(K^{\nu,\eta,i})$ are continuous by the non-free loop condition.
By the continuity of $(Y^{\nu,\eta,i})$ and $(K^{\nu,\eta,i})$, we conclude that a triplet $(Y^{\nu,\eta,i},Z^{\nu,\eta,i},K^{\nu,\eta,i})$ is a $\SSd{\nu}\times\HH{\nu}{d}\times\KK{\nu}$ limit of the sequence $(Y^{\nu,\eta,i,n},Z^{\nu,\eta,i,n},K^{\nu,\eta,i,n})_{n\geq 0}$.
\end{proof}

\begin{remark}\label{Remark:minimality}
According to Corollary 3.3 in \cite{hamadene2010switching}, the solution $(Y^{\nu,\eta,i})$ constructed in \cref{Thm:existence} is a minimum solution of the multidimensional RBSDE \cref{eq:multidimensional.RBSDE}: For any solution $(\widetilde Y^{\nu,\eta,i})$ of the multidimensional RBSDE \cref{eq:multidimensional.RBSDE},
\[
\widetilde Y^{\nu,\eta,i}_t \geq Y^{\nu,\eta,i}_t, \prob\mbox{-a.s.},
\]
for all $t\in[\nu,T]$ and $i\in\mathcal I$.
\end{remark}

\cref{Thm:existence} provides the existence of the multidimensional RBSDE \cref{eq:multidimensional.RBSDE}.
Other articles prove the uniqueness of the solution after proving the existence.
However, we do not prove the uniqueness.
Instead, we prove the pathwise uniqueness of the minimal solution to the multidimensional RBSDE \cref{eq:multidimensional.RBSDE} since this is a sufficient condition for verification of the optimal switching problem under ambiguity.

\begin{proposition}\label{Prop:minimal.uniqueness}
Suppose \cref{Assump:X.SDE,Assump:rect,Assump:funcs,Assump:varsigma}.
For any $(\nu,\widetilde \nu)\in (\mathcal T_0^T)^2$ and $\eta\in L^{2q}_\nu(\real^d)$ such that $\nu\leq \widetilde \nu$ $\prob$-a.s.,
we consider the minimum solutions of the multidimensional RBSDE \cref{eq:multidimensional.RBSDE} $Y^{\nu,\eta,i}$ and $Y^{\widetilde \nu,X_{\widetilde\nu}^{\nu,\eta,i},i}$.
Then,
\begin{equation}\label{eq:minimal.uniqueness}
Y_t^{\nu,\eta,i} = Y_t^{\widetilde\nu,X_{\widetilde\nu}^{\nu,\eta,i},i}\;\prob\mbox{-a.s.},
\end{equation}
for all $i\in\mathcal I$ and $t\in[\widetilde\nu,T]$.
\end{proposition}
\begin{proof}[Proof of \cref{Prop:minimal.uniqueness}]
By \cref{Assump:X.SDE}, the SDE \cref{eq:fix.X.SDE} has a strong solution for all $i\in\mathcal I$.
This implies that
\begin{equation*}
X_t^{\nu,\eta,i} = X_t^{\widetilde\nu,X_{\widetilde\nu}^{\nu,\eta,i},i}\;\prob\mbox{-a.s.},
\end{equation*}
for all $i\in\mathcal I$ and $t\in[\widetilde\nu,T]$.
Hence, $(Y^{\nu,\eta,i},Z^{\nu,\eta,i},\widehat K^{\nu,\eta,i} = K^{\nu,\eta,i} - K^{\nu,\eta,i}_{\widetilde\nu})$ satisfies the following multidimensional RBSDE on $[\widetilde\nu,T]$.
\begin{align} \label{eq:shift.multidimensional.RBSDE}
-\diff Y_t^{\nu,\eta,i} &= \Big(\psi(t,X_t^{\widetilde\nu,X_{\widetilde\nu}^{\nu,\eta,i},i},i) - \rho(t,X_t^{\widetilde\nu,X_{\widetilde\nu}^{\nu,\eta,i},i},i)Y_t^{\nu,\eta,i} \nonumber\\
&\qquad - \varsigma(t,X_{\widetilde\nu}^{\nu,\eta,i},i,Z_t^{\nu,\eta,i,n}) \Big)\diff t  - (Z_t^{\nu,\eta,i})^\prime\diff W_t + \diff\widehat K_t^{\nu,\eta,i},\; t\in[\widetilde\nu,T],\nonumber \\
Y_T^{\nu,\eta,i} &= g(X_T^{\widetilde\nu,X_{\widetilde\nu}^{\nu,\eta,i},i},i),\quad \widehat K^{\nu,\eta,i}_{\widetilde\nu} = 0, \nonumber \\
Y_t^{\nu,\eta,i} &\geq \max_{j\in\mathcal I\setminus\{i\}}\{Y_t^{\nu,\eta,j} - c_{i,j}(t,X_t^{\widetilde\nu,X_{\widetilde\nu}^{\nu,\eta,i},i})\},\quad t\in[\widetilde\nu,T],\\
&\int_{\widetilde\nu}^T\Big(Y_t^{\nu,\eta,i} - \max_{j\in\mathcal I\setminus\{i\}}\{Y_t^{\nu,\eta,j} - c_{i,j}(t,X_t^{\widetilde\nu,X_{\widetilde\nu}^{\nu,\eta,i},i})\}\Big)\diff\widehat K^{\nu,\eta,i}_t = 0, \nonumber \\
&(Y^{\nu,\eta,i},Z^{\nu,\eta,i},\widehat K^{\nu,\eta,i})\in \SSd{\widetilde\nu}\times\HH{\widetilde\nu}{d}\times\KK{\widetilde\nu}, \quad i\in\mathcal I. \nonumber
\end{align}
Since for each $i$, the multidimensional RBSDE \cref{eq:shift.multidimensional.RBSDE} is the same as the multidimensional RBSDE \cref{eq:multidimensional.RBSDE} starting from $(\widetilde\nu,X_{\widetilde\nu}^{\nu,\eta,i},i)$,
it holds that $Y_t^{\nu,\eta,i} \geq Y_t^{\widetilde\nu,X_{\widetilde\nu}^{\nu,\eta,i},i}\;\prob\mbox{-a.s.}$ for all $i\in\mathcal I$ and $t\in[\widetilde\nu,T]$ because of the minimality of $Y^{\widetilde\nu,X_{\widetilde\nu}^{\nu,\eta,i},i}$ (see \cref{Remark:minimality}).

On the other hand, recursively applying the comparison theorem to the Picard's iterations of $Y_t^{\nu,\eta,i}$ constructed in \cref{Thm:existence} on $[\widetilde\nu,T]$ leads to that
\[
Y_t^{\nu,\eta,i,n} \leq Y_t^{\widetilde\nu,X_{\widetilde\nu}^{\nu,\eta,i},i}\;\prob\mbox{-a.s.},
\]
for all $n\geq 0,\;i\in\mathcal I$ and $t\in[\widetilde\nu,T]$.
Taking a limit of the above inequality, we obtain that $Y_t^{\nu,\eta,i} \leq Y_t^{\widetilde\nu,X_{\widetilde\nu}^{\nu,\eta,i},i}$ for all $i\in\mathcal I$ and $t\in[\widetilde\nu,T]$.
Hence, the equality \cref{eq:minimal.uniqueness} holds.
\end{proof}

\section{Verification and Viscosity Solutions}\label{sec:verification.viscosity}

In this section, we show that the minimum solution in \cref{Thm:existence} can be interpreted as the value function of the optimal switching problem under ambiguity.
\Cref{Prop:excess.Y} provides a verification of $Y$.
The proof of \cref{Prop:excess.Y} is standard, so we put it on \cref{subsec:verification.Y}.

\begin{proposition}\label{Prop:excess.Y}
Suppose \cref{Assump:X.SDE,Assump:rect,Assump:funcs,Assump:varsigma}.
\begin{enumerate}
\item For an arbitrary $(\nu,\eta,\iota)\in \widetilde{\mathcal K_T^{2q}}$, let $Y^{\nu,\eta,\iota}$ be a minimum solution of the multidimensional RBSDE \cref{eq:multidimensional.RBSDE}.
      Then,
      \begin{equation*}
       Y^{\nu,\eta,\iota}_t \geq Y^{\nu,\eta,\iota,\alpha}_t,\quad \forall t\in[\nu,T],
      \end{equation*}
      for all $\alpha=(\tau_k,i_k)_{k\geq 0}\in\control{\nu}{\iota}$.
\item Let $\alpha^*=(\tau^*_k,i^*_k)_{k\geq 0}$ be a control such that $(\tau^*,i^*_0) = (\nu,\iota)$ and that for all $n\geq 1$,
      \begin{align*}
       \tau^*_n &:= \inf\left\{s\in[\tau^*_{n-1},T]\;\Big|\;Y_s^{\tau_{n-1}^*,X^*_{\tau^*_{n-1}},i^*_{n-1}} 
       = \max_{j\in\mathcal I\setminus\{i^*_{n-1}\}}\{Y_s^{\tau_{n-1}^*,X^*_{\tau^*_{n-1}},j} - c_{i^*_{n-1},j}(s,X^*_s)\}\right\}, \\
       i^*_n &\;\mbox{is such that }Y_{\tau^*_n}^{\tau_{n-1}^*,X^*_{\tau^*_{n-1}},i^*_{n-1}} = Y_{\tau^*_n}^{\tau_{n-1}^*,X^*_{\tau^*_{n-1}},i^*_n} - c_{i^*_{n-1},i^*_n}(\tau^*_n,X^*_{\tau^*_n}),
      \end{align*}
      where $X^* = X^{\nu,\eta,\iota,\alpha^*}$.
      Then, $\alpha^*$ is an admissible control and
      \begin{equation*}
       Y^{\nu,\eta,\iota}_t = Y^{\nu,\eta,\iota,\alpha^*}_t,\quad \forall t\in[\nu,T].
      \end{equation*}
\end{enumerate}
\end{proposition}

By \cref{Prop:excess.Y}, we obtain
\begin{equation*}
Y_t^{t,x,i} = \sup_{\alpha\in\control{t}{i}}Y_t^{t,x,i,\alpha} = \sup_{\alpha\in\control{t}{i}}J(t,x,i,\alpha),
\end{equation*}
for all $(t,x,i)\in \overline{\mathcal K_T}$.
Hence, $Y_t^{t,x,i}$ is the value function of the optimal switching problem under ambiguity.
Furthermore, $\alpha^*$ defined in \cref{Prop:excess.Y}.2 is an optimal control of the problem.

We next study a relationship between the multidimensional RBSDE \cref{eq:multidimensional.RBSDE} and partial differential equations (hereafter PDEs).
Let $u:[0,T]\times\real^d\times\mathcal I\rightarrow\real$ be a function.
Consider the following PDE,

\begin{align}\label{eq:PDE}
&\min\{-u_t(t,x,i) - \ig^i u(t,x,i) - \psi(t,x,i) + \rho(t,x,i)u(t,x,i) \nonumber \\
&\qquad\qquad + \varsigma(t,x,i,\sigma^\prime(t,x,i)\nabla u(t,x,i)), \\
&\qquad\qquad\qquad\qquad u(t,x,i) - \max_{j\in\mathcal I\setminus\{i\}}\{u(t,x,j) - c_{i,j}(t,x)\}\} = 0,\quad (t,x,i)\in\overline{\mathcal K_T}, \nonumber\\
&u(T,x,i) = g(x,i),\nonumber 
\end{align}
where $u_t(t,x,i) = \frac{\p u(t,x,i)}{\p t},\; \nabla u(t,x,i) = \frac{\p u(t,x,i)}{\p x}$ and
\[
\ig^i f(t,x) = (\nabla f(t,x))^\prime b(t,x,i) + \frac{1}{2}\mathrm{tr}\left(\sigma\sigma^\prime(t,x,i)\frac{\p f(t,x)}{\p x \p x^\prime}\right).
\]

If the PDE \cref{eq:PDE} has a classical solution, then we can easily show that this solution is a value function of the optimal switching problem under ambiguity.
However, the classical solution does not always exist.
We shall consider a more general concept of solutions, i.e., a viscosity solution.
Let $C^{1,2}([0,T)\times \real^d\times \mathcal I)$ be a set of functions that are continuously differentiable with respect to $t$ and twice continuously differentiable with respect $x$ on $[0,T)\times \real^d\times \mathcal I$.

\begin{definition}[Viscosity solution]\label{Def:viscosity.solution}
\
\begin{enumerate}
 \item \textit{Viscosity supersolution}.\\ A lower semi-continuous function $(u(\cdot,\cdot,1),\dots,u(\cdot,\cdot,I))$ is a viscosity supersolution of the PDE \cref{eq:PDE} if for any $(t,x,i)\in[0,T)\times\real^d\times\mathcal I$ and any $\varphi\in C^{1,2}([0,T)\times \real^d\times \mathcal I)$ such that $v(\cdot,\cdot,i) - \varphi(\cdot,\cdot,i)$ attains a local minimum at $(t,x)$ for all $i\in\mathcal I$,
 \begin{align*}
  &\min\{-\varphi_t(t,x,i) - \mathcal L^i\varphi(t,x,i)  - \psi(t,x,i) + \rho(t,x,i)u(t,x,i) + \varsigma(t,x,i,\sigma^\prime(t,x,i)\nabla \varphi(t,x,i)),\\
  &\qquad\qquad u(t,x,i) - \max_{j\in\mathcal I\setminus\{i\}}\{u(t,x,j) - c_{i,j}(t,x)\}\} \geq 0, \\
  &u(T,x,i) \geq g(x,i).
 \end{align*}
 \item \textit{Viscosity subsolution}.\\ A upper semi-continuous function $(u(\cdot,\cdot,1),\dots,u(\cdot,\cdot,I))$ is a viscosity subsolution of the PDE \cref{eq:PDE} if for any $(t,x,i)\in[0,T)\times\real^d\times\mathcal I$ and any $\varphi\in C^{1,2}([0,T)\times \real^d\times \mathcal I)$ such that $v(\cdot,\cdot,i) - \varphi(\cdot,\cdot,i)$ attains a local maximum at $(t,x)$ for all $i\in\mathcal I$,
 \begin{align*}
  &\min\{-\varphi_t(t,x,i) - \mathcal L^i\varphi(t,x,i)   - \psi(t,x,i) + \rho(t,x,i)u(t,x,i) + \varsigma(t,x,i,\sigma^\prime(t,x,i)\nabla \varphi(t,x,i)),\\
  &\qquad\qquad u(t,x,i) - \max_{j\in\mathcal I\setminus\{i\}}\{u(t,x,j) - c_{i,j}(t,x)\}\} \leq 0, \\
  &u(T,x,i) \leq g(x,i).
 \end{align*}
 \item \textit{Viscosity solution}.\\ A locally bounded function $(u(\cdot,\cdot,1),\dots,u(\cdot,\cdot,I))$ is a viscosity solution of the PDE \cref{eq:PDE} if its lower semi-continuous envelope is a viscosity supersolution of the PDE \cref{eq:PDE}, and if its upper semi-continuous envelope is a viscosity subsolution of the PDE \cref{eq:PDE}.
\end{enumerate}
\end{definition}

For details of the viscosity solutions, we refer to \cite{crandall1992user}. 
We define a set of functions $\mathcal{CP}([0,T]\times\real^d)$ as follows.
\begin{equation*}
\mathcal{CP}([0,T]\times\real^d)
:=\left\{f:[0,T]\times\real^d\rightarrow\real\;\Big|\;
\begin{array}{c}
f \mbox{ is jointly continuous and} \\
\mbox{there exist positive constants $C$ and $q$}\\
\mbox{such that }|f(t,x)| \leq C(1 + \|x\|^q),\\
\mbox{for all }(t,x)\in[0,T]\times\real^d.
\end{array}
\right\}.
\end{equation*}
Let
\begin{equation*}
v(t,x,i) := Y_t^{t,x,i},
\end{equation*}
for $(t,x,i)\in\overline{\mathcal K_T}$, where $Y_t^{t,x,i}$ is a minimum solution of the multidimensional RBSDE \cref{eq:multidimensional.RBSDE}.

Now, we will prove that $v$ is a unique viscosity solution of the PDE \cref{eq:PDE} in $\mathcal{CP}([0,T]\times\real^d)$.
In \cite{hamadene2013viscosity}, the viscosity solution of the PDE similar to \cref{eq:PDE} is investigated.
Main differences between our model and the model in \cite{hamadene2013viscosity} are as follows.
\begin{enumerate}
 \item The model in \cite{hamadene2013viscosity} allows for a generator of RBSDE for $Y^i$ to depend on the other $Y^j$, but we consider the case when it does not depend on the other $Y^j$.
 \item The model in \cite{hamadene2013viscosity} assumes that switching costs are non-negative, but we allow negative switching costs.
 \item The model in \cite{hamadene2013viscosity} assumes that a dynamics of the forward variable $X$ does not depend on a control process, but we allow the dynamics of $X$ to depend on the control.
\end{enumerate}
In fact, the results of \cite{hamadene2013viscosity} can be applied to our model.
In \cite{hamadene2013viscosity}, it is shown that there exists a unique viscosity solution without using non-negativity of the switching costs.
Furthermore, the controllability of $X$ does not affect to the results in \cite{hamadene2013viscosity}.
Hence, we can provide the existence and uniqueness of the solution to the PDE \cref{eq:multidimensional.RBSDE} in the viscosity sense and prove that the value function is a unique viscosity solution to \cref{eq:multidimensional.RBSDE}.

\begin{proposition}\label{Prop:visocity.property}
Suppose \cref{Assump:X.SDE,Assump:rect,Assump:funcs,Assump:varsigma}.
Let
\[
\vec v := (v(\cdot,\cdot,1),\dots,v(\cdot,\cdot,I)).
\]
Then, $\vec v$ is a unique viscosity solution to the PDE \cref{eq:PDE} in $(\mathcal{CP}([0,T]\times\real^d))^I$.
\end{proposition}
\begin{proof}[Proof of \cref{Prop:visocity.property}]
Let $(t,x,i)\in\overline{\mathcal K_T}$.
Let $(Y^{t,x,i,n})_{n\geq 0}$ be a sequence of the Picard's iterations defined in \cref{Thm:existence}.
Then, by \cite{el1997backward}, there exists $v_n(\cdot,\cdot,i)\in\mathcal{CP}([0,T]\times\real^d)$ for all $n\geq 0$ and $i\in\mathcal I$ such that
\[
Y_s^{t,x,i,n} = v_n(t,X_s^{t,x,i},i),
\]
for all $s\in[t,T]$.
Furthermore, we define $\overline v \in \mathcal{CP}([0,T]\times\real^d)$ as
\[
\overline v(t,x) := M^{t,x}_t,
\]
where $M^{t,x}$ is defined in \cref{Thm:existence}.
Recall that $Y^{t,x,i,n} \rightarrow Y^{t,x,i}$ in the mean-square sense.
Therefore, $\vec v$ is a lower semi-continuous function and it satisfies the polynomial growth condition with respect to $x$ since $v_0\leq v_n\leq \overline v$ and $v_n\leq v_{n+1}$ for all $n\geq 1$.

On the other hand, Corollary 1 in \cite{hamadene2013viscosity} provides the continuity and uniqueness of a viscosity solution to the PDE \cref{eq:multidimensional.RBSDE}.
Furthermore, by Theorem 1 in \cite{hamadene2013viscosity}, $\vec v$ is a viscosity solution of the PDE \cref{eq:multidimensional.RBSDE}.
Hence, we conclude that $\vec v$ is a unique viscosity solution of the PDE \cref{eq:multidimensional.RBSDE} in $(\mathcal{CP}([0,T]\times\real^d))^I$.
\end{proof}

\section{The Infinite Horizon Problem}\label{sec:infinite.horizon}
In this section, we consider the infinite horizon optimal switching problem under ambiguity.
Let $\controlinf{\nu}{i}$ be a set of admissible controls like \cref{eq:control} but $\tau_k\rightarrow\infty\;\prob$-almost surely.
Furthermore, we assume as follows.
\begin{assumption}\label{Assump:infinite}
\
\begin{enumerate}
\item \textit{Time-homogeneity.} $b,\sigma,\psi,\phi,\varsigma,$ and $c$ do not depend on $t$. There exists a positive constant $\rho$ such that
      \[
       \rho(t,x,i) = \rho > 0,
      \]
      for all $(t,x,i)\in[0,\infty)\times\real^d\times\mathcal I$. $\Theta_t$ only depends on the values of $X_t$ and $\alpha_t$. We denote $\Theta_t$ with $X_t = x\in\real^d$ and $\alpha_t=i\in\mathcal I$ by $\Theta^{x,i}$.
\item \textit{Sufficiently large discount.} $\rho$ is sufficiently large in the following sense. There exist constants $C\geq 0$ and $c_\infty > 0$ such that
      \begin{align}
       \expect\left[e^{-\rho t}\zeta_t^{\theta,0}\|X_t^{x,i,\alpha}\|^q\right] &\leq C(1+\|x\|^q)e^{-c_\infty t}, \label{eq:large.discount.01}\\
       \expect\left[\sup_{s\geq t}e^{-\rho s}\|X_s^{x,i,\alpha}\|^q\right] &\leq C(1+\|x\|^q)e^{-c_\infty t}, \label{eq:large.discount.02}
      \end{align}
      for all $(t,x,i)\in[0,\infty)\times\real^d\times\mathcal I,\;\theta\in\Theta[0,\infty)$ and $\alpha\in\controlinf{0}{i}$, where $X^{x,i,\alpha}$ is a solution to the SDE \cref{eq:X.SDE} starting at $X_0^{x,i,\alpha}=x$ and controlled by $\alpha\in\controlinf{0}{i}$.
\item \textit{Polynomial growth conditions.} $\psi,\phi$ and $c$ are continuous and satisfy the polynomial growth condition in \cref{Assump:funcs}.2.
\item \textit{Non-negative reward condition.}
      \begin{equation}\label{eq:psi.varsigma.condition}
       \psi(x,i) - \varsigma(x,i,0) \geq 0,
      \end{equation}
      for all $(x,i)\in\real^d\times\mathcal I$.
\item \textit{Temporary terminal condition.} There exist polynomial growth functions $g(x,1),$ $\dots,$ $g(x,I)$ such that
      \
      \begin{enumerate}
      \item \begin{equation}\label{eq:terminal.non.positive.infinite}
             g(x,i)\leq 0,
            \end{equation}
            for all $i\in\mathcal I$ and $x\in\real^d$;
      \item \begin{equation}\label{eq:terminal.freeloop.infinite}
             g(x,i) \geq \max_{j\in\mathcal I\setminus\{i\}}\{g(x,j) - c_{i,j}(x)\},
            \end{equation}
            for all $i\in\mathcal I$ and $x\in\real^d$;
      \item \begin{equation}\label{eq:terminal.submartingale.infinite}
             \inf_{\theta_t\in\Theta[T,\widetilde T)}\expect\left[e^{-\rho \widetilde T}\zeta_{\widetilde T}^{\theta,T}g(X_{\widetilde T}^{\nu,\eta,i},i)\;\Big|\;\filt_T\right] \geq e^{-\rho T}g(X_T^{\nu,\eta,i},i),
            \end{equation}
            for all $0\leq T\leq \widetilde T,\;\nu\in\mathcal T_0^T,\;\eta\in L_\nu^{2q}(\real^d)$ and $i\in\mathcal I$.
      \end{enumerate}
\item \textit{Non-free loop condition in the infinite horizon.} For all finite loops $(i_0,i_1,\dots,i_m)\in\mathcal I^{m+1}$ with $i_0 = i_m$ and $i_0\neq i_1$ and for all $x\in\real^d$, $c$ satisfies
      \begin{equation*}
             c_{i_0,i_1}(x) + \cdots + c_{i_{m-1},i_m}(x) > 0.
      \end{equation*}
\item \textit{Strong triangular condition in the infinite horizon.}
      \begin{equation*}
       c_{k,j}(x) \leq c_{k,i}(x) - C_i(1 + C_{q,X}^\infty(1 + \|x\|^q)),
      \end{equation*}
      for all $i\in\mathcal N,\;(j,k)\in\mathcal I$ and $x\in\real^d$ with $j\neq i$ and $k\neq j$, where $C_i,C_{q,X}^\infty$ and $q$ are defined in \cref{Prop:q.th.integrable,Assump:funcs}.
\end{enumerate}
\end{assumption}

The time-homogeneity (\cref{Assump:infinite}.1) is a standard condition.
With taking account of the time-homogeneity and the Markov property of $X$, the starting time does not matter to the optimal switching problem.
The sufficiently large discount condition (\cref{Assump:infinite}.2) is also standard.
If it is not postulated, then the value function can diverge.
Therefore, we need this condition to consider meaningful problems.
However, the condition \cref{eq:large.discount.01} is slightly strong.
Indeed, it is sufficient to satisfy \cref{eq:large.discount.01} with $\theta = 0$ and \cref{eq:large.discount.02} in order to prove the finiteness of the value function (\cref{Prop:value.func.poly.growth}).
The condition \cref{eq:large.discount.01} is needed to prove the convergent property of the value function from the finite horizon to the infinite horizon (\cref{Prop:convergence.verification}).

Under the non-negative reward condition (\cref{Assump:infinite}.4), the rewards of the optimal switching problem in the infinite horizon is non-negative.
Indeed, by the definition of $\varsigma$, we have
\[
 \psi(X^{x,i,\alpha}_t,\alpha_t) - \theta_t^\prime\phi(X_t^{x,i,\alpha},\alpha_t) 
 \geq \psi(X^{x,i,\alpha}_t,\alpha_t) - \varsigma(X^{x,i,\alpha}_t,\alpha_t,0)
 \geq 0,
\]
for all $(t,x,i)\in[0,\infty)\times\real^d\times\mathcal I,\;\theta_t\in\Theta_t$ and $\alpha\in\controlinf{0}{i}$.
The non-negative reward condition guarantees that an optimal switching problem in a longer finite horizon has a large value function.
This restriction is needed to exchange the orders of taking limits of Picard's iterations $n$ and time horizons $T$.
This is slightly restrictive, however, it can be replaced to a lower bounded condition (\cref{Remark:lower.bounded.reward}).

The temporary terminal conditions (\cref{Assump:infinite}.5) are assumed for purely technical reasons.
However, they are not so restrictive.
If all switching costs are non negative, then we can choose $g(x,i) = 0$ for all $(x,i)\in\real^d\times\mathcal I$ satisfying all the temporary terminal conditions.
Once we find the constants $g_1,\dots,g_I$ satisfying the inequality \cref{eq:terminal.freeloop.infinite},
then $g_1-\max_{j\in\mathcal I}g_j,\dots,g_I-\max_{j\in\mathcal I}g_j$ satisfy all the temporary terminal conditions.
If $g(x,i)$ satisfies the inequalities \cref{eq:terminal.non.positive.infinite,eq:terminal.freeloop.infinite} and if $g(\cdot,i)$ is twice continuously differentiable for all $i\in\mathcal I$,
then one of sufficient conditions to satisfy the inequality \cref{eq:terminal.submartingale.infinite} is
\begin{equation}\label{eq:sufficient.condition.terminal.submartingale.infinite}
\ig^i g(x,i) - \rho g(x,i) - (\nabla g(x,i))^\prime\sigma(x,i)\theta \geq 0,
\end{equation}
for all $(x,i)\in\real^d\times\mathcal I$ and $\theta\in\Theta^{x,i}$.
The condition \cref{eq:sufficient.condition.terminal.submartingale.infinite} can be derived by applying the Ito's lemma to $e^{-\rho t}\zeta^\theta_tg(X_t,i)$.
If the switching costs are constants, we can easily find the constants satisfying the temporary terminal conditions.
On the other hand, in the major applications such as the buy low and sell high problem and the pair-trading problem, we can also find the functions satisfying the temporary terminal conditions.
The other assumptions are essentially the same as the assumptions in the finite horizon problem.

The objective function in the infinite horizon is
\begin{multline*}
J(x,i,\alpha) = \inf_{\theta\in\Theta[0,\infty)}\expect\biggl[\int_0^\infty e^{-\rho t}\zeta_t^\theta\Big(\psi(X_t^{x,i,\alpha},\alpha_t) - \theta_t^\prime\phi(X_t^{x,i,\alpha},\alpha_t)\Big)\diff t \\
- \sum_{k=1}^\infty e^{-\rho\tau_k}\zeta_{\tau_k}^\theta c_{i_{k-1},i_k}(X_{\tau_k}^{x,i,\alpha})\biggr],
\end{multline*}
for $(x,i)\in\real^d\times\mathcal I$ and $\alpha\in\controlinf{0}{i}$.
The optimal switching problem under ambiguity in the infinite horizon is
\begin{equation}\label{eq:infinite.problem}
v^\infty(x,i) := \sup_{\alpha\in\controlinf{0}{i}}J(x,i,\alpha),
\end{equation}
for $(x,i)\in\real^d\times\mathcal I$.
We can easily show that $v^\infty$ is polynomial growth with respect to $x$.
\begin{proposition}\label{Prop:value.func.poly.growth}
Under \cref{Assump:X.SDE,Assump:infinite},
there exists a positive constant $C$ such that
\[
0\leq v^\infty(x,i)\leq C(1 + \|x\|^q),
\]
for all $x\in\real^d$ and $i\in\mathcal I$.
Thus, $v^\infty$ is polynomial growth with respect to $x$.
\end{proposition}
\begin{proof}[Proof of \cref{Prop:value.func.poly.growth}]
It is clear that $v^\infty$ is non-negative by the non-negative reward condition.
Fix an arbitrary $x\in\real^d$ and $i\in\mathcal I$.
Then, by the polynomial growth condition of $\psi$ and $c$ and the strong triangular condition, we have
\begin{align*}
J(x,i,\alpha)
&\leq \expect\left[\int_0^\infty e^{-\rho t}\psi(X_t^{x,i,\alpha},\alpha_t)\diff t - \sum_{k=1}^\infty e^{-\rho\tau_k} c_{i_{k-1},i_k}(X_{\tau_k}^{x,i,\alpha})\right] \leq C(1 + \|x\|^q),
\end{align*}
for all $\alpha\in\controlinf{0}{i}$, where $C$ is a positive constant not depending on $x,i$ and $\alpha$.
Hence, we obtain the desired result.
\end{proof}

\begin{remark}\label{Remark:lower.bounded.reward}
\Cref{Assump:infinite}.6 (the inequality \cref{eq:psi.varsigma.condition}) can be replaced to a lower bounded condition. We assume that there exists some constant $c_{\psi,\varsigma}$ such that
\[
\psi(x,i) - \varsigma(x,i,0) \geq c_{\psi,\varsigma},
\]
for all $(x,i)\in\real^d\times\mathcal I$.
Then,
\begin{align*}
 J(x,i,\alpha) - \frac{c_{\psi,\varsigma}}{\rho}
 &= J(x,i,\alpha) - \int_0^\infty e^{-\rho t}c_{\psi,\varsigma}\diff t \\
 &= \inf_{\theta\in\Theta[0,\infty)}\expect\left[\int_0^\infty e^{-\rho t}\zeta_t^\theta\Big(\psi(X_t^{x,i,\alpha},\alpha_t) - \theta_t^\prime\phi(X_t^{x,i,\alpha},\alpha_t) - c_{\psi,\varsigma}\Big)\diff t \right.\\
&\qquad\qquad\qquad \left.- \sum_{k=1}^\infty e^{-\rho\tau_k}\zeta_{\tau_k}^\theta c_{i_{k-1},i_k}(X_{\tau_k}^{x,i,\alpha})\right],
\end{align*}
for all $(x,i)\in\real^d\times\mathcal I$ and $\alpha\in\controlinf{0}{i}$.
By the definition $\varsigma$, we have
\begin{align*}
 \psi(X^{x,i,\alpha}_t,\alpha_t) - \theta_t^\prime\phi(X_t^{x,i,\alpha},\alpha_t) - c_{\psi,\varsigma} 
 \geq \psi(X^{x,i,\alpha}_t,\alpha_t) - \varsigma(X^{x,i,\alpha}_t,\alpha_t,0) - c_{\psi,\varsigma} 
 \geq 0,
\end{align*}
for all $(t,x,i)\in[0,\infty)\times\real^d\times\mathcal I,\;\theta_t\in\Theta_t$ and $\alpha\in\controlinf{0}{i}$.
Hence, we can replace the original rewards to non-negative rewards.
$c_{\psi,\varsigma}$ may be negative, but it is finite.
\end{remark}

\begin{remark}\label{Remark:cost.upper.bounded.infinite}
Similarly to \cref{Remark:cost.upper.bounded}, the strong triangular condition in the infinite horizon is not necessarily needed.
Instead of the strong triangular condition, it is sufficient to hold the following inequality
\[
\expect\left[-\sum_{k=1}^\infty e^{-\rho \tau_k}c_{i_{k-1},i_k}(X_{\tau_k}^{x,i,\alpha})\right] \leq C(1 + \|x\|^q),
\]
for all $x\in\real^d,\;i\in\mathcal I$ and $\alpha\in\controlinf{0}{i}$, where $C$ is a positive constant not depending on $(x,i)$ and $\alpha$.
Furthermore, under the above inequality, we do not also need the inequality \cref{eq:large.discount.02}.
\end{remark}

We consider the following multidimensional RBSDE on $[\nu,T]$ for $\nu\in \mathcal T_0^T$ and $\eta\in L_\nu^{2q}(\real^d)$,
\begin{align}\label{eq:convergence.RBSDE}
-\diff \widehat Y^{T,\nu,\eta,i}_t &= \Big(\psi(X_t^{\nu,\eta,i},i) - \rho \widehat Y^{T,\nu,\eta,i}_t - \varsigma(X_t^{\nu,\eta,i},i,\widehat Z_t^{T,\nu,\eta,i})\Big)\diff t \nonumber\\
&\qquad\qquad- (\widehat Z_t^{T,\nu,\eta,i})^\prime \diff W_t + \diff \widehat K^{T,\nu,\eta,i},\; t\in[\nu,T],\nonumber\\
\widehat Y^{T,\nu,\eta,i}_T &= g(X_T^{\nu,\eta,i},i),\quad \widehat K^{T,\nu,\eta,i}_\nu = 0,\nonumber \\
\widehat Y^{T,\nu,\eta,i}_t &\geq \max_{j\in\mathcal I\setminus\{i\}}\left\{\widehat Y^{T,\nu,\eta,j}_t - c_{i,j}(X_t^{\nu,\eta,i})\right\},\; t\in[\nu,T], \\
&\int_0^T\Big(\widehat Y^{T,\nu,\eta,i}_t - \max_{j\in\mathcal I\setminus\{i\}}\left\{\widehat Y^{T,\nu,\eta,j}_t - c_{i,j}(X_t^{\nu,\eta,i})\right\}\Big)\diff t = 0,\nonumber \\
&(\widehat Y^{T,\nu,\eta,i},\widehat Z^{T,\nu,\eta,i},\widehat K^{T,\nu,\eta,i}) \in \SSd{\nu}\times\HH{\nu}{d}\times\KK{\nu}, \quad i \in\mathcal I, \nonumber
\end{align}
where $g$ is a function satisfying the temporary terminal conditions.
By \cref{Thm:existence} and \cref{Prop:minimal.uniqueness},
there exists a unique minimum solution of the multidimensional RBSDE \cref{eq:convergence.RBSDE}.
Now, we show that the solution to the multidimensional RBSDE \cref{eq:convergence.RBSDE} converges to the value function \cref{eq:infinite.problem} as $T\rightarrow\infty$.

\begin{proposition}\label{Prop:convergence.verification}
Under \cref{Assump:X.SDE,Assump:rect,Assump:varsigma,Assump:infinite},
$\widehat Y_t^{T,\nu,\eta,\iota} \leq \widehat Y_t^{\widetilde T,\nu,\eta,\iota}$ for all $\nu\in\mathcal T_0^T,\;\nu\leq t\leq T\leq \widetilde T,\;\eta\in L_\nu^{2q}(\real^d)$ and $\iota\in\widetilde{\mathcal I}_\nu$.
Furthermore, for all $(t,x,i)\in[0,\infty)\times\real^d\times\mathcal I$,
\begin{equation}\label{eq:convergence.verification}
\lim_{T\rightarrow\infty}\widehat Y^{T,t,x,i}_t = v^\infty (x,i).
\end{equation}
Finally, $v^\infty(\cdot,i)$ is continuous for all $i\in\mathcal I$.
\end{proposition}
Since the proof of \cref{Prop:convergence.verification} is too long, we put it on \cref{subsec:verification.infinite.horizon}.

We next study the relationships between $v^\infty$ and PDE.
Consider the following PDE.
\begin{multline}\label{eq:PDE.infinite}
\min\{-\ig^i u(x,i) - \psi(x,i) + \rho u(x,i) + \varsigma(x,i,\sigma^\prime(x,i)\nabla u(x,i)),\\
 u(x,i) - \max_{j\in\mathcal I\setminus\{i\}}\{u(x,j) - c_{i,j}(x)\}\} = 0,\quad (x,i)\in\real^d\times\mathcal I,
\end{multline}
where
\[
\ig^i f(x) = (\nabla f(x))^\prime b(x,i) + \frac{1}{2}\mathrm{tr}\left(\sigma\sigma^\prime(x,i)\frac{\p f(x)}{\p x\p x^\prime}\right).
\]
Then, the following proposition holds.

\begin{proposition}\label{Prop:uniform.viscosity}
Under \cref{Assump:X.SDE,Assump:rect,Assump:varsigma,Assump:infinite},
$v^\infty$ is a viscosity solution of the PDE \cref{eq:PDE.infinite}.
\end{proposition}
The proof of \cref{Prop:uniform.viscosity} is in \cref{subsec:viscosity.infinite}.
By \cref{Prop:uniform.viscosity}, we can study the optimal switching problem under ambiguity through the PDE \cref{eq:PDE.infinite}.
Moreover, we can easily show the uniqueness of the solution to the PDE \cref{eq:PDE.infinite} using the method of Proposition 3.1 in \cite{hamadene2013viscosity}, so we omit the proof of the uniqueness.

\section{Financial Applications}\label{sec:financial.appliciation}
\subsection{Monotone Conditions}\label{subsec:monotone.condition}
We first prove that under certain conditions, the optimal switching problem under ambiguity can be interpreted as an optimal switching problem with a shift of the drift of $X$ not depending on its value function.
We first assume the followings.
\begin{assumption}{\textit{Monotone conditions.}}\label{Assump:monotone}
We assume $d=1$.
\begin{enumerate}
 \item \textit{$\kappa$-ignorance.} There exist non-negative constants $\kappa_1,\dots,\kappa_I$ such that
 \[
  \Theta_t^{x,i} = [-\kappa_i,\kappa_i],
 \]
 for all $i\in\mathcal I,\;x\in\real^d$ and $t\in[0,\infty)$.
 \item For every $x,y\in\real$, $X$ satisfies,
 \begin{equation*}
  x \leq y \quad\Rightarrow\quad X_s^{t,x,i} \leq X_s^{t,y,i},\;\prob\mbox{-a.s.},
 \end{equation*}
 for all $t,s\in[0,T],\;i\in\mathcal I$ with $t\leq s$.
 \item $\rho$ does not depend on a value of $x$.
 \item For every $(t,i)\in[0,T]\times\mathcal I$, $\psi(t,\cdot,i)$ is non-decreasing.
 \item $\phi(t,x,i) = 0$ for every $(t,x,i)\in\overline{\mathcal K_T}$.
 \item For every $i\in\mathcal I$, $g(\cdot,i)$ is non-decreasing.
 \item For every $(t,i,j)\in[0,T]\times(\mathcal I)^2$, $c_{i,j}(t,\cdot)$ is non-increasing.
\end{enumerate}
\end{assumption}

In \cite{chen2002ambiguity}, \cref{Assump:monotone}.1 is called $\kappa$-ignorance.
The other conditions guarantee the monotonicity of the value function with respect to the initial value of $X$.
Under \cref{Assump:monotone}, we can prove the following result.

\begin{proposition}\label{Prop:monotone.conditions}
Suppose \cref{Assump:X.SDE,Assump:rect,Assump:funcs,Assump:varsigma,Assump:monotone}.
For all $(t,x,i)\in\overline{\mathcal K_T}$ and $\alpha\in\control{t}{i}$, let ${}^{-\kappa} X^{t,x,i,\alpha}$ be a solution to the following SDE,
\begin{align*}
\diff {}^{-\kappa} X^{t,x,i,\alpha}_s &= \Big(b(s,{}^{-\kappa} X^{t,x,i,\alpha}_s,\alpha_s) - \kappa_{\alpha_s}|\sigma(s,{}^{-\kappa} X^{t,x,i,\alpha}_s,\alpha_s)|\Big) \diff s 
 + \sigma(s,{}^{-\kappa} X^{t,x,i,\alpha}_s,\alpha_s)\diff W_s,\\
&{}^{-\kappa} X^{t,x,i,\alpha}_t = x.
\end{align*}
Then, the value function $v(t,x,i)$ satisfies
\begin{multline}\label{eq:kappa.minus.switching}
v(t,x,i) = \sup_{\alpha\in\control{t}{i}}\expect\biggl[\int_t^T{}^{-\kappa} D_s^{t,i,\alpha}\psi(s,{}^{-\kappa} X_s^{t,x,i,\alpha},\alpha_s)\diff s 
 + {}^{-\kappa} D_T^{t,i,\alpha}g({}^{-\kappa} X_T^{t,x,i,\alpha},\alpha_T)\\
 - \sum_{t\leq \tau_k\leq T}{}^{-\kappa} D_{\tau_k}^{t,i,\alpha}c_{i_{k-1},i_{k}}(\tau_k,{}^{-\kappa} X_{\tau_k}^{t,x,i,\alpha})\;\Big|\;\filt_t\biggr],
\end{multline}
where
\[
{}^{-\kappa} D_s^{t,i,\alpha} = \exp\left\{-\int_t^s\rho(u,\alpha_u) \diff u\right\},\quad s\in[t,T].
\]
Furthermore, $x\rightarrow v(t,x,i)$ is non-decreasing for all $(t,i)\in[0,T]\times\mathcal I$.
\end{proposition}
\begin{proof}[Proof of \cref{Prop:monotone.conditions}]
By the $\kappa$-ignorance and $\phi=0$, we have
\[
\varsigma(t,x,i,z) = \kappa_i|z|,
\]
for all $(t,x,i,z)\in\overline{\mathcal K_T}\times\real$.
Now, fix an arbitrary $t\in[0,T]$ and $x,\widetilde x\in\real$ with $x\leq \widetilde x$.
Then, by the monotone conditions 2-6, we have
\begin{align}
\psi(s,X_s^{t,x,i},i) - \rho(s,i)y -\kappa_i|z| 
 &\leq \psi(s,X_s^{t,\widetilde x,i},i) - \rho(s,i)y -\kappa_i|z|,\label{eq:monotone.ineq.01} \\
g(X_T^{t,x,i},i) &\leq g(X_T^{t,\widetilde x,i},i), \label{eq:monotone.ineq.02}
\end{align}
for all $(s,i,y,z)\in[t,T]\times\mathcal I\times\real\times\real$.
Furthermore, by the monotone conditions 2 and 7, we have
\begin{equation}\label{eq:monotone.ineq.03}
\max_{j\in\mathcal I\setminus\{i\}}\Big\{y^j - c_{i,j}(s,X_s^{t,x,i})\Big\} \leq \max_{j\in\mathcal I\setminus\{i\}}\Big\{\overline y^j - c_{i,j}(s,X_s^{t,\widetilde x,i})\Big\},
\end{equation}
for all $(s,i)\in[t,T]\times\mathcal I$ and $(y^1,\dots,y^I),(\overline y^1,\dots,\overline y^I)\in\real^I$ with $y^k\leq\overline y^k$ for all $k\in\mathcal I$.
Let $(Y^{t,x,i,n})_{i\in\mathcal I,\;n\geq 0}$ and $(Y^{t,\widetilde x,i,n})_{i\in\mathcal I,\;n\geq 0}$ be the Picard's iterations defined in \cref{Thm:existence} with starting $x$ and $\widetilde x$, respectively.
Then, by the inequalities \cref{eq:monotone.ineq.01,eq:monotone.ineq.02,eq:monotone.ineq.03},
recursively applying the comparison theorem leads to that
\[
Y^{t,x,i,n}_s \leq Y^{t,\widetilde x,i,n}_s,
\]
for all $i\in\mathcal I,\;s\in[t,T]$ and $n\geq 0$.
Taking a limit of the above inequality, we have
\begin{equation}\label{eq:monotone.ineq.04}
v(t,x,i) = Y_t^{t,x,i} \leq Y_t^{t,\widetilde x,i} = v(t,\widetilde x,i),
\end{equation}
for all $i\in\mathcal I$.
Since we arbitrarily choose $t,x$ and $\widetilde x$ with $x\leq \widetilde x$, the inequality \cref{eq:monotone.ineq.04} implies that a mapping $x\rightarrow v(t,x,i)$ is non-decreasing for all $t\in[0,T]$ and $i\in\mathcal I$.

Now, let us consider the following PDE,
\begin{align}\label{eq:monotone.PDE}
&\min\{-w_t(t,x,i) - \ig^{-\kappa,i} w(t,x,i) - \psi(t,x,i) + \rho(t,i)w(t,x,i), \nonumber\\
&\qquad\qquad w(t,x,i) - \max_{j\in\mathcal I\setminus\{i\}}\{w(t,x,j) - c_{i,j}(t,x)\}\} = 0,\quad (t,x,i)\in\overline{\mathcal K_T}, \\
&w(T,x,i) = g(x,i),\nonumber
\end{align}
where
\begin{align*}
\ig^{-\kappa,i} f(t,x) = (b(t,x,i) - \kappa_i|\sigma(t,x,i)|)\nabla f(t,x) + \frac{1}{2}(\sigma(t,x,i))^2\frac{\p^2 f(t,x)}{\p x^2}.
\end{align*}
The PDE \cref{eq:monotone.PDE} has a unique continuous viscosity solution, denoted by $w$.
Let $(t,x)\in[0,T)\times\real$ and let $\varphi\in C^{1,2}([0,T)\times\real\times\mathcal I)$ be a test function such that $v(\cdot,\cdot,i) - \varphi(\cdot,\cdot,i)$ attains a local minimum at $(t,x)$ for all $i\in\mathcal I$.
Since $y\rightarrow v(s,y,j)$ is monotone non-decreasing for all $(s,j)\in[0,T)\times\mathcal I$, we have $\nabla \varphi(t,x,i)\geq 0$ for all $i\in\mathcal I$.
Since $v$ is the viscosity supersolution to the PDE \cref{eq:PDE} by \cref{Prop:visocity.property},
we have
\begin{align*}
&\min\{-\varphi_t(t,x,i) - \ig^{-\kappa,i} \varphi(t,x,i) - \psi(t,x,i) + \rho(t,i)v(t,x,i),\\
&\qquad\qquad v(t,x,i) - \max_{j\in\mathcal I\setminus\{i\}}\{v(t,x,j) - c_{i,j}(t,x)\}\} \\
&= \min\{-\varphi_t(t,x,i) - \ig^{i} \varphi(t,x,i) - \psi(t,x,i) + \rho(t,i)v(t,x,i) + \kappa|\sigma(t,x,i)\nabla \varphi(t,x,i)|,\\
&\qquad\qquad v(t,x,i) - \max_{j\in\mathcal I\setminus\{i\}}\{v(t,x,j) - c_{i,j}(t,x)\}\} \geq 0,
\end{align*}
for all $i\in\mathcal I$.
Hence, $v$ is a viscosity supersolution to the PDE \cref{eq:monotone.PDE}.
The comparison theorem of viscosity solutions gives $v\geq w$.
Using the similar argument, we also have $v\leq w$.
Thus, $v = w$.
Since a value function of the optimal switching problem in the right hand side of our desired equality \cref{eq:kappa.minus.switching} is a unique viscosity solution to the PDE \cref{eq:monotone.PDE}, we obtain the equality \cref{eq:kappa.minus.switching}.
\end{proof}

In the infinite horizon case,
\cref{Prop:monotone.conditions} also holds under the same conditions as \cref{Assump:monotone}.
\Cref{Prop:monotone.conditions} implies that under the monotone conditions, the optimal switching problem under ambiguity can be regarded as usual optimal switching problems.
Thus, we can use existing results in the literature of the optimal switching if the monotone conditions are satisfied.
In fact, under the monotone conditions, it is sufficient to solve the PDE \cref{eq:monotone.PDE} instead of the PDE \cref{eq:PDE} in order to derive the value function.

The monotone conditions and \cref{Prop:monotone.conditions} are very similar to the results in Cheng and Riedel \cite{cheng2013optimal}.
In \cite{cheng2013optimal}, Cheng and Riedel consider the optimal stopping problem under ambiguity and show that if a payoff function $f(t,x)$ is non-decreasing in $x$ and $\kappa$-ignorance is satisfied, then the optimal stopping problem under ambiguity can be regarded as a standard optimal stopping problem in which the drift of $X$ shifts into $b-\kappa|\sigma|$ (Theorem 4.1 in \cite{cheng2013optimal}).
Our result implies that the optimal switching problem under ambiguity holds the same property as the optimal stopping under ambiguity.

In \cref{subsec:selection.funds,subsec:buy.low.sell.high}, we consider two applications of the optimal switching problem under ambiguity in finance.
The first application in \cref{subsec:selection.funds} is a selection of investment funds and it satisfies the monotone conditions.
However, the second application (the buy low and sell high problem) in \cref{subsec:buy.low.sell.high} does not satisfy the monotone conditions and it definitely needs negative switching costs.

\subsection{Selection of Investment Funds}\label{subsec:selection.funds}
In this section, we consider an optimal selection of two investment funds under ambiguity in the infinite horizon.
Let $d=1$ and $\mathcal I = \{1,2\}$.
Assume that $X$ satisfies the following SDE,
\begin{equation}\label{eq:fund.selection.X.SDE}
\diff X_t = b_{\alpha_t} X_t \diff t + \sigma_{\alpha_t}X_t\diff W_t,
\end{equation}
where $b_i\in\real,\;\sigma_i>0,\;i=1,2$ are constants.
The solution to the SDE \cref{eq:fund.selection.X.SDE} is
\begin{align*}
X_t^{x,i,\alpha} &= x\exp\left\{\int_0^t\left(b_{\alpha_s} - \frac{1}{2}\sigma^2_{\alpha_s}\right)\diff s + \int_0^t\sigma_{\alpha_s}\diff W_s\right\},
\end{align*}
for all $\alpha\in\controlinf{0}{i}$.
Assume that $\phi=0$ and that $\psi$ is
\[
\psi(x) = x^p,\quad x\in[0,\infty),\; 0<p<1.
\]
The switching costs $c_{1,2}$ and $c_{2,1}$ are constants over $x$, and they satisfy $c_{1,2} + c_{2,1} > 0$.
The constant discount rate $\rho$, satisfies
\[
\rho > p\max_{i\in\mathcal I}\left\{b_i - \frac{1- p}{2}\sigma_i^2\right\}.
\]
The set of multiple priors is
\[
\Theta^{x,i} = [-\kappa_i,\kappa_i],\quad \kappa_i\geq 0,
\]
for all $x\in\real$ and $i\in\mathcal I$.
In the above settings, an optimal switching problem of interest is
\begin{equation}\label{eq;fund.selection.problem}
v^\infty(x,i) = \sup_{\alpha\in\controlinf{0}{i}}\inf_{\theta\in\Theta[0,\infty)}\expect\left[\int_0^\infty e^{-\rho t}\zeta_t^{\theta,0}(X_t^{x,i,\alpha})^p \diff t - \sum_{k=1}^\infty e^{-\rho\tau_k}\zeta_{\tau_k}^{\theta,0} c_{i_{k-1},i_k}\right].
\end{equation}
Since the problem \cref{eq;fund.selection.problem} satisfies \cref{Assump:X.SDE,Assump:rect,Assump:varsigma,Assump:infinite}, we can use the results in \cref{sec:infinite.horizon}.
Furthermore, the problem \cref{eq;fund.selection.problem} also satisfies the monotone conditions (\cref{Assump:monotone}).

Without ambiguity (i.e., $\kappa_i=0$ for all $i\in\mathcal I$),
the problem \cref{eq;fund.selection.problem} is well studied in \cite{ly2007explicit}.
We shortly summarize the results in \cite{ly2007explicit} as follows.
\begin{proposition}[Theorem 4.1 in \cite{ly2007explicit}]\label{Prop:LyBath.Pham.2007}
Let
\begin{equation}\label{eq:funds.selection.K}
K_i = \frac{1}{\rho - b_ip + \frac{1}{2}\sigma_i^2p(1-p)},
\end{equation}
for all $i\in\mathcal I$.
Let $i,j\in\mathcal I,\;i\neq j$.
\begin{enumerate}
\item If $K_i = K_j$, then it is always optimal to switch from regime $i$ to $j$ if the corresponding switching cost is non positive, and never optimal to switch otherwise.
\item If $K_j > K_i$, then the following switching strategies depending on the switching costs are optimal.
      \begin{enumerate}
       \item $c_{i,j}\leq 0$: it is always optimal to switch from regime $i$ to $j$ if one first stands in $i$, and it is always optimal not to switch from $j$ to $i$ otherwise.
       \item $c_{i,j}> 0$:
             \begin{enumerate}
             \item $c_{j,i}\geq 0$: there exists $\underline x^*_i\in[0,\infty)$ such that if one first stands in regime $i$, then it is optimal to switch from $i$ to $j$ whenever $X$ exceeds $\underline x^*_i$. If one first stands in regime $j$, then it is optimal not to switch from $j$ to $i$.
             \item $c_{j,i}< 0$: there exist $\underline x^*_i,\overline x^*_j\in[0,\infty)$ with $\overline x^*_j < \underline x^*_i$ such that if one first stands in regime $i$, then it is optimal to switch from $i$ to $j$ whenever $X$ exceeds $\underline x^*_i$, and that if one first stands in regime $j$, then it is optimal to switch from $j$ to $i$ whenever $X$ falls below $\overline x^*_j$.
             \end{enumerate}
      \end{enumerate}
\end{enumerate}
\end{proposition}

For details of $\underline x_i^*$ and $\overline x^*_j$ and the functional form of the value function, we refer to \cite{ly2007explicit}.
By \cref{Prop:LyBath.Pham.2007}, the types of the switching strategies are determined by $K_i$ defined in \cref{eq:funds.selection.K} and the switching costs.
The most interesting case is \cref{Prop:LyBath.Pham.2007}.2.(b).ii in which the decision maker continuously switches the regimes.

The problem \cref{eq;fund.selection.problem} can be interpreted as an optimal selection of investment funds.
An investor chooses a fund to maximize her expected utility with multiple priors.
The switching costs are interpreted as costs or benefits in changing funds.

We now assume $K_2 > K_1$ and $c_{1,2} > 0 > c_{2,1}$.
Then, heuristically speaking, the fund 2 (regime 2) is more attractive than the fund 1 (regime 1), but one requires the positive switching cost $c_{1,2}$ to switch from the fund 1 to the fund 2.
On the other hand, one gets the switching benefit $-c_{2,1}$ when switching from the fund 2 to the fund 1.
We can also interpret the fund 2 as a new fund well performing and the fund 1 as an old fund less performing.
To obtain customers, the fund 1 begins the campaign that one switching from the fund 2 to the fund 1 obtains the benefit $-c_{2,1}$.
Then, the investor has a motivation switching between the fund 1 and 2.

However, in practice, the investor may doubt the good performance of the fund 2 since the fund 2 is new and less experienced.
The investor therefore considers that the fund 2 has a premium of ambiguity.
Mathematically, this implies that $\kappa_2 > 0$ and $\kappa_1 = 0$.
We now consider the case that $\kappa_2 > 0$ and $\kappa_1 = 0$.

Since the problem \cref{eq;fund.selection.problem} satisfies the monotone conditions, we can use the results in \cite{ly2007explicit}.
Let
\[
K^\kappa_2 = \frac{1}{\rho - (b_2 - \kappa_2\sigma_2)p + \frac{1}{2}\sigma_2^2p(1-p)} > 0.
\]
Then, we have
\[
K^\kappa_2 - K_1 = \frac{p K^\kappa_2 K_1}{2}\Big((1 - p)(\sigma_1^2 - \sigma_2^2) - 2(b_1 - b_2) - 2\kappa_2\sigma_2\Big).
\]
Therefore, the sign of $(1 - p)(\sigma_1^2 - \sigma_2^2) - 2(b_1 - b_2) - 2\kappa_2\sigma_2$ determines the type of the switching strategy.
On the other hand, we have
\[
K_2 - K_1 = \frac{p K_2 K_1}{2}\Big((1 - p)(\sigma_1^2 - \sigma_2^2) - 2(b_1 - b_2)\Big) > 0.
\]
Hence, $(1 - p)(\sigma_1^2 - \sigma_2^2) - 2(b_1 - b_2)$ is positive.
However, if $\kappa_2$ is sufficiently large such that $(1 - p)(\sigma_1^2 - \sigma_2^2) - 2(b_1 - b_2) < 2\kappa_2\sigma_2$,
then $K^\kappa_2 < K_1$.
Therefore, the large ambiguity with respect to the fund 2 can change the type of the switching strategy.

To illustrate effects of ambiguity, we conduct a numerical simulation.
Let $b_1 = 0.03,\;b_2 = 0.07,\;\sigma_1 = 0.1,\;\sigma_2 = 0.3,\;p = 0.5,\;\rho = 0.03,\;c_{1,2} = 30000,$ and $c_{2,1} = -1000$.
Then,
\[
K_1 = 61.53846\cdots < 160 = K_2.
\]
Hence, the investor continuously switches between the fund 1 and 2 without ambiguity.
On the other hand, we have
\[
\frac{(1 - p)(\sigma_1^2 - \sigma_2^2) - 2(b_1 - b_2)}{2\sigma_2} = \frac{1}{15} = 0.0666\cdots.
\]
Thus, if $\kappa_2 > 1/15$, then the type of switching strategy changes to that one always chooses the fund 1.

\begin{figure}[tbhp]
  \centering
  \includegraphics[width=10.0cm,height=7.0cm]{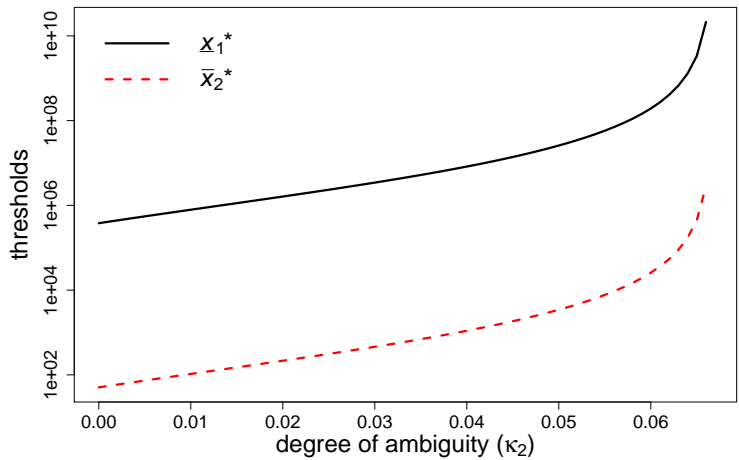}
  \caption{\textbf{The Optimal Switching Thresholds in the Selection of Investment Funds.}
           The vertical axis is a logarithmic scale.
           $\underline x^*_1$ under different $\kappa_2$ is plotted in the solid line.
           $\overline x^*_2$ under different $\kappa_2$ is plotted in the dashed line.}
  \label{fig:fund.selection.thresholds}
\end{figure}

\cref{fig:fund.selection.thresholds} displays the switching thresholds $\underline x_1^*$ and $\overline x_2^*$ with different degrees of ambiguity $\kappa_2$.
If the investor is investing in the fund 1 at time $t$ and if $X_t \geq \underline x_1^*$, then the investor switches from the fund 1 to the fund 2.
On the other hand, if the investor is investing in the fund 2 at time $t$ and if $X_t \leq \overline x_2^*$, then the investor switches from the fund 2 to the fund 1.

According to \cref{fig:fund.selection.thresholds}, in a higher degree of ambiguity $\kappa_2$, both of the thresholds $\underline x_1^*$ and $\overline x_2^*$ are large.
This implies that if $\kappa_2$ is large, then the investor investing in the fund 1 needs sufficiently large wealth $X$ to switch from the fund 1 to the fund 2.
On the other hand, if $\kappa_2$ is large, then the investor investing in the fund 2 switches to the fund 1 with smaller wealth than that in small $\kappa_2$.
Each behavior is well convincing.
The large ambiguity makes the fund 2 less attractive, so the investor tends to choose the fund 1.

\begin{remark}
Suppose $\kappa_1 > 0$. Let
\[
K_1^\kappa = \frac{1}{\rho - (b_1 - \kappa_1\sigma_1)p + \frac{1}{2}\sigma_1^2p(1-p)} > 0,
\]
and
\[
K_2^\kappa - K_1^\kappa = \frac{p K^\kappa_2 K_1^\kappa}{2}\Big((1 - p)(\sigma_1^2 - \sigma_2^2) - 2(b_1 - b_2) - 2 (\kappa_2\sigma_2 - \kappa_1\sigma_1)\Big).
\]
Hence, in this case, if $(1 - p)(\sigma_1^2 - \sigma_2^2) - 2(b_1 - b_2) < 2 (\kappa_2\sigma_2 - \kappa_1\sigma_1)$,
then $K_2^\kappa < K_1^\kappa$.
\end{remark}

\subsection{Buy Low and Sell High}\label{subsec:buy.low.sell.high}
Next, we consider an optimal trading (buy and sell) rule under ambiguity.
Without ambiguity, this problem in trading a mean-reverting asset is well studied in \cite{zhang2008trading}.
We adopt the settings in \cite{zhang2008trading} and consider an optimal trading rule under ambiguity.
Let $d=1$.
A trader concerns with trading of a certain asset.
A cumulative log return of this asset at time $t$ is denoted by $X_t$ and it satisfies the following SDE.
\begin{align}\label{eq:buy.sell.SDE}
\diff X_t &= a(b - X_t)\diff t + \sigma\diff W_t,
\end{align}
where $a>0,\;b\in\real$ and $\sigma>0$ are constants.
Therefore, the asset price at time $t$ is given by $S_t = \exp(X_t)$.
We denote the solution to the SDE \cref{eq:buy.sell.SDE} starting from $X_0 = x$ by $X^x$.
Furthermore, this asset does not have any dividend and coupon.
This implies $\psi = 0$ and $\phi = 0$.

Let $\mathcal I = \{1,2\}$.
The regime $i=1$ means that the trader's position is flat. Hence, the trader wants to buy the asset at as low a price as possible.
The regime $i=2$ means that the trader's position is long. Hence, the trader wants to sell the asset at as high a price as possible.
If the trader goes from the regime 1 to the regime 2, in other words, if the trader buys the asset,
then the switching cost function is
\begin{equation}\label{eq:buying.cost}
c_{1,2}(x) = e^x(1 + K),
\end{equation}
where $K\in(0,1)$ is a constant percentage of slippage or commission per transaction.
On the other hand, if the trader goes from the regime 2 to the regime 1, in other words, if the trader sells the asset,
then the cost (benefit) function is
\begin{equation}\label{eq:selling.benefit}
c_{2,1}(x) = -e^x(1 - K).
\end{equation}
The set of multiple priors is
\[
\Theta^{x,i} = [-\kappa,\kappa],\quad\kappa\geq 0,
\]
for all $x\in\real$ and $i\in\mathcal I$.
Therefore, we assume $\kappa$-ignorance.

The buy low and sell high problem under ambiguity can be interpreted as the following optimal switching problem,
\begin{equation}\label{eq:buy.low.sell.high.problem}
v(x,i) = \sup_{\alpha\in\controlinf{0}{i}}\inf_{\theta\in\Theta[0,\infty)}\expect\left[-\sum_{k=1}^\infty e^{-\rho \tau_k}\zeta^{\theta,0}_{\tau_k} c_{i_{k-1},i_k}(X_{\tau_k}^{x})\right].
\end{equation}
More directly, the problem \cref{eq:buy.low.sell.high.problem} can be expressed as
\begin{align*}
&v(x,1) \\
&= \sup_{\alpha\in\controlinf{0}{i}}\inf_{\theta\in\Theta[0,\infty)}\expect\left[\sum_{k=1}^\infty \Big(e^{-\rho \tau_{2k}}\zeta^{\theta,0}_{\tau_{2k}} e^{X_{\tau_{2k}}^x}(1 - K) - e^{-\rho \tau_{2k-1}}\zeta^{\theta,0}_{\tau_{2k-1}} e^{X_{\tau_{2k-1}}^x}(1 + K)\Big)\right], \\
&v(x,2) \\
&= \sup_{\alpha\in\controlinf{0}{i}}\inf_{\theta\in\Theta[0,\infty)}\expect\left[e^{-\rho \tau_{1}}\zeta^{\theta,0}_{\tau_{1}} e^{X_{\tau_{1}}^x}(1 - K) \right.\\
&\qquad\qquad\qquad\left.+\sum_{k=1}^\infty \Big(e^{-\rho \tau_{2k+1}}\zeta^{\theta,0}_{\tau_{2k+1}} e^{X_{\tau_{2k+1}}^x}(1 - K) - e^{-\rho \tau_{2k}}\zeta^{\theta,0}_{\tau_{2k}} e^{X_{\tau_{2k}}^x}(1 + K)\Big)\right].
\end{align*}

The cost/benefit functions \cref{eq:buying.cost,eq:selling.benefit} do not satisfy the polynomial growth condition and the strong triangular condition.
However, changing variables from $X$ to $S$, then these functions satisfy the polynomial growth condition.
Furthermore, we can easily prove \cref{Prop:cost.upper.bounded} in the problem \cref{eq:buy.low.sell.high.problem} (see Lemma 4 in \cite{zhang2008trading} and \cref{Remark:cost.upper.bounded.infinite} in this paper).
Therefore, we can apply the method in \cref{sec:infinite.horizon}.
Note that for sufficiently large constant $C\geq 0$, the following function satisfies the temporary terminal conditions:
\[
g(x,i) = -\1_{\{i=1\}}e^x(1 - K) - C.
\]
It is easy to show that $g$ satisfies the sufficient condition \cref{eq:sufficient.condition.terminal.submartingale.infinite} for sufficiently large $C$.

According to \cref{Prop:uniform.viscosity}, the value function $v$ is a viscosity solution of the following system of PDEs.
\begin{align}
&\min\{-\ig v(x,1) + \rho v(x,1) + \kappa\sigma|\nabla v(x,1)|, v(x,1) - v(x,2) + e^x(1 + K)\} = 0, \label{eq:short.PDE}\\
&\min\{-\ig v(x,2) + \rho v(x,2) + \kappa\sigma|\nabla v(x,2)|, v(x,2) - v(x,1) - e^x(1 - K)\} = 0, \label{eq:long.PDE}
\end{align}
where
\[
\ig f(x) = a(b - x)\nabla f(x) + \frac{\sigma^2}{2}\frac{\p^2 f(x)}{\p x^2}.
\]
Unfortunately, the problem \cref{eq:buy.low.sell.high.problem} does not satisfy the monotone conditions,
therefore we need to solve the system of PDEs \cref{eq:short.PDE,eq:long.PDE}.
It seems to be difficult to solve this system since it contains the absolute values of the first derivatives of $v$.
However, we can find a continuous solution to the system of PDEs \cref{eq:short.PDE,eq:long.PDE} using the smooth-fit techniques (for details of the smooth-fit techniques, we refer to Chapter 5 in \cite{pham2009continuous}).

First, let $\mathcal C_1$ be a continuation region of the regime 1 such that
\[
\mathcal C_1 = (x_1,\infty),
\]
for some $x_1$. Thus, the trader in the flat position buys the asset whenever the asset price falls below $e^{x_1}$. Also let $\mathcal C_2$ be a continuation region of the regime 2 such that
\[
\mathcal C_2 = (-\infty, x_2),
\]
for some $x_2$. Thus, the trader in the long position sells the asset whenever the asset price exceeds $e^{x_2}$.
Naturally we impose $x_1\leq x_2$.
We assume that
\begin{equation}\label{eq:monotones.buy.sell}
\nabla v(x,1) \leq 0,\; \forall x\in \mathcal C_1, \quad\mbox{and}\quad
\nabla v(x,2) \geq 0,\; \forall x\in \mathcal C_2.
\end{equation}
By \cite{zhang2008trading},
the PDE,
\[
-\ig V(x,1) + \rho V(x,1) - \kappa\sigma\nabla V(x,1) = 0,
\]
on $\mathcal C_1$ has a solution such that
\[
V(x,1) = C_1\varphi_1(x),
\]
where $C_1,\;m=\sqrt{2a}/\sigma,$ and $\lambda = \rho / a$ are constants, and
\[
\varphi_1(x) = \int_0^\infty t^{\lambda - 1} e^{-0.5 t^2 + m (b + \kappa \sigma / a - x)t}\diff t.
\]
Similarly, the PDE,
\[
-\ig V(x,2) + \rho V(x,2) + \kappa\sigma\nabla V(x,2) = 0,
\]
on $\mathcal C_2$ has a solution such that
\[
V(x,2) = C_2\varphi_2(x),
\]
where $C_2$ is a constant and
\[
\varphi_2(x) = \int_0^\infty t^{\lambda - 1} e^{-0.5 t^2 - m (b - \kappa \sigma / a - x)t}\diff t.
\]
Now, let us guess that candidates of the solution to the PDEs \cref{eq:short.PDE,eq:long.PDE} are
\begin{align}
v(x,1) &= \left\{\begin{array}{ll}
                 V(x,1),              &\quad \mbox{if }x\in\mathcal C_1, \\
                 V(x,2) - e^x(1 + K), &\quad \mbox{if }x\notin\mathcal C_1,
                 \end{array}\right.\label{eq:buy.sell.value.function.1}\\
v(x,2) &= \left\{\begin{array}{ll}
                 V(x,2),              &\quad \mbox{if }x\in\mathcal C_2, \\
                 V(x,1) + e^x(1 - K), &\quad \mbox{if }x\notin\mathcal C_2.
                 \end{array}\right.\label{eq:buy.sell.value.function.2}
\end{align}
Let
\begin{align*}
\varphi_1^*(x) = \int_0^\infty t^{\lambda} e^{-0.5 t^2 + m (b + \kappa \sigma / a - x)t}\diff t,\quad
\varphi_2^*(x) = \int_0^\infty t^{\lambda} e^{-0.5 t^2 - m (b - \kappa \sigma / a - x)t}\diff t.
\end{align*}
Then, $\nabla V(x,1) = -m C_1\varphi_1^*(x)$ and $\nabla V(x,2) = mC_2\varphi_2^*(x)$.
Hence, by the conditions \cref{eq:monotones.buy.sell}, we need $C_1\geq 0$ and $C_2 \geq 0$.
By the smooth-fit conditions, we need
\begin{align}
&\left\{\begin{array}{l}
V(x_1,1) = V(x_1,2) - e^{x_1}(1+K), \\
\nabla V(x_1,1) = \nabla V(x_1,2) - e^{x_1}(1+K), \\
V(x_2,2) = V(x_2,1) + e^{x_2}(1-K), \\
\nabla V(x_2,2) = \nabla V(x_2,1) + e^{x_2}(1-K),
\end{array}\right. \label{eq:smooth.fit.01}\\
&\left\{\begin{array}{l}
v(x,1) \geq v(x,2) - e^{x}(1+K), \quad\mbox{on }(x_1,\infty), \\
v(x,2) \geq v(x,1) + e^{x}(1-K), \quad\mbox{on }(-\infty,x_2),
\end{array}\right. \label{eq:smooth.fit.02}\\
&\left\{\begin{array}{l}
(- \ig + \rho + \kappa\sigma|\nabla|)(V(x,2) - e^{x}(1+K))\geq 0, \quad\mbox{on }(-\infty,x_1), \\
(- \ig + \rho + \kappa\sigma|\nabla|)(V(x,1) + e^{x}(1-K))\geq 0, \quad\mbox{on }(x_2,\infty).
\end{array}\right. \label{eq:smooth.fit.03}
\end{align}
After simple algebraic computation, the equalities \cref{eq:smooth.fit.01} can be expressed as
\begin{align}\label{eq:smooth.fit.a01}
\left(
      \begin{array}{c}
       C_1 \\
       C_2
      \end{array}
\right)
 &=
e^{x_1}(1+K)
\left(
      \begin{array}{cc}
       -\varphi_1(x_1) &\varphi_2(x_1) \\
       \varphi_1^*(x_1) &\varphi_2^*(x_1) 
      \end{array}
\right)^{-1}
\left(
      \begin{array}{c}
       1 \\
       1 / m
      \end{array}
\right)  \\
 &= e^{x_2}(1-K)
\left(
      \begin{array}{cc}
       -\varphi_1(x_2) &\varphi_2(x_2) \\
       \varphi_1^*(x_2) &\varphi_2^*(x_2)
      \end{array}
\right)^{-1}
\left(
      \begin{array}{c}
       1 \\
       1 / m
      \end{array}
\right) \geq 0.\nonumber
\end{align}
By the definitions of $v$, the inequalities \cref{eq:smooth.fit.02} are equivalent to
\begin{equation}\label{eq:smooth.fit.a02}
V(x,1) \geq V(x,2) - e^{x}(1+K), \quad
V(x,2) \geq V(x,1) + e^{x}(1-K),
\end{equation}
on $(x_1,x_2)$.
For the first inequality of \cref{eq:smooth.fit.03}, we have
\begin{align*}
(- \ig + \rho + \kappa\sigma|\nabla|)(V(x,2) - e^{x}(1+K)) 
& = (- \ig + \rho + \kappa\sigma|\nabla|)( - e^{x}(1+K)) \\
& = -\Big(\rho - a (b - x) - \frac{\sigma^2}{2} - \kappa\sigma\Big)e^x(1+K) \geq 0
\end{align*}
on $(-\infty,x_1)$ since $(-\infty,x_1)\subseteq\mathcal C_2$.
Thus, the condition expressed by the first inequality is equivalent to
\begin{equation}\label{eq:smooth.fit.a03}
x_1 \leq \frac{1}{a}\left(\frac{\sigma^2}{2} + ab + \kappa\sigma - \rho\right).
\end{equation}
Similarly, the condition expressed by the second inequality of \cref{eq:smooth.fit.03} is equivalent to
\begin{equation}\label{eq:smooth.fit.a04}
x_2 \geq \frac{1}{a}\left(\frac{\sigma^2}{2} + ab - \kappa\sigma - \rho\right).
\end{equation}
Finally, we need
\begin{align}\label{eq:smooth.fit.a05}
e^{x_2}(1 - K) > e^{x_1}(1 + K)
&\Leftrightarrow x_2 - x_1 > \log(1 + K) - \log(1 - K).
\end{align}
Hence, if $C_1,C_2,x_1$ and $x_2$ satisfy the conditions \cref{eq:smooth.fit.a01,eq:smooth.fit.a02,eq:smooth.fit.a03,eq:smooth.fit.a04,eq:smooth.fit.a05},
then the candidates of the solutions \cref{eq:buy.sell.value.function.1,eq:buy.sell.value.function.2} are true viscosity solutions to the system of the PDEs \cref{eq:short.PDE,eq:long.PDE}.

To illustrate effects of ambiguity, we conduct a numerical simulation.
Let $a = 0.8,\;b = 2,\;\sigma=0.5,\;\rho = 0.5,$ and $K=0.01$.
The values of these parameters are the same as \cite{zhang2008trading}.
We compute thresholds $(x_1,x_2)$ with different degrees of ambiguity $\kappa$.

\begin{figure}[tbhp]
 \centering
  \includegraphics[width=10.0cm,height=7.0cm]{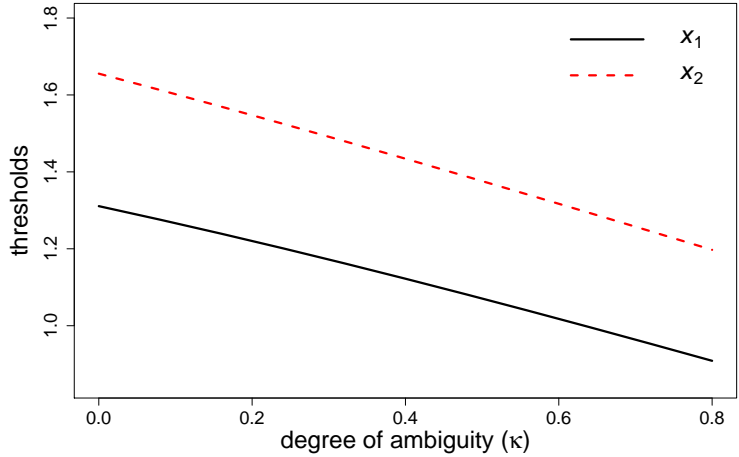}
  \caption{\textbf{The Optimal Switching Thresholds in the Buy Low and Sell High Problem.}
           $x_1$ under different $\kappa$ is plotted in the solid line.
           $x_2$ under different $\kappa$ is plotted in the dashed line.}
  \label{fig:buy.sell.01}
\end{figure}

\Cref{fig:buy.sell.01} displays the thresholds.
According to \cref{fig:buy.sell.01}, in a larger degree of ambiguity, both of the optimal thresholds become small.
The long position trader (that is, the initial regime is 2) considers the worst case that the steady mean of $X$ is smaller than that without ambiguity.
Therefore, the trader sells the asset at a lower price than that without ambiguity.

On the other hand, in the flat position case (that is, the initial regime is 1), the trader also buys the asset at a lower price than that without ambiguity.
That is because a gain of the trader in the flat position does not realize until he or she sells the asset.
Now, we assume that the trader considers the case when the steady mean of $X$ is larger than that without ambiguity.
Then, the trader can expect a bigger profit in his or her belief than that in the true probability measure.
This is a contradiction since the trader considers the worst case.
Therefore, even if the trader has the flat position, he or she considers the case that the steady mean of $X$ is smaller than that without ambiguity.
Hence, the optimal thresholds of buying the asset under ambiguity is lower than that without ambiguity.

In \cite{zhang2008trading}, the comparative statics with varying the steady mean of $X$, (i.e., $b$), is conducted.
The results in \cite{zhang2008trading} are that in a small $b$, both of the optimal thresholds are also small.
These are similar to the results in large ambiguity.
However, the results under large ambiguity can not be reproduced by a small $b$.
By the equality \cref{eq:smooth.fit.a01} with $\kappa = 0$,
the optimal thresholds under the steady mean $b$ are equal to the optimal thresholds under the steady mean $\widetilde b$ plus $b - \widetilde b$ for all $b,\widetilde b\in\real$ if the other parameters are the same.
Therefore, the optimal thresholds are linear in the steady mean $b$.

\begin{figure}[tbhp]
 \centering
  \includegraphics[width=10.0cm,height=7.0cm]{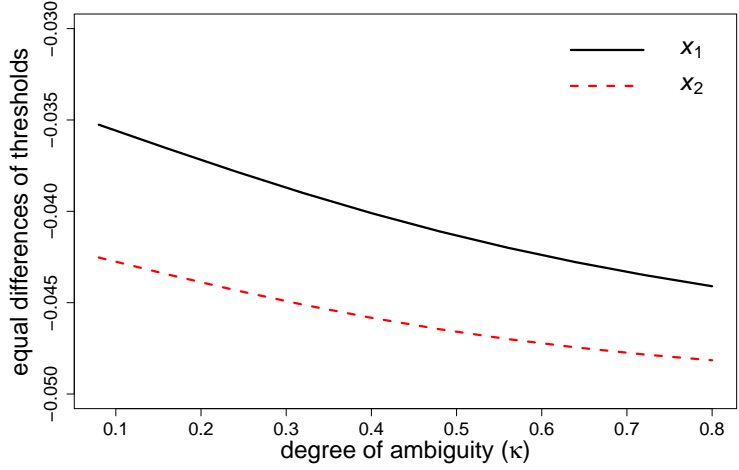}
  \caption{\textbf{The Equal Differences of the Optimal Switching Thresholds in the Buy Low and Sell High Problem.}
           $x_1$ under different $\kappa$ is plotted in the solid line.
           $x_2$ under different $\kappa$ is plotted in the dashed line.
           Each interval of $\kappa$ is 0.08.}
  \label{fig:buy.sell.02}
\end{figure}

On the other hand, \cref{fig:buy.sell.02} displays equal differences of the optimal thresholds with different degrees of ambiguity.
According to \cref{fig:buy.sell.02}, the equal differences are not constant, therefore the optimal thresholds are not linear in the degree of ambiguity $\kappa$.
Our PDEs \cref{eq:short.PDE,eq:long.PDE} cause these non-linearities.
The PDEs \cref{eq:short.PDE,eq:long.PDE} can not be expressed as any variational inequality of an optimal switching problem without ambiguity since these do not satisfy the monotone conditions.
Indeed, the difference $x_2 - x_1$ without ambiguity is constant over $b$,
whereas $x_2 - x_1$ is small with large $\kappa$.
Thus, the optimal switching problem under ambiguity can generate this interesting result which can not be reproduced by the problem without ambiguity.

\section*{Acknowledgments}
The author is grateful to Masahiko Egami for helpful advices and discussions on this paper.

\appendix
\section{\texorpdfstring{The Moment Estimates of $X$}{The Moment Estimates of X}}\label{subsec:moment.estimate.X}

\begin{proof}[Proof of \cref{Prop:q.th.integrable}]
Since $x\rightarrow\|x\|^q$ is twice continuously differentiable for all $q\geq 4$,
we can apply the Ito's lemma to $\|X_s^{t,x,i,\alpha}\|^q$.
Then, for all $s\in[t,T]$, using the quadratic growth condition for $b$ and $\sigma$, we have
\begin{align*}
\|X_s^{t,x,i,\alpha}\|^q
 &= \|x\|^q + \int_t^sq\|X_r^{t,x,i,\alpha}\|^{q-2}(X_r^{t,x,i,\alpha})^\prime b(r,X_r^{t,x,i,\alpha},\alpha_r)\diff r \\
 &\quad + \frac{1}{2}\int_t^s\Big(q(q-2)\|X_r^{t,x,i,\alpha}\|^{q-4}\|\sigma^\prime(r,X_r^{t,x,i,\alpha},\alpha_r)X_r^{t,x,i,\alpha}\|^2 \\
 &\qquad\qquad q\|X_r^{t,x,i,\alpha}\|^{q-2}\|\sigma(r,X_r^{t,x,i,\alpha},\alpha_r)\|^2\Big)\diff r \\
 &\quad + \int_t^sq\|X_r^{t,x,i,\alpha}\|^{q-2}(X_r^{t,x,i,\alpha})^\prime\sigma(r,X_r^{t,x,i,\alpha},\alpha_r)\diff W_r\\
 &\leq \|x\|^q + \widehat C_q \int_t^s\Big(1 + \|X_r^{t,x,i,\alpha}\|^q\Big)\diff r \\
 &\quad + q\int_t^s\|X_r^{t,x,i,\alpha}\|^{q-2}(X_r^{t,x,i,\alpha})^\prime\sigma(r,X_r^{t,x,i,\alpha},\alpha_r)\diff W_r,
\end{align*}
where $\widehat C_q$ is the constant only depending on $q$ and $L$.
The above stochastic integral in the right hand side is a local martingale.
Hence, there exists an increasing sequence of stopping times $(\tau_n)_{n\geq 1}$ such that $\tau_n\rightarrow\infty$ and
\begin{align*}
\expect[\|X_{s\wedge\tau_n}^{t,x,i,\alpha}\|^q]
\leq \|x\|^2 + \widehat C_q \expect\left[\int_t^{s\wedge\tau_n}\Big(1 + \|X_r^{t,x,i,\alpha}\|^q\Big)\diff r\right],
\end{align*}
for all $s\in[t,T]$ and $n\geq 1$, where $a\wedge b = \min\{a,b\}$.
By the Fatou lemma, the monotone convergence theorem and the continuity of $X^{t,x,i,\alpha}$,
taking a limit, we have
\begin{align*}
1 + \expect[\|X_{s}^{t,x,i,\alpha}\|^q]
&\leq 1 + \liminf_{n\rightarrow\infty}\expect[\|X_{s\wedge\tau_n}^{t,x,i,\alpha}\|^q] \\
&\leq 1 + \|x\|^2 + \widehat C_q \expect\left[\int_t^{s}\Big(1 + \|X_r^{t,x,i,\alpha}\|^q\Big)\diff r\right] \\
&\quad = 1 + \|x\|^2 + \widehat C_q \int_t^{s}\expect\left[1 + \|X_r^{t,x,i,\alpha}\|^q\right]\diff r.
\end{align*}
By the Gronwall lemma, we have
\begin{equation}
\expect[\|X_{s}^{t,x,i,\alpha}\|^q]
\leq 1 + \expect[\|X_{s}^{t,x,i,\alpha}\|^q]
\leq (1 + \|x\|^q)e^{\widehat C_q(s-t)}, \label{eq:usual.estimate.X}
\end{equation}
for all $0\leq t\leq s$ and $x\in\real^d$.
Similarly, we have
\begin{align*}
\max_{t\leq s\leq T}\|X_s^{t,x,i,\alpha}\|^q
 &\leq \|x\|^q + \widehat C_q \int_t^T\Big(1 + \|X_r^{t,x,i,\alpha}\|^q\Big)\diff r \\
 &\qquad + q\max_{t\leq s\leq T}\int_t^s\|X_r^{t,x,i,\alpha}\|^{q-2}(X_r^{t,x,i,\alpha})^\prime\sigma(r,X_r^{t,x,i,\alpha},\alpha_r)\diff W_r.
\end{align*}
By the Burkholder-Davis-Gundy inequality and Jensen inequality, we have
\begin{align*}
&\expect\left[\max_{t\leq s\leq T}\int_t^s\|X_r^{t,x,i,\alpha}\|^{q-2}(X_r^{t,x,i,\alpha})^\prime\sigma(r,X_r^{t,x,i,\alpha},\alpha_r)\diff W_r\right] \\
&\leq \expect\left[\left(\int_t^T\|X_r^{t,x,i,\alpha}\|^{2q-4}\|\sigma^\prime(r,X_r^{t,x,i,\alpha},\alpha_r)X_r^{t,x,i,\alpha}\|^2\diff r\right)^{1/2}\right] \\
&\leq \expect\left[\left(\int_t^T\|X_r^{t,x,i,\alpha}\|^{2q-2}\|\sigma(r,X_r^{t,x,i,\alpha},\alpha_r)\|^2\diff r\right)^{1/2}\right] \\
&\leq L\expect\left[\left(\int_t^T\|X_r^{t,x,i,\alpha}\|^{2q-2}\Big(1 + \|X_r^{t,x,i,\alpha}\|^2\Big)\diff r\right)^{1/2}\right] \\
&\leq \sqrt{2}L\left(\int_t^T\expect\left[1 + \|X_r^{t,x,i,\alpha}\|^{2q}\right]\diff r\right)^{1/2}.
\end{align*}
Furthermore, using the inequality \cref{eq:usual.estimate.X}, we have
\begin{align*}
\left(\int_t^T\expect\left[1 + \|X_r^{t,x,i,\alpha}\|^{2q}\right]\diff r\right)^{1/2} 
&\leq \left(\int_t^T(1 + \|x\|^{2q})e^{\widehat C_{2q}(r-t)}\diff r\right)^{1/2} \\
&\leq \frac{1}{\widehat C_{2q}^{1/2}}(1 + \|x\|^q)e^{\widehat C_{2q}(T-t)/2}.
\end{align*}
Thus, we obtain
\begin{align*}
\expect\left[\max_{t\leq s\leq T}\|X_{s}^{t,x,i,\alpha}\|^q\right]
&\leq 1 + \expect\left[\max_{t\leq s\leq T}\|X_{s}^{t,x,i,\alpha}\|^q\right] 
\leq C_{q,X}(1 + \|x\|^q)e^{C_q(T-t)},
\end{align*}
where
\[
C_{q,X} = \max\left\{1,\widehat C_q,\sqrt{\frac{2}{\widehat C_q}}qL\right\},\quad
C_q = \frac{\widehat C_{2q}}{2}.
\]
If $q\in(0, 4)$, then by the Jensen inequality, we have
\begin{align*}
\expect\left[\max_{t\leq s\leq T}\|X_{s}^{t,x,i,\alpha}\|^q\right]
 = \expect\left[\Big(\max_{t\leq s\leq T}\|X_{s}^{t,x,i,\alpha}\|^4\Big)^{q/4}\right] 
 &\leq \left(\expect\left[\max_{t\leq s\leq T}\|X_{s}^{t,x,i,\alpha}\|^4\right]\right)^{q/4}\\
 &\leq C_{4,X}^{q/4}(1 + \|x\|^4)^{q/4}e^{(qC_4/4)(T-t)} \\
 &\leq C_{4,X}^{q/4}(1 + \|x\|^q)e^{(qC_4/4)(T-t)}.
\end{align*}
It is easy to show the inequality \cref{eq:Lp.estimate.X.infinite} applying the Ito's lemma to $e^{-\rho s}(1+\|X_s^{t,x,i,\alpha}\|^q)$.
\end{proof}

\section{\texorpdfstring{Verification of $Y$}{Verification of Y}}\label{subsec:verification.Y}
\begin{proof}[Proof of \cref{Prop:excess.Y}]
\textit{Step.1 $Y$ is at least as large as any objective function.}
We define a sequence of random variables as follows.
\[
X^0 := \eta,\quad
X^k := X_{\tau_k}^{\tau_{k-1},X^{k-1},i_{k-1}},\quad k\geq 1.
\]
By the definition, $X^k\in L^{2q}_{\tau_k}(\real^d)$ for all $k$.
Furthermore, for all $k\geq 1$ and $t\in[\tau_{k-1},\tau_k)$,
the strong uniqueness of $X$ leads to that
\begin{equation}\label{eq:X.uniqueness2}
X_t^{\nu,\eta,\iota,\alpha} = X_t^{\tau_{k-1},X^{k-1},i_{k-1}},
\end{equation}
$\prob$-almost surely.

Let $N = \inf\{k\;|\;\tau_k\geq T\}$ and $\tau_0 = \nu$.
By the admissibility of $\alpha=(\tau_k,i_k)_{k\geq 0}$,
$N$ is finite $\prob$-almost surely.
Let $\overline Z^{\nu,\eta,\iota,\alpha}$ be a stochastic process such that
\begin{align}
\overline Z^{\nu,\eta,\iota,\alpha}_t &= \sum_{k=1}^{N}Z^{\tau_{k-1},X^{k-1},i_{k-1}}_{t}\1_{[\tau_{k-1},\tau_{k})}(t),\quad t\in[0,T],\label{eq:Z.over}
\end{align}

Let $D^k$ be a stochastic process on $[\tau_{k-1},\tau_k]$ such that
\[
D_t^k = \exp\left\{-\int_{\tau_{k-1}}^t\rho(s,X_s^{\tau_{k-1},X^{k-1},i_{k-1}},i_{k-1})\diff s\right\},\quad t\in[\tau_{k-1},\tau_k].
\]
By the equality \cref{eq:X.uniqueness2}, we have
\begin{align}
D_t^{\nu,\eta,\iota,\alpha} &=  D_t^1,\;t\in[\tau_0,\tau_1],\nonumber\\
D_t^{\nu,\eta,\iota,\alpha} &=  D_{\tau_{k-1}}^{\nu,\eta,\iota,\alpha}D_t^k,\;t\in[\tau_{k-1},\tau_k],\;k\geq 2. \nonumber
\end{align}

Then, for any $k\geq 1$,
applying the Ito's lemma to $D_t^kY_t^{\tau_{k-1},X^{k-1},i_{k-1}}$ leads to
\begin{align*}
Y_{\tau_{k-1}}^{\tau_{k-1},X^{k-1},i_{k-1}}
&\geq D_{\tau_{k}}^k Y_{\tau_{k}}^{\tau_{k-1},X^{k-1},i_{k-1}} 
 + \int_{\tau_{k-1}}^{\tau_k}D_s^k\Big(\psi(s,X_s^{\tau_{k-1},X^{k-1},i_{k-1}},i_{k-1}) \\
&\quad - \varsigma(s,X_s^{\tau_{k-1},X^{k-1},i_{k-1}},i_{k-1},Z_s^{\tau_{k-1},X^{k-1},i_{k-1}})\Big)\diff s  \\
&\qquad - \int_{\tau_{k-1}}^{\tau_k}D_s^k(Z_s^{\tau_{k-1},X^{k-1},i_{k-1}})^\prime \diff W_s,
\end{align*}
where we have used the non-negativity of $D_t^k$ and monotonicity of $K_t^{\tau_{k-1},X^{k-1},i_{k-1}}$.
Furthermore, by the pathwise uniqueness of $X$ and $Y$ (see \cref{Prop:minimal.uniqueness,eq:X.uniqueness2,eq:Z.over}),
we have
\begin{multline*}
Y_{\tau_{k-1}}^{\tau_{k-1},X^{k-1},i_{k-1}} \\
\geq D_{\tau_{k}}^k Y_{\tau_{k}}^{\tau_{k},X^k,i_{k-1}} 
 + \int_{\tau_{k-1}}^{\tau_k}D_s^k\Big(\psi(s,X_s^{\nu,\eta,\iota,\alpha},\alpha_s) - \varsigma(s,X_s^{\nu,\eta,\iota,\alpha},\alpha_s,\overline Z_s^{\nu,\eta,\iota,\alpha})\Big)\diff s \\
 - \int_{\tau_{k-1}}^{\tau_k}D_s^k(\overline Z_s^{\nu,\eta,\iota,\alpha})^\prime \diff W_s.
\end{multline*}

Since each $Y_{\tau_{k}}^{\tau_{k},X^k,i_{k-1}}$ dominates the lower barrier, we obtain
\begin{align*}
 Y^{\nu,\eta,\iota}_\nu 
 &\geq D_{\tau_{1}}^1 Y_{\tau_{1}}^{\tau_{1},X^1,i_{0}} + \int_{\tau_{0}}^{\tau_1}D_s^1\Big(\psi(s,X_s^{\nu,\eta,\iota,\alpha},\alpha_s) - \varsigma(s,X_s^{\nu,\eta,\iota,\alpha},\alpha_s,\overline Z_s^{\nu,\eta,\iota,\alpha})\Big)\diff s \\
 &\quad - \int_{\tau_{0}}^{\tau_1}D_s^1(\overline Z_s^{\nu,\eta,\iota,\alpha})^\prime \diff W_s\\
 &\geq D_{\tau_{1}}^1\left(Y_{\tau_{1}}^{\tau_{1},X^1,i_{1}} - c_{i_0,i_1}(\tau_1,X_{\tau_1}^{\tau_1,X^1,i_0})\right)\\
 &\quad + \int_{\tau_{0}}^{\tau_1}D_s^1\Big(\psi(s,X_s^{\nu,\eta,\iota,\alpha},\alpha_s) - \varsigma(s,X_s^{\nu,\eta,\iota,\alpha},\alpha_s,\overline Z_s^{\nu,\eta,\iota,\alpha})\Big)\diff s \\
 &\quad - \int_{\tau_{0}}^{\tau_1}D_s^1(\overline Z_s^{\nu,\eta,\iota,\alpha})^\prime \diff W_s\\
 &= D_{\tau_{1}}^1 Y_{\tau_{1}}^{\tau_{1},X^1,i_{1}} - D_{\tau_{1}}^1 c_{i_0,i_1}(\tau_1,X_{\tau_1}^{\nu,\eta,\iota,\alpha})\\
 &\quad + \int_{\tau_{0}}^{\tau_1}D_s^1\Big(\psi(s,X_s^{\nu,\eta,\iota,\alpha},\alpha_s) - \varsigma(s,X_s^{\nu,\eta,\iota,\alpha},\alpha_s,\overline Z_s^{\nu,\eta,\iota,\alpha})\Big)\diff s \\
 &\quad - \int_{\tau_{0}}^{\tau_1}D_s^1(\overline Z_s^{\nu,\eta,\iota,\alpha})^\prime \diff W_s\\
 &\geq D_{\tau_{1}}^1 Y_{\tau_{2}}^{\tau_{1},X^1,i_{1}} - D_{\tau_{1}}^1 c_{i_0,i_1}(\tau_1,X_{\tau_1}^{\nu,\eta,\iota,\alpha}) \\
 &\quad + \int_{\tau_{0}}^{\tau_1}D_s^1\Big(\psi(s,X_s^{\nu,\eta,\iota,\alpha},\alpha_s) -\varsigma(s,X_s^{\nu,\eta,\iota,\alpha},\alpha_s,\overline Z_s^{\nu,\eta,\iota,\alpha})\Big)\diff s \\
 &\quad + D_{\tau_{1}}^1\int_{\tau_{1}}^{\tau_2}D_s^2\Big(\psi(s,X_s^{\nu,\eta,\iota,\alpha},\alpha_s) - \varsigma(s,X_s^{\nu,\eta,\iota,\alpha},\alpha_s,\overline Z_s^{\nu,\eta,\iota,\alpha})\Big)\diff s \\
 &\quad - \int_{\tau_{0}}^{\tau_1}D_s^1(\overline Z_s^{\nu,\eta,\iota,\alpha})^\prime \diff W_s - D_{\tau_{1}}^1\int_{\tau_{1}}^{\tau_2}D_s^2(\overline Z_s^{\nu,\eta,\iota,\alpha})^\prime \diff W_s \\
 &= D_{\tau_{1}}^1 Y_{\tau_{2}}^{\tau_{2},X^2,i_{1}} - D_{\tau_{1}}^1 c_{i_0,i_1}(\tau_1,X_{\tau_1}^{\nu,\eta,\iota,\alpha})\\
 &\quad + \int_{\tau_{0}}^{\tau_2}D_s^{\nu,\eta,\iota,\alpha}\Big(\psi(s,X_s^{\nu,\eta,\iota,\alpha},\alpha_s) - \varsigma(s,X_s^{\nu,\eta,\iota,\alpha},\alpha_s,\overline Z_s^{\nu,\eta,\iota,\alpha})\Big)\diff s \\
 &\quad - \int_{\tau_{0}}^{\tau_2}D_s^{\nu,\eta,\iota,\alpha}(\overline Z_s^{\nu,\eta,\iota,\alpha})^\prime \diff W_s,
\end{align*}
where we have used \cref{Prop:minimal.uniqueness,eq:X.uniqueness2}.
By repeating this up to $n\geq 1$, we have
\begin{align*}
Y^{\nu,\eta,\iota}_\nu
 &\geq  D_{\tau_{n}}^{\nu,\eta,\iota,\alpha}Y^{\tau_n,X^n,i_{n-1}}_{\tau_n} - \sum_{k=1}^{n-1} D_{\tau_{k}}^{\nu,\eta,\iota,\alpha}c_{i_{k-1},i_k}(\tau_k,X_{\tau_k}^{\nu,\eta,\iota,\alpha}) \\
 &\quad + \int_{\tau_{0}}^{\tau_n}D_s^{\nu,\eta,\iota,\alpha}\Big(\psi(s,X_s^{\nu,\eta,\iota,\alpha},\alpha_s) - \varsigma(s,X_s^{\nu,\eta,\iota,\alpha},\alpha_s,\overline Z_s^{\nu,\eta,\iota,\alpha})\Big)\diff s \\
 &\qquad - \int_{\tau_{0}}^{\tau_n}D_s^{\nu,\eta,\iota,\alpha}(\overline Z_s^{\nu,\eta,\iota,\alpha})^\prime \diff W_s,
\end{align*}
for all $n$.
Since $\tau_n\rightarrow T\;\prob$-a.s. and $Y^{\nu,\eta,\iota}$ is continuous,
taking a limit, we have
\begin{align*}
Y^{\nu,\eta,\iota}_\nu
&\geq  D_{T}^{\nu,\eta,\iota,\alpha}g(X_T^{\nu,\eta,\iota,\alpha},\alpha_T) - \sum_{\nu\leq \tau_k\leq T} D_{\tau_{k}}^{\nu,\eta,\iota,\alpha}c_{i_{k-1},i_k}(\tau_k,X_{\tau_k}^{\nu,\eta,\iota,\alpha}) \\
&\quad + \int_{\nu}^{T}D_{s}^{\nu,\eta,\iota,\alpha}\Big(\psi(s,X_s^{\nu,\eta,\iota,\alpha},\alpha_s) - \varsigma(s,X_s^{\nu,\eta,\iota,\alpha},\alpha_s,\overline Z_s^{\nu,\eta,\iota,\alpha})\Big)\diff s \\
&\qquad - \int_{\nu}^{T}D_{s}^{\nu,\eta,\iota,\alpha}(\overline Z^{\nu,\eta,\iota,\alpha}_s)^\prime \diff W_s.
\end{align*}
Similarly to the above, we have
\begin{align*}
D_{t}^{\nu,\eta,\iota,\alpha}Y^{\nu,\eta,\iota}_t
&\geq  D_{T}^{\nu,\eta,\iota,\alpha}g(X_T^{\nu,\eta,\iota,\alpha},\alpha_T) - \sum_{t\leq \tau_k\leq T} D_{\tau_{k}}^{\nu,\eta,\iota,\alpha}c_{i_{k-1},i_k}(\tau_k,X_{\tau_k}^{\nu,\eta,\iota,\alpha})  \\
&\quad + \int_{t}^{T}D_{s}^{\nu,\eta,\iota,\alpha}\Big(\psi(s,X_s^{\nu,\eta,\iota,\alpha},\alpha_s) - \varsigma(s,X_s^{\nu,\eta,\iota,\alpha},\alpha_s,\overline Z_s^{\nu,\eta,\iota,\alpha})\Big)\diff s \nonumber \\
&\qquad - \int_{t}^{T}D_{s}^{\nu,\eta,\iota,\alpha}(\overline Z^{\nu,\eta,\iota,\alpha}_s)^\prime \diff W_s, \nonumber
\end{align*}
for all $t\in[\nu,T]$. On the other hand, we have
\begin{align*}
D_{t}^{\nu,\eta,\iota,\alpha}Y^{\nu,\eta,\iota,\alpha}_t
&= D_{T}^{\nu,\eta,\iota,\alpha}g(X_T^{\nu,\eta,\iota,\alpha},\alpha_T) - \sum_{t\leq \tau_k\leq T} D_{\tau_{k}}^{\nu,\eta,\iota,\alpha}c_{i_{k-1},i_k}(\tau_k,X_{\tau_k}^{\nu,\eta,\iota,\alpha}) \\
&\quad + \int_{t}^{T}D_{s}^{\nu,\eta,\iota,\alpha}\Big(\psi(s,X_s^{\nu,\eta,\iota,\alpha},\alpha_s) - \varsigma(s,X_s^{\nu,\eta,\iota,\alpha},\alpha_s, Z_s^{\nu,\eta,\iota,\alpha})\Big)\diff s \\
&\qquad - \int_{t}^{T}D_{s}^{\nu,\eta,\iota,\alpha}(Z^{\nu,\eta,\iota,\alpha}_s)^\prime \diff W_s,
\end{align*}
for all $t\in[\nu,T]$.
Hence, it holds that
\begin{multline}\label{eq:ineq.optim.01}
D_{t}^{\nu,\eta,\iota,\alpha}\Big(Y^{\nu,\eta,\iota}_t - Y^{\nu,\eta,\iota,\alpha}_t\Big) \\
\geq -\int_{t}^{T}D_{s}^{\nu,\eta,\iota,\alpha}\Big(\varsigma(s,X_s^{\nu,\eta,\iota,\alpha},\alpha_s,\overline Z_s^{\nu,\eta,\iota,\alpha}) - \varsigma(s,X_s^{\nu,\eta,\iota,\alpha},\alpha_s, Z_s^{\nu,\eta,\iota,\alpha})\Big)\diff s \\
 - \int_{t}^{T}D_{s}^{\nu,\eta,\iota,\alpha}(\overline Z^{\nu,\eta,\iota,\alpha}_s - Z^{\nu,\eta,\iota,\alpha}_s)^\prime \diff W_s \\
= \int_{t}^{T}D_{s}^{\nu,\eta,\iota,\alpha}\Delta_s^\prime(\overline Z^{\nu,\eta,\iota,\alpha}_s - Z^{\nu,\eta,\iota,\alpha}_s)\diff s 
 - \int_{t}^{T}D_{s}^{\nu,\eta,\iota,\alpha}(\overline Z^{\nu,\eta,\iota,\alpha}_s - Z^{\nu,\eta,\iota,\alpha}_s)^\prime \diff W_s, 
\end{multline}
where $(\Delta_s)_{\nu\leq s\leq T}$ is a $d$-dimensional adopted process as follows:
Now, we denote by $x_{i,s}$ the $i$th component of a random vector process $(x_u)_{u\geq 0}$ at time $s$.\\
Let $\overline Z^{\nu,\eta,\iota,\alpha,\; i}_{s} = (\overline Z^{\nu,\eta,\iota,\alpha}_{1,s},\dots,\overline Z^{\nu,\eta,\iota,\alpha}_{i-1,s},\overline Z^{\nu,\eta,\iota,\alpha}_{i,s},Z^{\nu,\eta,\iota,\alpha}_{i+1,s},\dots,Z^{\nu,\eta,\iota,\alpha}_{d,s})^\prime$
and \\
let $Z^{\nu,\eta,\iota,\alpha,\; i}_{s} = (\overline Z^{\nu,\eta,\iota,\alpha}_{1,s},\dots,\overline Z^{\nu,\eta,\iota,\alpha}_{i-1,s}, Z^{\nu,\eta,\iota,\alpha}_{i,s},Z^{\nu,\eta,\iota,\alpha}_{i+1,s},\dots,Z^{\nu,\eta,\iota,\alpha}_{d,s})^\prime$.
$\Delta_{i,s}$ is
\[
\Delta_{i,s} = -\dfrac{\varsigma(s,X_s^{\nu,\eta,\iota,\alpha},\alpha_s,\overline Z_s^{\nu,\eta,\iota,\alpha,\; i}) - \varsigma(s,X_s^{\nu,\eta,\iota,\alpha},\alpha_s, Z_s^{\nu,\eta,\iota,\alpha,\; i})}{\overline Z^{\nu,\eta,\iota,\alpha}_{i,s} - Z^{\nu,\eta,\iota,\alpha}_{i,s}},
\]
if $\overline Z^{\nu,\eta,\iota,\alpha}_{i,s} \neq Z^{\nu,\eta,\iota,\alpha}_{i,s}$ and $\Delta_{i,s} = 0$ otherwise.
Then, $(\Delta_s)_{\nu\leq s\leq T}$ is uniformly bounded since $z\rightarrow\varsigma(s,X_s^{\nu,\eta,\iota,\alpha},\alpha_s,z)$ is uniformly Lipschitz for all $s$.
This implies that the following process,
\[
\zeta_s^\Delta = \exp\left\{\int_\nu^s \Delta_u^\prime \diff W_u - \frac{1}{2}\int_\nu^s\|\Delta_u\|^2\diff u\right\},\quad s\geq \nu
\]
is a martingale.
Hence, we can define a new probability measure such that
\[
\prob^\Delta_T(A) := \expect[\1_A\zeta_T^\Delta],\quad A\in\mathcal F_T.
\]
Furthermore, by the Girsanov theorem,
the following process,
\[
W^\Delta_t := \int_\nu^t\Delta_s\diff s - W_t,\quad t\in[\nu,T],
\]
is a $d$-dimensional Brownian motion under $\prob^\Delta_T$.
We denote by $\expect^\Delta_T$ an expectation operator under $\prob_T^\Delta$.
Since $\overline Z^{\nu,\eta,\iota,\alpha}$ and $Z^{\nu,\eta,\iota,\alpha}$ are in $\HH{\nu}{d}$,
it holds that
\[
\expect^\Delta_T\left[\int_\nu^T(D_{s}^{\nu,\eta,\iota,\alpha})^2\|\overline Z^{\nu,\eta,\iota,\alpha}_s - Z^{\nu,\eta,\iota,\alpha}_s\|^2\diff s\right] < \infty.
\]
This implies that the stochastic integral
\[
\int_\nu^u D_{s}^{\nu,\eta,\iota,\alpha}(\overline Z^{\nu,\eta,\iota,\alpha}_s - Z^{\nu,\eta,\iota,\alpha}_s)^\prime\diff W_s^\Delta,\quad u\in[\nu,T],
\]
is a martingale under $\prob_T^\Delta$.
Hence, taking conditional expectation of the inequality \cref{eq:ineq.optim.01} under the probability measure $\prob_T^\Delta$ given by $\filt_t$,
we obtain
\[
Y^{\nu,\eta,\iota}_t - Y^{\nu,\eta,\iota,\alpha}_t \geq 0,
\]
$\prob$-almost surely for all $t\in[\nu,T]$.

\textit{Step.2 Optimality of $Y$.}
We first prove the admissibility of $\alpha^*$.
Let $\overline Z^{\nu,\eta,\iota,\alpha^*}_s$ be a stochastic process defined as \cref{eq:Z.over}.
Then, by the definition $\alpha^*$, $K^{\tau_{k-1},X_{\tau_{k-1}}^*,i_{k-1}^*}_s = 0$ for all $k\geq 1$ and $s\in[\tau_{k-1}^*,\tau_{k}^*]$.
Furthermore, it holds that
\[
Y_{\tau^*_k}^{\tau_{k-1}^*,X^*_{\tau^*_{k-1}},i^*_{k-1}} = Y_{\tau^*_k}^{\tau_{k-1}^*,X^*_{\tau^*_{k-1}},i^*_k} - c_{i^*_{k-1},i^*_k}(\tau^*_k,X^*_{\tau^*_k}),
\]
for all $k\geq 1$.
Hence, the following equality holds.
\begin{align}\label{eq:temp.Y.eq}
D_t^{\nu,\eta,\iota,\alpha^*}Y^{\nu,\eta,\iota}_t
&= D_{t\vee\tau_n^*}^{\nu,\eta,\iota,\alpha^*}Y^{\tau_n^*,X_{\tau^*_n}^*,i_{n-1}^*}_{t\vee \tau_n^*} - \sum_{k=1}^n D_{\tau_k^*}^{\nu,\eta,\iota,\alpha^*}c_{i_{k-1}^*,i_k^*}(\tau_k^*,X_{\tau_k^*}^*)\1_{[\nu,\tau_k^*]}(t) \\
 &\quad + \int_{t}^{t\vee\tau_n^*}D_s^{\nu,\eta,\iota,\alpha^*}\Big(\psi(s,X_s^*,\alpha_s^*) - \varsigma(s,X_s^*,\alpha_s,\overline Z_s^{\nu,\eta,\iota,\alpha^*})\Big)\diff s \nonumber\\
&\qquad - \int_{t}^{t\vee\tau_n^*}D_s^{\nu,\eta,\iota,\alpha^*}(\overline Z^{\nu,\eta,\iota,\alpha^*}_s)^\prime \diff W_s,\nonumber
\end{align}
for all $n\geq 1$, where $a\vee b = \max\{a,b\}$.
Let $N^* = \inf\{k\;|\;\tau_k^*\geq T\}$ and $B=\{N^*=+\infty\}$.
Suppose that $\prob(B) > 0$.
Then, as $\mathcal I$ is a finite set, there exists a finite loop $i_0,i_1,\dots,i_m,i_0$, $i_0\in\mathcal I,i_0\neq i_1$ such that
\[
Y_{\tau^*_{k_{q+l}}}^{\nu,\eta,i_{l-1}} = Y_{\tau^*_{k_{q+l}}}^{\nu,\eta,i_{l}} - c_{i_{l-1},i_l}(\tau^*_{k_{q+l}},X^*_{\tau^*_{k_{q+l}}})\mbox{ on } B,
\]
for all $l=1,\dots,m+1,\;q\geq 0$ and $i_{m+1}=i_0$, where $(\tau_{k_q}^*)_{q\geq 1}$ is a subsequence of $(\tau_k^*)_{k\geq 0}$.
Let $\overline\tau = \lim_{q\rightarrow\infty}\tau_{k_q}^*$.
Then $\overline\tau < T$ on $B$ and
\[
Y_{\overline\tau}^{\nu,\eta,i_{l-1}} = Y_{\overline\tau}^{\nu,\eta,i_{l}} - c_{i_{l-1},i_l}(\overline\tau,X^*_{\overline\tau})\mbox{ on } B,
\]
for all $l=1,\dots,m+1$.
This implies that
\[
\sum_{l=1}^{m+1}c_{i_{l-1},i_l}(\overline\tau,X^*_{\overline\tau}) = 0\mbox{ on } B,
\]
which is contradiction to \cref{Assump:funcs}.3. Therefore, $\prob(B) = 0$ and $N^*$ is finite $\prob$-almost surely.
Hence, taking the limit of \cref{eq:temp.Y.eq}, we have
\begin{align}\label{eq:now.Y.eq}
D_t^{\nu,\eta,\iota,\alpha^*}Y^{\nu,\eta,\iota}_t
&= D_T^{\nu,\eta,\iota,\alpha^*}g(X_T^*,\alpha_T^*) - \sum_{t\leq \tau_k^*\leq T} D_{\tau_k^*}^{\nu,\eta,\iota,\alpha^*}c_{i_{k-1}^*,i_k^*}(\tau_k^*,X_{\tau_k^*}^*) \\
 &\quad  + \int_{t}^{T}D_s^{\nu,\eta,\iota,\alpha^*}\Big(\psi(s,X_s^*,\alpha_s^*) - \varsigma(s,X_s^*,\alpha_s,\overline Z_s^{\nu,\eta,\iota,\alpha^*})\Big)\diff s \nonumber \\
&\qquad - \int_{t}^{T}D_s^{\nu,\eta,\iota,\alpha^*}(\overline Z^{\nu,\eta,\iota,\alpha^*}_s)^\prime \diff W_s.\nonumber
\end{align}
Since $(Y^{\nu,\eta,\iota},\overline Z^{\nu,\eta,\iota,\alpha^*})\in\SSd{\nu}\times\HH{\nu}{d}$ and since \cref{Assump:X.SDE,Assump:rect,Assump:funcs,Assump:varsigma} are satisfied,
$\sum_{\nu\leq \tau_k^*\leq T} c_{i_{k-1}^*,i_k^*}(\tau_k^*,X_{\tau_k^*}^*)$ is quadratic integrable under $\prob$.
Hence, $\alpha^*$ is admissible.

We consider the solution to the BSDE \cref{eq:BSDE.minimal} at $(\nu,\eta,\iota,\alpha^*)$,
denoted by $(Y^{\nu,\eta,\iota,\alpha^*},Z^{\nu,\eta,\iota,\alpha^*})$.
Then, combining \cref{eq:now.Y.eq} and $(Y^{\nu,\eta,\iota,\alpha^*},Z^{\nu,\eta,\iota,\alpha^*})$,
we obtain that
\begin{multline*}
D_t^{\nu,\eta,\iota,\alpha^*}\Big(Y^{\nu,\eta,\iota}_t - Y^{\nu,\eta,\iota,\alpha^*}_t\Big)\\
 = \int_{t}^{T}D_s^{\nu,\eta,\iota,\alpha^*}\Delta_s^\prime(\overline Z^{\nu,\eta,\iota,\alpha^*}_s - Z^{\nu,\eta,\iota,\alpha^*}_s)\diff s  - \int_{t}^{T}D_s^{\nu,\eta,\iota,\alpha^*}(\overline Z^{\nu,\eta,\iota,\alpha^*}_s - Z^{\nu,\eta,\iota,\alpha^*}_s)^\prime \diff W_s,
\end{multline*}
where $(\Delta_s)_{0\leq s\leq T}$ is the stochastic process defined in Step.1.
As well as Step.1, we conclude that
\[
Y^{\nu,\eta,\iota}_t = Y^{\nu,\eta,\iota,\alpha^*}_t,
\]
$\prob$-almost surely for all $t\in[\nu,T]$.
\end{proof}

\section{Verification in the Infinite Horizon}\label{subsec:verification.infinite.horizon}

\begin{proof}[Proof of \cref{Prop:convergence.verification}]
\textit{Step.1 Monotonicity of $\widehat Y$.}
Fix an arbitrary $0\leq T\leq \widetilde T,\;\nu\in\mathcal T_0^T$ and $\eta\in L_\nu^{2q}(\real^d)$.
Let $(\widehat Y^{T,\nu,\eta,i,n},\widehat Z^{T,\nu,\eta,i,n},\widehat K^{T,\nu,\eta,i,n})_{n\geq 0}$ be the Picard's iterations of 
$(\widehat Y^{T,\nu,\eta,i},\widehat Z^{T,\nu,\eta,i},\widehat K^{T,\nu,\eta,i})$ constructed in \cref{Thm:existence}.\\
Also let $(\widehat Y^{\widetilde T,\nu,\eta,i,n},\widehat Z^{\widetilde T,\nu,\eta,i,n},\widehat K^{\widetilde T,\nu,\eta,i,n})_{n\geq 0}$ be the Picard's iterations of \\
$(\widehat Y^{\widetilde T,\nu,\eta,i},\widehat Z^{\widetilde T,\nu,\eta,i},\widehat K^{\widetilde T,\nu,\eta,i})$ constructed in \cref{Thm:existence}.
Then, by the non-negative reward condition, temporary terminal condition and Proposition 2.2 in \cite{el1997backward}, we have
\begin{align*}
e^{-\rho T}\widehat Y_T^{\widetilde T,\nu,\eta,i,0} 
&= \inf_{\theta\in\Theta[T,\widetilde T]}\expect\biggl[e^{-\rho \widetilde T}\zeta^{\theta,T}_{\widetilde T}g(X_{\widetilde T}^{\nu,\eta,i},i)  \\
&\qquad\qquad + \int_T^{\widetilde T}e^{-\rho t}\zeta_t^{\theta, T}\Big(\psi(X_t^{\nu,\eta,i},i) - \theta_t^\prime\phi(X_t^{\nu,\eta,i},i)\Big)\diff t\;\Big|\;\filt_T\biggr] \\
&\geq \inf_{\theta\in\Theta[T,\widetilde T]}\expect\left[e^{-\rho \widetilde T}\zeta^{\theta,T}_{\widetilde T}g(X_{\widetilde T}^{\nu,\eta,i},i)\;\Big|\;\filt_T\right] \geq e^{-\rho T}g(X_{T}^{\nu,\eta,i},i).
\end{align*}
Hence, $\widehat Y_T^{\widetilde T,\nu,\eta,i,0}\geq g(X_{T}^{\nu,\eta,i},i)$ for all $i\in\mathcal I$.
On the other hand,
$(\widehat Y_t^{\widetilde T,\nu,\eta,i,0},\widehat Z_t^{\widetilde T,\nu,\eta,i,0})$ for all $i\in\mathcal I$ is the solution to the following BSDE on $[\nu,T]$,
\begin{align*}
-\diff y_t &= \Big(\psi(X_t^{\nu,\eta,i},i) - \rho y_t - \varsigma(X_t^{\nu,\eta,i},i,z_t)\Big)\diff t - z_t^\prime \diff W_t,\\
y_T &= \widehat Y_T^{\widetilde T,\nu,\eta,i,0},\quad(y,z) \in \SSd{\nu}\times\HH{\nu}{d}.
\end{align*}
By the comparison theorem, $\widehat Y_t^{\widetilde T,\nu,\eta,i,0} \geq \widehat Y_t^{T,\nu,\eta,i,0}$ for all $t\in[\nu,T]$ and $i\in\mathcal I$.
Similarly, by the non-negative reward condition, temporary terminal condition and Proposition 7.1 in \cite{el1997reflected}, we have
\[
\widehat Y_T^{\widetilde T,\nu,\eta,i,n} \geq g(X_{T}^{\nu,\eta,i},i),
\]
for all $n\geq 1$.
Hence, recursively applying the comparison theorem, we obtain that $\widehat Y_t^{\widetilde T,\nu,\eta,i,n} \geq \widehat Y_t^{T,\nu,\eta,i,n}$ for all $t\in[\nu,T],\;i\in\mathcal I$ and $n\geq 1$.
Taking a limit, we also have $\widehat Y_t^{\widetilde T,\nu,\eta,i} \geq \widehat Y_t^{T,\nu,\eta,i}$ for all $t\in[\nu,T]$ and $i\in\mathcal I$.

\textit{Step.2 $n$-step dominated.}
Since $T\rightarrow \widehat Y_t^{T,\nu,\eta,i,n}$ is increasing by Step.1 and since $n\rightarrow \widehat Y_t^{T,\nu,\eta,i,n}$ is also increasing,
we can exchange the orders of taking the limits such that
\[
\lim_{T\rightarrow\infty}\widehat Y_t^{T,\nu,\eta,i} = \lim_{T\rightarrow\infty}\lim_{n\rightarrow\infty}\widehat Y_t^{T,\nu,\eta,i,n}
 = \lim_{n\rightarrow\infty}\lim_{T\rightarrow\infty}\widehat Y_t^{T,\nu,\eta,i,n}
 = \lim_{n\rightarrow\infty}\widehat Y_t^{\infty,\nu,\eta,i,n},
\]
where
\[
\widehat Y_t^{\infty,\nu,\eta,i,n} = \lim_{T\rightarrow\infty}\widehat Y_t^{T,\nu,\eta,i,n},\quad n\geq 1.
\]
By Proposition 2.2 in \cite{el1997backward} and the comparison theorem, it holds that
\begin{align}\label{eq:n.step.veri.inf}
e^{-\rho \nu}\widehat Y_{\nu}^{T,\nu,\eta,i,0}
 &= \inf_{\theta\in\Theta[\nu,T]}\expect\biggl[\zeta^{\theta,\nu}_T e^{-\rho T}g(X_T^{\nu,\eta,i},i)  \\
 &\qquad\qquad + \int_\nu^Te^{-\rho t}\zeta^{\theta,\nu}_t\Big(\psi(X_t^{\nu,\eta,i},i) - \theta_t^\prime \phi(X_t^{\nu,\eta,i},i)\Big)\diff t\;\Big|\;\filt_\nu\biggr],\nonumber
\end{align}
for all $T\geq \nu$.
Now, we choose an arbitrary $\theta\in\Theta[\nu,\infty)$.
Then, by the equality \cref{eq:n.step.veri.inf} and the temporary terminal condition, we have
\begin{align*}
e^{-\rho \nu}\widehat Y_{\nu}^{T,\nu,\eta,i,0}
 &\leq \expect\left[\int_\nu^Te^{-\rho t}\zeta^{\theta,\nu}_t\Big(\psi(X_t^{\nu,\eta,i},i) - \theta_t^\prime \phi(X_t^{\nu,\eta,i},i)\Big)\diff t\;\Big|\;\filt_\nu\right],
\end{align*}
for all $T\geq \nu$.
By the Lebesgue dominated convergence theorem, we have
\begin{align*}
e^{-\rho \nu}\widehat Y_{\nu}^{\infty,\nu,\eta,i,0}
 &\leq \expect\left[\int_\nu^\infty e^{-\rho t}\zeta^{\theta,\nu}_t\Big(\psi(X_t^{\nu,\eta,i},i) - \theta_t^\prime \phi(X_t^{\nu,\eta,i},i)\Big)\diff t\;\Big|\;\filt_\nu\right].
\end{align*}
Since $\theta$ is arbitrary, we obtain that
\begin{align}\label{eq:zero.step.veri}
e^{-\rho \nu}\widehat Y_{\nu}^{\infty,\nu,\eta,i,0}
 &\leq \inf_{\theta\in\Theta[\nu,\infty)}\expect\left[\int_\nu^\infty e^{-\rho t}\zeta^{\theta,\nu}_t\Big(\psi(X_t^{\nu,\eta,i},i) - \theta_t^\prime \phi(X_t^{\nu,\eta,i},i)\Big)\diff t\;\Big|\;\filt_\nu\right].
\end{align}
Now, we assume that for some $n\geq 1$,
\begin{align*}
e^{-\rho\widetilde\tau}\widehat Y_{\widetilde\tau}^{\infty,\widetilde\tau,\widetilde\eta,j,n-1} 
&\leq \sup_{\alpha\in\controlinf{\widetilde\tau}{j,n-1}}\inf_{\theta\in\Theta[\nu,\infty)}\expect\left[\int_{\widetilde\tau}^\infty e^{-\rho t}\zeta^{\theta,\widetilde\tau}_t
\Big( \psi( X_t^{\widetilde\tau,\widetilde\eta,j,\alpha},\alpha_t) - \theta_t^\prime \phi(X_t^{\widetilde\tau,\widetilde\eta,j,\alpha},\alpha_t)\Big)\diff t\right. \\
&\qquad\qquad\qquad \left.- \sum_{k=1}^{n-1}e^{-\rho\tau_k}\zeta^{\theta,\widetilde\tau}_{\tau_k}c_{i_{k-1},i_k}(X_{\tau_k}^{\widetilde\tau,\widetilde\eta,j,\alpha})\;\Big|\;\filt_{\widetilde\tau}\right],
\end{align*}
where $\widetilde\tau\in\mathcal T_\nu$, $\widetilde\eta\in L_{\widetilde\tau}^{2q}(\real^d)$, and $\controlinf{\widetilde\tau}{j,n-1}$ is a set of the admissible controls on $[\widetilde\tau,\infty)$ changing the regimes at most $n-1$ times.
On the other hand, by Proposition 7.1 in \cite{el1997reflected} and the uniqueness of $\widehat Y$,
it holds that 
\begin{align}\label{eq:temp.Snell}
e^{-\rho\nu} \widehat Y_{\nu}^{T,\nu,\eta,i,n}
&= \sup_{\widetilde\tau\in\mathcal T_\tau^T}\inf_{\theta\in\Theta[\nu,T]}\expect\biggl[
e^{-\rho T}
\zeta^{\theta,\nu}_T
g(X_T^{\nu,\eta,i},i)
\1_{\{\widetilde\tau = T\}} \\
&\qquad\qquad 
+ e^{-\rho \widetilde\tau}
\zeta^{\theta,\nu}_{\widetilde\tau}
\max_{ j\in\mathcal I\setminus \{ i \} }\left\{\widehat Y_{\widetilde\tau}^{T,\widetilde\tau,X^{\nu,\eta,i}_{\widetilde\tau},j,n-1}
 - c_{i,j}(X_{\widetilde\tau}^{\nu,\eta,i})\right\}\1_{ \{ \widetilde\tau < T \}} \nonumber\\
&\qquad\qquad\quad +\int_\nu^{\widetilde\tau} e^{-\rho t}\zeta^{\theta,\nu}_t
\Big( \psi( X_t^{\nu,\eta,i},i) - \theta_t^\prime \phi(X_t^{\nu,\eta,i},i)\Big)\diff t \;\Big|\;\filt_\nu\biggr].\nonumber
\end{align}
Let $\tau^*$ be an optimal stopping time of the maximization problem in the right hand side of \cref{eq:temp.Snell}.
Then, by Proposition 2.3 in \cite{el1997reflected}, we have
\[
\tau^* = \inf\left\{t\in[\nu,T]\;|\;\widehat Y_{t}^{T,t,X^{\nu,\eta,i}_{t},i,n} = \max_{ j\in\mathcal I\setminus \{ i \} }\left\{\widehat Y_{t}^{T,t,X^{\nu,\eta,i}_{t},j,n-1} - c_{i,j}(X_{t}^{\nu,\eta,i})\right\}\right\}.
\]
Hence,
\begin{align*}
&e^{-\rho T}\zeta^{\theta,\nu}_T g(X_{T}^{\nu,\eta,i},i)\1_{\{\tau^* = T\}} + e^{-\rho \tau^*}\zeta^{\theta,\nu}_{\tau^*}\max_{ j\in\mathcal I\setminus \{ i \} }\left\{\widehat Y_{\tau^*}^{T,\tau^*,X^{\nu,\eta,i}_{\tau^*},j,n-1} - c_{i,j}(X_{\tau^*}^{\nu,\eta,i})\right\}\1_{ \{ \tau^* < T \}} \\
&= e^{-\rho\tau^*}\zeta^{\theta,\nu}_{\tau^*}\Big(\widehat Y_{\tau^*}^{T,\tau^*,X^{\nu,\eta,i}_{\tau^*},j^*,n-1} - c_{i,j^*}(X_{\tau^*}^{\nu,\eta,i})\1_{\{\tau^*<T\}}\Big),
\end{align*}
where $j^*$ satisfies
\[
\widehat Y_{\tau^*}^{T,\tau^*,X^{\nu,\eta,i}_{\tau^*},i,n} = \widehat Y_{\tau^*}^{T,\tau^*,X^{\nu,\eta,i}_{\tau^*},j^*,n-1} - c_{i,j^*}(X_{\tau^*}^{\nu,\eta,i}),
\]
if $\tau^* < T$, and $j^*=i$ otherwise.
By the monotonicity of $\widehat Y$, we have
\[
\widehat Y_{\tau^*}^{T,\tau^*,X^{\nu,\eta,i}_{\tau^*},j^*,n-1} \leq \widehat Y_{\tau^*}^{\infty,\tau^*,X^{\nu,\eta,i}_{\tau^*},j^*,n-1}.
\]
Hence, we obtain
\begin{align*}
e^{-\rho\nu} \widehat Y_{\nu}^{T,\nu,\eta,i,n}
&\leq \inf_{\theta\in\Theta[0,T]}\expect\left[\int_\nu^{\tau^*} e^{-\rho t}\zeta^{\theta,\nu}_t\Big( \psi( X_t^{\nu,\eta,i},i) - \theta_t^\prime \phi(X_t^{\nu,\eta,i},i)\Big)\diff t\right.\\
&\qquad\qquad + \left. e^{-\rho\tau^*}\zeta^{\theta,\nu}_{\tau^*}\widehat Y_{\tau^*}^{\infty,\tau^*,X^{\nu,\eta,i}_{\tau^*},j^*,n-1} - e^{-\rho\tau^*}\zeta^{\theta,\nu}_{\tau^*}c_{i,j^*}(X_{\tau^*}^{\nu,\eta,i})\1_{\{\tau^*<T\}}\;\Big|\;\filt_\nu\right] \\
&\leq \inf_{\theta\in\Theta[\nu,\tau^*)}\expect\left[\int_\nu^{\tau^*} e^{-\rho t}\zeta^{\theta,\nu}_t\Big( \psi( X_t^{\nu,\eta,i},i) - \theta_t^\prime \phi(X_t^{\nu,\eta,i},i)\Big)\diff t\right.\\
&\qquad\qquad - e^{-\rho\tau^*}\zeta^{\theta,\nu}_{\tau^*}c_{i,j^*}(X_{\tau^*}^{\nu,\eta,i})\1_{\{\tau^*<T\}}\\
&\qquad\qquad + \zeta^{\theta,\nu}_{\tau^*}\sup_{\alpha\in\controlinf{\tau^*}{j^*,n-1}}\inf_{\theta\in\Theta[\tau^*,\infty)}\expect\left[\int_{\tau^*}^\infty e^{-\rho t}\zeta^{\theta,\tau^*}_t
\Big( \psi( X_t^{\tau^*,X_{\tau^*}^{\nu,\eta,i},j^*,\alpha},\alpha_t) \right. \\
&\qquad\qquad\qquad\qquad\qquad - \theta_t^\prime \phi(X_t^{\tau^*,X_{\tau^*}^{\nu,\eta,i},j^*,j,\alpha},\alpha_t)\Big)\diff t \\
&\qquad\qquad\qquad\qquad\qquad \left.\left.- \sum_{k=1}^{n-1}e^{-\rho\tau_k}\zeta^{\theta,\tau^*}_{\tau_k}c_{i_{k-1},i_k}(X_{\tau_k}^{\tau^*,X_{\tau^*}^{\nu,\eta,i},j^*,\alpha})\;\Big|\;\filt_{\tau^*}\right]\;\Big|\;\filt_\nu\right] \\
&\leq \sup_{\alpha\in\controlinf{\nu}{i,n}}\inf_{\theta\in\Theta[\nu,\infty)}\expect\left[\int_{\nu}^\infty e^{-\rho t}\zeta^{\theta,\nu}_t
\Big( \psi( X_t^{\nu,\eta,i,\alpha},\alpha_t) - \theta_t^\prime \phi(X_t^{\nu,\eta,i,\alpha},\alpha_t)\Big)\diff t\right. \\
&\quad \left.- \sum_{k=1}^{n}e^{-\rho\tau_k}\zeta^{\theta,\nu}_{\tau_k}c_{i_{k-1},i_k}(X_{\tau_k}^{\nu,\eta,i,\alpha})\;\Big|\;\filt_\nu\right],
\end{align*}
where we have used the uniqueness of the strong solution of $X$.
Taking a limit, we have
\begin{align}\label{eq:n.step.veri}
e^{-\rho\nu} \widehat Y_{\nu}^{\infty,\nu,\eta,i,n} 
&\leq \sup_{\alpha\in\controlinf{\nu}{i,n}}\inf_{\theta\in\Theta[\nu,\infty)}\expect\left[\int_{\nu}^\infty e^{-\rho t}\zeta^{\theta,\nu}_t
\Big( \psi( X_t^{\nu,\eta,i,\alpha},\alpha_t) - \theta_t^\prime \phi(X_t^{\nu,\eta,i,\alpha},\alpha_t)\Big)\diff t\right. \\
&\quad \left.- \sum_{k=1}^{n}e^{-\rho\tau_k}\zeta^{\theta,\nu}_{\tau_k}c_{i_{k-1},i_k}(X_{\tau_k}^{\nu,\eta,i,\alpha})\;\Big|\;\filt_\nu\right].\nonumber
\end{align}
By the inequalities \cref{eq:zero.step.veri,eq:n.step.veri},
we can prove that the inequality \cref{eq:n.step.veri} holds for all $n\geq 1$ using the induction method.
Since $\controlinf{t}{i,n}\subseteq\controlinf{t}{i}$ for all $n\geq 1$,
the inequality \cref{eq:n.step.veri} leads to
\begin{equation}\label{eq:dominated.infinity}
\lim_{T\rightarrow\infty}\widehat Y_t^{T,t,x,i} = \lim_{n\rightarrow\infty}\widehat Y_t^{\infty,t,x,i,n} \leq v^\infty(x,i),
\end{equation}
for all $(t,x,i)\in[0,\infty)\times\real^d\times\mathcal I$.
By the monotonicity of $\widehat Y$ and the inequality \cref{eq:dominated.infinity}, we have
\begin{equation}\label{eq:value.dominated.infinite}
\widehat Y_t^{T,t,x,i} \leq v^\infty(x,i),
\end{equation}
for all $(t,T,x,i)\in[0,\infty)^2\times\real^d\times\mathcal I$.

\textit{Step.3 Convergence.}
To prove the opposite inequality of \cref{eq:dominated.infinity}, we use the $\epsilon$-optimal argument such as Corollary 2.1 in \cite{bayraktar2010one}.
Fix any $(t,x,i)\in[0,\infty)\times\real^d\times\mathcal I$.
Let $J^T(t,x,i,\alpha)$ be an objective function in the finite horizon $[0,T]$.
Then, by the time-homogeneity, we have
\[
Y_t^{T,t,x,i} = Y_0^{T-t,0,x,i} \geq J^{T-t}(0,x,i,\alpha),
\]
for all $0\leq t\leq T,\;x\in\real^d,\;i\in\mathcal I$ and $\alpha\in\mathbb A_i[0,T-t]$.
Now, we fix an arbitrary $t\geq 0$ and $x\in\real^d$.
For any $\epsilon>0$, we choose a control $\alpha^\epsilon=(\tau_k^\epsilon,i_k^\epsilon)_{k\geq 0}\in\controlinf{0}{i}$ such that
\[
J(x,i,\alpha^\epsilon) \geq v^\infty(x,i) - \epsilon.
\]
For all $T\geq t$, define
\[
\alpha^{\epsilon,T-t}_s := \alpha_{s}^\epsilon,\quad s\in[0,T-t].
\]
Then, $\alpha^{\epsilon,T-t}\in\mathbb A_i[0,T-t]$ for all $T\geq t$.
For all $T\geq t$, let
\begin{multline*}
\theta^{T-t} := \arg\inf_{\theta\in\Theta[0,T-t]}\expect\biggl[\int_0^{T-t} e^{-\rho s}\zeta^{\theta,0}_s
\Big( \psi( X_s^{0,x,i,\alpha^{\epsilon}},\alpha_s^{\epsilon}) - \theta_s^\prime \phi(X_s^{0,x,i,\alpha^{\epsilon}},\alpha_s^{\epsilon})\Big)\diff s\\
- \sum_{k=1}^{\infty}e^{-\rho\tau_k^{\epsilon}}\zeta^{\theta,0}_{\tau_k^{\epsilon}}c_{i_{k-1}^{\epsilon},i_k^{\epsilon}}(X_{\tau_k^{\epsilon}}^{0,x,i,\alpha^{\epsilon}})\1_{\{\tau_k< T-t\}} + e^{-\rho (T-t)}\zeta_{T-t}^{\theta,0}g(X_{T-t}^{0,x,i,\alpha^{\epsilon}},\alpha_{T-t}^{\epsilon})\biggr].
\end{multline*}
Also let
\[
\theta^{\infty,T-t}_s := \left\{\begin{array}{ll}
 \theta_s^{T-t}, &\quad \mbox{if }s< T-t, \\
 0, &\quad\mbox{otherwise,}
\end{array}
\right.
\]
for all $T\geq t$.
It is easy to check $\theta^{\infty,T-t}\in\Theta[0,\infty)$.
Then, we have
\begin{align*}
&J(x,i,\alpha^\epsilon) \\
&\leq \expect\left[\int_0^\infty e^{-\rho s}\zeta^{\theta^{\infty,T-t},0}_s
\Big( \psi( X_s^{0,x,i,\alpha^{\epsilon}},\alpha_s^{\epsilon}) - (\theta_s^{\infty,T-t})^\prime \phi(X_s^{0,x,i,\alpha^{\epsilon}},\alpha_s^{\epsilon})\Big)\diff s\right. \nonumber\\
&\quad\quad \left.- \sum_{k=1}^{\infty}e^{-\rho\tau_k^{\epsilon}}\zeta^{\theta^{\infty,T-t},0}_{\tau_k^{\epsilon}}c_{i_{k-1}^{\epsilon},i_k^{\epsilon}}(X_{\tau_k^{\epsilon}}^{0,x,i,\alpha^{\epsilon}})\right]\\
&\quad = J^{T-t}(0,x,i,\alpha^{\epsilon,T-t}) + \expect\left[\zeta^{\theta^{\infty,T-t},0}_{T-t}\int_{T-t}^\infty e^{-\rho s}\psi( X_s^{0,x,i,\alpha^{\epsilon}},\alpha_s^{\epsilon})\diff s\right.\\
&\quad\quad \left.- \zeta^{\theta^{\infty,T-t},0}_{T-t}\sum_{k=1}^{\infty}e^{-\rho\tau_k^{\epsilon}}c_{i_{k-1}^{\epsilon},i_k^{\epsilon}}(X_{\tau_k^{\epsilon}}^{0,x,i,\alpha^{\epsilon}})\1_{\{\tau_k^\epsilon>T-t\}} - e^{-\rho (T-t)}\zeta^{\theta^{\infty,T-t},0}_{T-t} g(X_{T-t}^{0,x,i,\alpha^{\epsilon}},\alpha^{\epsilon}_{T-t})\right],
\end{align*}
for all $T\geq t$.
By the polynomial growth condition and the strong triangular condition, we have
\begin{align*}
&\expect\left[\int_{T-t}^\infty e^{-\rho s}\psi( X_s^{0,x,i,\alpha^{\epsilon}},\alpha_s^{\epsilon})\diff s - \sum_{k=1}^{\infty}e^{-\rho\tau_k^{\epsilon}}c_{i_{k-1}^{\epsilon},i_k^{\epsilon}}(X_{\tau_k^{\epsilon}}^{0,x,i,\alpha^{\epsilon}})\1_{\{\tau_k^\epsilon>T-t\}}\;\Big|\;\filt_{T-t}\right] \\
&\leq C_1(1 + \|X_{T-t}^{0,x,i,\alpha^{\epsilon}}\|^q) e^{-\rho (T-t)},
\end{align*}
for all $T\geq t$, where $C_1$ is a positive constant not depending on $T,t$ and $x$.
Thus, by the inequality \cref{eq:large.discount.01}, we have
\begin{align*}
&\expect\left[\zeta^{\theta^{\infty,T-t},0}_{T-t}\int_{T-t}^\infty e^{-\rho s}\psi( X_s^{0,x,i,\alpha^{\epsilon}},\alpha_s^{\epsilon})\diff s\right.\\
&\quad\quad \left.- \zeta^{\theta^{\infty,T-t},0}_{T-t}\sum_{k=1}^{\infty}e^{-\rho\tau_k^{\epsilon}}c_{i_{k-1}^{\epsilon},i_k^{\epsilon}}(X_{\tau_k^{\epsilon}}^{0,x,i,\alpha^{\epsilon}})\1_{\{\tau_k^\epsilon>T-t\}} - e^{-\rho (T-t)}\zeta^{\theta^{\infty,T-t},0}_{T-t} g(X_{T-t}^{0,x,i,\alpha^{\epsilon}},\alpha^{\epsilon}_{T-t})\right] \\
&\leq C_2\expect\left[\zeta^{\theta^{\infty,T-t},0}_{T-t}\left(1 + \|X_{T-t}^{0,x,i,\alpha^{\epsilon}}\|^q\right)e^{-\rho (T-t)}\right] \\
&\leq C_3(1 + \|x\|^q)e^{-c_\infty (T-t)},
\end{align*}
for all $T\geq t$, where $C_2,\;C_3$ and $c_\infty$ are positive constants not depending on $T,t$ and $x$.
This implies that for sufficiently large $\widetilde T$, it holds that
\begin{align}\label{eq:local.uniform.conv}
J(x,i,\alpha^\epsilon) \leq J^{T-t}(0,x,i,\alpha^{\epsilon,T-t}) + C_3(1 + \|x\|^q)e^{-c_\infty (T-t)} \leq J^{T-t}(0,x,i,\alpha^{\epsilon,T-t}) + \epsilon,
\end{align}
for all $T\geq \widetilde T$.
Hence, we have
\[
\liminf_{T\rightarrow\infty}Y_t^{T,t,x,i} \geq \liminf_{T\rightarrow\infty}J^{T-t}(0,x,i,\alpha^{\epsilon,T-t}) \geq J(x,i,\alpha^\epsilon) - \epsilon \geq v^\infty(x,i) - 2\epsilon.
\]
Since $\epsilon$ is arbitrarily chosen,
we obtain
\begin{equation}\label{eq:final.convergence}
\liminf_{T\rightarrow\infty}\widehat Y_t^{T,t,x,i} \geq v^\infty (x,i),
\end{equation}
for all $(t,x,i)\in[0,\infty)\times\real^d\times\mathcal I$.
Thus, we obtain the desired equality \cref{eq:convergence.verification}.
For all $i\in\mathcal I$, the convergence of \cref{eq:final.convergence} is locally uniform with respect to $t$ and $x$ by the inequalities \cref{eq:value.dominated.infinite,eq:local.uniform.conv}.
Furthermore, $Y_t^{T,t,x,i}$ is continuous in $t$ and $x$ for all $T\geq 0$ and $i\in\mathcal I$.
Therefore, $v^\infty(x,i)$ is continuous in $x$ for all $i\in\mathcal I$.
\end{proof}

\section{A Viscosity Solution in the Infinite Horizon}\label{subsec:viscosity.infinite}

\begin{proof}[Proof of \cref{Prop:uniform.viscosity}]
Let $v^T(t,x,i) = \widehat Y^{T,t,x,i}_t$ for $0\leq t\leq T,\;x\in\real^d$ and $i\in\mathcal I$.
By the definition, we have
\[
v^T(t,x,i) \geq \max_{j\in\mathcal I\setminus\{i\}}\{v^T(t,x,j) - c_{i,j}(x)\},
\]
for all $0\leq t\leq T,\;x\in\real^d$ and $i\in\mathcal I$.
Hence, taking a limit, we have
\[
v^\infty(x,i) \geq \max_{j\in\mathcal I\setminus\{i\}}\{v^\infty(x,j) - c_{i,j}(x)\},
\]
for all $(x,i)\in\real^d\times\mathcal I$.

Furthermore, we can show the followings.
\begin{lemma}\label{Lemma:finite.v.property}
For all $T>0,\;x\in\real^d$ and $i\in\mathcal I$, $v^T(\cdot,x,i)$ is non-increasing.
Furthermore,
there exists a positive constant $C$ such that
\begin{equation}\label{eq:T.increment.condition}
|v^T(t,x,i) - v^T(s,x,i)| \leq C(1 + \|x\|^q),
\end{equation}
for all $0\leq s\leq t\leq T,\;x\in\real^d$ and $i\in\mathcal I$.
\end{lemma}
We will show \cref{Lemma:finite.v.property} after the proof of \cref{Prop:uniform.viscosity}.

Let $C^2(\real^d)$ be a set of twice continuously differentiable functions from $\real^d$ onto $\real$.
Let $B(x) = \{y\in\real^d\;|\;\|y-x\|\leq 1\}$ be a closed unit ball on $\real^d$ centered on $x$.
Now, let us show the viscosity solution property of $v^\infty$.

\textit{Step.1 Viscosity subsolution.}
We arbitrarily choose $\varphi\in C^2(\real^d)$ and $\overline x\in\real^d$ such that $\max\{v^\infty(\cdot,i)-\varphi\} = v^\infty(\overline x,i) - \varphi(\overline x) = 0$.
Let
\[
\widehat \varphi(x) := \varphi(x) + \|x - \overline x\|^4.
\]
Let $(t_k,x_k)\in [0,k]\times B(\overline x)$ for all $k=1,2,3,\dots$ such that
\[
\max\{v^k(\cdot,\cdot,i) - \widehat \varphi\} = v^k(t_k,x_k,i) - \widehat \varphi(x_k).
\]
Since $v^k(\cdot,x,i)$ is non-increasing for all $x\in\real^d$ by \cref{Lemma:finite.v.property},
we have $t_k = 0$ for all $k$.
We choose a subsequence of $(x_k)_{k\geq 1}$ which converges to some $x_0\in\real^d$.
For convenience, we also denote this subsequence by $(x_k)_{k\geq 1}$.
Then, since $(x_k)_{k\geq 1}\subseteq B(\overline x)$, the Dini theorem leads to
\[
\lim_{k\rightarrow\infty}v^k(0,x_k,i) = v^\infty(x_0,i).
\]
Thus, we have
\begin{align*}
0&\leq v^\infty(\overline x,i) - \varphi(\overline x) - (v^\infty(x_0,i) - \varphi(x_0)) \\
&\leq \lim_{k\rightarrow\infty}\Big(v^k(0,\overline x,i) - \widehat\varphi(\overline x) - (v^k(0,x_k,i) - \widehat\varphi(x_k)) - \|x_k - \overline x\|^4\Big) \\
&\leq \lim_{k\rightarrow\infty}\Big( - \|x_k - \overline x\|^4\Big) =  - \|x_0 - \overline x\|^4.
\end{align*}
Hence, $x_0 = \overline x$.

Now, by \cref{Prop:visocity.property}, for all $k\geq 1$, we have
\begin{align*}
0&\geq - \frac{\p\widehat\varphi(x_k)}{\p t} - \ig^i\widehat\varphi(x_k) - \psi(x_k,i) + \rho v^k(0,x_k,i) 
+ \varsigma(x_k,i,\sigma^\prime(x_k,i)\nabla \widehat\varphi(x_k)) \\
&= - \ig^i\widehat\varphi(x_k) - \psi(x_k,i) + \rho v^k(0,x_k,i) + \varsigma(x_k,i,\sigma^\prime(x_k,i)\nabla \widehat\varphi(x_k)).
\end{align*}
Hence, by the Dini theorem, taking a limit of the above inequality, we have
\[
0 \geq - \ig^i\varphi(\overline x) - \psi(\overline x,i) + \rho v^\infty(\overline x,i) + \varsigma(\overline x,i,\sigma^\prime(\overline x,i)\nabla \varphi(\overline x)).
\]
This implies that $v^\infty$ is a viscosity subsolution of the PDE \cref{eq:PDE.infinite}.

\textit{Step.2 Viscosity supersolution.}
We arbitrarily choose $\varphi\in C^2(\real^d)$ and $\underline x\in\real^d$ such that $\min\{v^\infty(\cdot,i)-\varphi\} = v^\infty(\underline x,i) - \varphi(\underline x) = 0$.
For $m=1,2,3,\dots$, let
\[
\varphi_m(t,x) := \varphi(x) - \|x - \underline x\|^4 - \frac{t}{m}
\]
Now, fix an arbitrary $m$ temporarily.
Let $(t_k,x_k)\in [0,k]\times B(\underline x)$ for all $k=1,2,3,\dots$ such that
\[
\min\{v^k(\cdot,\cdot,i) - \varphi_m\} = v^k(t_k,x_k,i) - \varphi_m(t_k,x_k).
\]
For any $k\geq 1,\;t\in[0,k]$ and $x\in B(\underline x)$, by \cref{Lemma:finite.v.property}, we have
\begin{align*}
v^k(0,x,i) - \varphi_m(0,x) - (v^k(t,x,i) - \varphi_m(t,x)) 
&\leq -\frac{t}{m} + C\Big(1 + \|x\|^q\Big) \\
&\leq -\frac{t}{m} + C\Big(1 + \max_{y\in B(\underline x)}\|y\|^q\Big).
\end{align*}
We now suppose that
\begin{equation}\label{eq:t.compact}
t > m C\Big(1 + \max_{y\in B(\underline x)}\|y\|^q\Big).
\end{equation}
Then,
\begin{align*}
v^k(0,x,i) - \varphi_m(0,x) - (v^k(t,x,i) - \varphi_m(t,x)) 
&\leq -\frac{t}{m} + C\Big(1 + \max_{y\in B(\underline x)}\|y\|^q\Big) \\
&< 0,
\end{align*}
for all $t$ satisfying the inequality \cref{eq:t.compact}.
This implies that for sufficient large $\widetilde k$, all $t_k$ with $k\geq \widetilde k$ are in the following compact subset.
\[
\left[0,m C\Big(1 + \max_{y\in B(\underline x)}\|y\|^q\Big)\right].
\]
Now, we choose a subsequence of $(t_k,x_k)_{k\geq \widetilde k}$ converging some $(t_0,x_0)$.
We also write this subsequence as $(t_k,x_k)_{k\geq 1}$ for convenience.
Then, by the Dini theorem, we have
\[
\lim_{k\rightarrow\infty}v^k(t_k,x_k,i) = v^\infty(x_0,i).
\]
Hence, we have
\begin{align*}
0&\leq v^\infty(x_0,i) - \varphi(x_0) - (v^\infty(\underline x,i) - \varphi(\underline x)) \\
&\leq \lim_{k\rightarrow\infty}\Big(v^k(t_k,x_k,i) - \varphi_m(t_k,x_k) - (v^k(t_k,\underline x,i) - \varphi_m(t_k,\underline x)) - \|x_k - \underline x\|^4\Big) \\
&\leq \lim_{k\rightarrow\infty}\Big( - \|x_k - \underline x\|^4\Big) =  - \|x_0 - \underline x\|^4,
\end{align*}
so $x_0 = \underline x$.

Now, by \cref{Prop:visocity.property}, we have
\begin{align*}
0&\leq - \frac{\p\varphi_m(t_k,x_k)}{\p t} - \ig^i\varphi_m(t_k,x_k) - \psi(x_k,i) + \rho v^k(t_k,x_k,i) 
  + \varsigma(x_k,i,\sigma^\prime(x_k,i)\nabla \varphi_m(t_k,x_k)) \\
 &\quad = \frac{1}{m} - \ig^i\varphi_m(t_k,x_k) - \psi(x_k,i) + \rho v^k(t_k,x_k,i) 
 + \varsigma(x_k,i,\sigma^\prime(x_k,i)\nabla \varphi_m(t_k,x_k)),
\end{align*}
for all $k\geq 1$. Thus, by the Dini theorem, taking a limit with respect to $k$, we have
\begin{align*}
0&\leq\frac{1}{m} - \ig^i\varphi(\underline x) - \psi(\underline x,i) + \rho v^\infty(\underline x,i) + \varsigma(\underline x,i,\sigma^\prime(\underline x,i)\nabla \varphi(\underline x)).
\end{align*}
Since $m$ is arbitrarily chosen, tending $m$ to infinity, we have
\[
0\leq - \ig^i\varphi(\underline x) - \psi(\underline x,i) + \rho v^\infty(\underline x,i) + \varsigma(\underline x,i,\sigma^\prime(\underline x,i)\nabla \varphi(\underline x)).
\]
This implies that $v^\infty$ is a viscosity supersolution of the PDE \cref{eq:PDE.infinite}.
\end{proof}

\begin{proof}[Proof of \cref{Lemma:finite.v.property}]
For all $0\leq h\leq t\leq T,\;x\in\real^d$ and $i\in\mathcal I$, we have
\begin{align*}
v^T(t,x,i) = \widehat Y_t^{T,t,x,i} 
&= \widehat Y_{t-h}^{T-h,t-h,x,i} \quad\mbox{(time-homogeneous Markov property)} \\
&\leq \widehat Y_{t-h}^{T,t-h,x,i} \quad\mbox{(monotonicity of $\widehat Y$)}\\
&\quad = v^T(t-h,x,i).
\end{align*}
Hence, $v^T(\cdot,x,i)$ is non-increasing for all $T>0,\;x\in\real^d$ and $i\in\mathcal I$.

Now, we prove the inequality \cref{eq:T.increment.condition}.
Since $t\rightarrow v^T(t,x,i)$ is non-increasing for all $T,x$ and $i$,
it suffices to derive an upper boundary of $v^T(0,x,i) - v^T(T,x,i)$.
Then, by the polynomial growth conditions for $\phi,\psi,g,$ and $c$ and \cref{Prop:q.th.integrable,Prop:cost.upper.bounded},
it is easy to show that
\begin{align*}
&v^T(0,x,i) - v^T(T,x,i) = \widehat Y_{0}^{T,0,x,i} - g(x,i) \\
&\leq \sup_{\alpha\in\control{0}{i}}\inf_{\theta\in\Theta[0,T]}\expect\left[ \int_0^{T}\zeta_t^{\theta,0}e^{-\rho t}\Big(\psi(X_t^{0,x,i,\alpha},\alpha_t) - \theta_t^\prime \phi(X_t^{0,x,i,\alpha},\alpha_t)\Big)\diff t\right. \\
&\qquad\qquad \left. + \zeta_{T}^{\theta,0} e^{-\rho T}g(X_T^{0,x,i,\alpha},i) - \sum_{0\leq \tau_k\leq T}\zeta_{\tau_k}^{\theta,0}e^{-\rho \tau_k}c_{i_{k-1},i_k}(X_{\tau_k}^{0,x,i,\alpha})\right] -  g(x,i) \\
&\leq\sup_{\alpha\in\control{0}{i}}\expect\Big[ e^{-\rho T}g(X_T^{0,x,i,\alpha},i) -  g(x,i)\\
&\qquad\qquad + \int_0^{T}e^{-\rho t}\psi(X_t^{0,x,i,\alpha},\alpha_t)\diff t\left. - \sum_{0\leq \tau_k\leq T}e^{-\rho \tau_k}c_{i_{k-1},i_k}(X_{\tau_k}^{0,x,i,\alpha})\right] \\
&\leq C(1 + \|x\|^q),
\end{align*}
where $C$ is a positive constant not depending on $T$ and $x$.
\end{proof}

\bibliographystyle{siamplain}
\bibliography{master_citation}

\end{document}